%% file: Shepherd_Quad_Layout_Equivalence.tex
\newcommand{\documentversion}{12pt}
\documentclass[\documentversion]{elsarticle}

\makeatletter
\def\ps@pprintTitle{%
   \let\@oddhead\@empty
   \let\@evenhead\@empty
   \let\@oddfoot\@empty
   \let\@evenfoot\@oddfoot
}
\makeatother

\journal{}

\input{my_package_list.tex}
\newboolean{isELS}
\setboolean{isELS}{true}
\newboolean{isReview}
\setboolean{isReview}{false}
\input{new_commands.tex}

\begin{document}

\begin{frontmatter}



\title{The \QuadMeshImmersion{}: A Mathematically Equivalent Representation of a Surface Quadrilateral Layout{}}


\corref{cor1}
\author[address1]{Kendrick M. Shepherd}
\ead{kendrick@oden.utexas.edu}

\author[address2]{Ren\'e R. Hiemstra}
\ead{rene.hiemstra@ibnm.uni-hannover.de}

\author[address1]{Thomas J.R. Hughes}
\ead{hughes@oden.utexas.edu}

\cortext[cor1]{Corresponding author}
\address[address1]{Oden Institute for Computational Engineering and Sciences, University of Texas at Austin}

\address[address2]{Institut f\"ur Baumechanik und Numerisch Mechanik (IBNM), Leibniz Universit\"at Hannover}

\begin{abstract}
	\input{Abstract.tex}
\end{abstract}

%
%

\end{frontmatter}


\input{Introduction.tex}

\input{Equivalence.tex}

\input{Results.tex}

\input{Conclusion.tex}



\section*{Acknowledgements}
K. Shepherd was supported by United States Department of Defense Navy Contract N6833518C0014 and the National Science Foundation Graduate Research Fellowship under Grant No. DGE-1610403.  R. R. Hiemstra and T. J. R. Hughes were partially supported by the National Science Foundation Industry/University Cooperative Research Center (IUCRC) for Efficient Vehicles and Sustainable Transportation Systems (EV-STS), and the United States Army CCDC Ground Vehicle Systems Center (TARDEC/NSF Project \# 1650483 AMD 2). Any opinion, findings, and conclusions or recommendations expressed in this material are those of the authors and do not necessarily reflect the views of the National Science Foundation.

\bibliographystyle{plain}
\bibliography{Combined_Bibliography}

\newpage{}
\input{Mathematical_Background.tex}

\input{Appendix.tex}

\end{document}

%% file: my_package_list.tex
\usepackage{amsfonts}
\usepackage{amssymb}
\usepackage{amsmath}
\usepackage{amsthm}
\usepackage{amscd}
\usepackage{graphicx}
\usepackage{color}
\usepackage{enumerate}
\usepackage{float}
\usepackage{subcaption}
\usepackage[titletoc,title]{appendix}
\usepackage{ragged2e}
\usepackage{hyperref}
\usepackage[font={small}]{caption}
\hypersetup{
  colorlinks   = true, 
  urlcolor     = blue, 
  linkcolor    = black, 
  citecolor   = black 
}
\usepackage[ruled,noend]{algorithm2e}
\usepackage{todonotes}
\usepackage{tikz}
\usetikzlibrary{shapes,arrows}
\usetikzlibrary{positioning}
\usepackage{ifthen}
\usepackage{xstring}

\usepackage{algpseudocode}
\algdef{SE}
[OBJECT]
{Obj}
{EndObj}
[1]
{\textbf{object} \textsc{#1}}
{\textbf{end object}}

\usepackage{ulem}

\usepackage{rotating}

\usepackage{wrapfig}

%% file: new_commands.tex
\graphicspath{{./}{figures/}{Figures/}{../Actual/Figures/}{../Actual/}}

\ifthenelse{\boolean{isReview}}{
	
	\newcommand{\footerfontsize}{\large}
	}
	{
	
	\newcommand{\footerfontsize}{\footnotesize}
	}

\newcommand{\Rd}{\color{black}}
\newcommand{\Bd}{\color{black}}

\definecolor{forestgreen}{cmyk}{0.76,0,0.76,0.45}

\definecolor{lightgray}{cmyk}{0.2,0.2,0.2,0.2}

\definecolor{burntorange}{rgb}{0.74902,0.341176,0}

\definecolor{myBlue}{rgb} {0,    0.4470,    0.7410}

\ifthenelse{\boolean{isELS}}{
	\newtheorem{definition}{Definition}[section]
	\newtheorem{theorem}{Theorem}[section]
	\newtheorem{lemma}[theorem]{Lemma}
	\newtheorem{proposition}[theorem]{Proposition}
	\newtheorem{corollary}[theorem]{Corollary}
	
	\newtheorem{assumption}{Assumption}
	
	\theoremstyle{remark}
	\newtheorem{remark}{Remark}[section]
	}
	{}

\newcommand{\genus}{g}
\newcommand{\boundarycomp}{k}
\newcommand{\surf}{S}
\newcommand{\interior}[1]{\mathring{#1}}
\newcommand{\holonomy}{\mathrm{Hol}}

\newcommand{\singpts}{P}
\newcommand{\curvaturemap}{\theta}
\newcommand{\intsingpts}{\singpts_\circ}
\newcommand{\bdrysingpts}{\singpts_\partial}
\newcommand{\mathcrossfield}{C}
\newcommand{\quadcrossfield}{\mathcrossfield_Q}
\newcommand{\connection}[1]{\tau_{#1}}
\newcommand{\graph}{G}

\newcommand{\surfmsinggraph}{\surf - \left(\singpts\cup\graph\right)}
\newcommand{\toposurfmsinggraph}{\big(\surf - (\singpts \cup \graph)\big)_d}
\newcommand{\toposurfmsinggraphrtwo}{\big(\surf - (\singpts \cup \graph)\big)_{d_{\mathbb{R}^2}}}
\newcommand{\immersion}{\Psi}

\newcommand{\quotientmap}{\mathcal{Q}}
\newcommand{\coordline}[2]{\ell_{#2_{#1}}}
\newcommand{\barcoordline}[2]{\bar{\ell}_{#2_{#1}}}
\newcommand{\quotientcurve}[2]{\mathfrak{Q}_{#1,#2}}

\newcommand{\metric}[2]{\left\langle #1, #2 \right\rangle}
\newcommand{\emptymetric}{\metric{\cdot}{\cdot}}
\newcommand{\flatmetric}[2]{\metric{#1}{#2}_\mathbb{E}}
\newcommand{\emptyflatmetric}{\flatmetric{\cdot}{\cdot}}
\newcommand{\quadmeshmetric}[2]{\metric{#1}{#2}_Q}
\newcommand{\emptyquadmeshmetric}{\quadmeshmetric{\cdot}{\cdot}}
\newcommand{\distancemetric}{d}
\newcommand{\quaddistancemetric}{d_Q}
\newcommand{\euclideandistancemetric}{d_{\mathbb{R}^2}}
\newcommand{\completion}[1]{\overline{\overline{#1}}}
\newcommand{\transverse}{\mathrel{\text{\tpitchfork}}}
\makeatletter
\newcommand{\tpitchfork}{%
  \vbox{
    \baselineskip\z@skip
    \lineskip-.52ex
    \lineskiplimit\maxdimen
    \m@th
    \ialign{##\crcr\hidewidth\smash{$-$}\hidewidth\crcr$\pitchfork$\crcr}
  }%
}
\makeatother

\makeatletter
\newcounter{datastruct}

\makeatother
\newcommand{\pl}{piecewise linear}

\newcommand{\CuttingGraph}{Cutting Graph}

\newcommand{\cuttinggraph}{cutting graph}

\newcommand{\node}{node}

\newcommand{\Arc}{Arc}
\newcommand{\arc}{arc}

\newcommand{\patch}{patch}

\newcommand{\integralcurve}{integral curve}

\newcommand{\QuadLayout}{Quad Layout}

\newcommand{\quadlayout}{quad layout}

\newcommand{\textcrossfield}{cross field}
\newcommand{\textCrossField}{Cross Field}
\newcommand{\textframefield}{frame field}

\newcommand{\textFramefield}{Frame field}
\newcommand{\textquadmeshmetric}{\quadlayout{} metric}

\newcommand{\textQuadMeshMetric}{\QuadLayout{} Metric}

\newcommand{\affineflatmetric}{partial \textquadmeshmetric{}}

\newcommand{\quadmeshimmersion}{quad layout immersion}

\newcommand{\QuadMeshImmersion}{Quad Layout Immersion}

\newcommand{\textquadmeshimmersion}{\quadmeshimmersion}

\newcommand{\affineflatimmersion}{partial \quadmeshimmersion{}}
\newcommand{\AffineFlatImmersion}{Partial \QuadMeshImmersion{}}

\usepackage{adjustbox}
\usepackage{array}

\newcolumntype{R}[2]{%
    >{\adjustbox{angle=#1,lap=\width-(#2)}\bgroup}%
    l%
    <{\egroup}%
}



%% file: Abstract.tex
{\footerfontsize
Quadrilateral layouts on surfaces are valuable in texture mapping, and essential in generation of quadrilateral meshes and in fitting splines.
Previous work has characterized such layouts as a special metric on a surface or as a meromorphic quartic differential with finite trajectories. 
In this work, a surface quadrilateral layout is alternatively characterized as a special immersion of a cut representation of the surface into the Euclidean plane.
We call this a \textquadmeshimmersion{}.
This characterization, while posed in smooth topology, naturally generalizes to piecewise-linear representations.
As such, it mathematically describes and generalizes integer grid maps, which are common in computer graphics settings.
Finally, the utility of the representation is demonstrated by computationally extracting quadrilateral layouts on surfaces of interest.
}

%% file: Introduction.tex
\section{Introduction} \label{sec:introduction}

Recently, an extensive amount of research has been devoted to redefining parametric domains (reparameterizing) of surfaces. While the motivation for this work varies greatly (from texture mapping to structured finite element mesh extraction to rebuilding trimmed spline geometries), dozens of works have explored how to partition unstructured, less-optimal discretizations into ones with better structure. Of particular interest are partitions into quadrilateral domains. These segmentations are ideal for texture mapping \cite{Bommes:2013b,Ray:2010}, for use as meshes in solving PDEs \cite{Benzley:1995,Dazevedo:2000,Simons:2017}, and for rebuilding trimmed splines \cite{Hiemstra:2020,Urick:2019}.
{\Bd Particularly, the booming field of isogeometric analysis\textemdash{}which aims to streamline the engineering design through analysis process \cite{Hughes:2005}\textemdash{}cannot achieve its ultimate goal of a streamlined engineering workflow without a clear way to convert trimmed computer-aided design models into well-structured quadrilateral partitions of surfaces using splines without any trimming operations \cite{Hiemstra:2020,Marussig_Hughes:2017}.
}

While the ultimate goal of extracting a quadrilateral layout on a surface is clear, methods to extract such a layout differ greatly; \cite{Bommes:2013b} provides a good overview of current methods.
Recently, advances in the theory of layout generation have led to a Renaissance in well-structured quadrilateral mesh generation. 
For example, a quadrilateral layout on a surface has been shown to be equivalent to a special Riemannian metric with cone singularities \cite{Chen:2019} (called a ``quad mesh metric'') and also to a meromorphic quartic differential \cite{Lei:2020}.
Both works led to computational algorithms, driven by theory, to extract quadrilateral layouts.

\begin{figure}
	\centering
	\includegraphics[trim=0cm 0cm 0cm 0cm, clip, width=.95\textwidth]{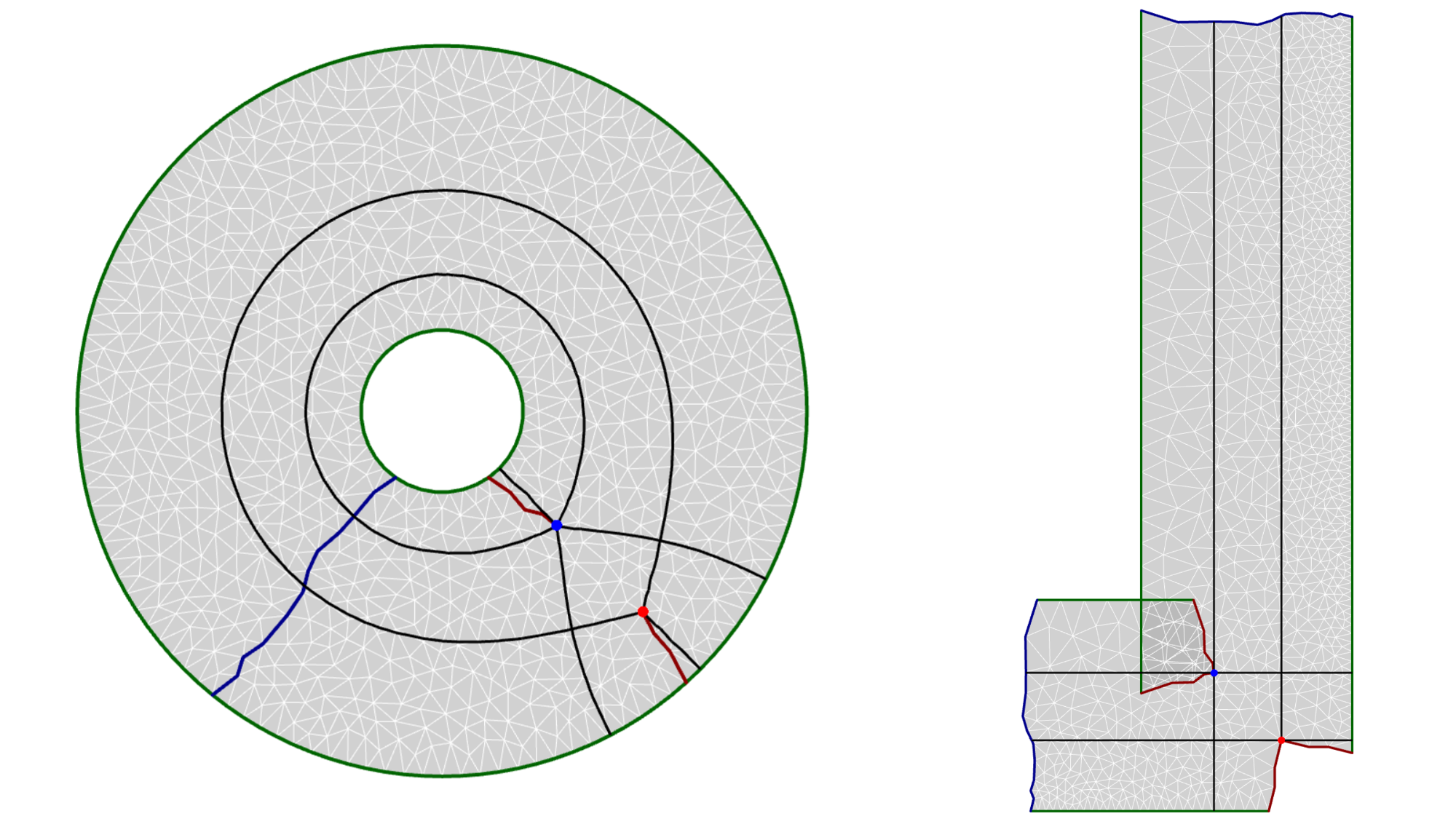}
	\caption{A quadrilateral layout on a surface (left) is equivalent to an immersion of the surface (right) after cutting it to disk topology (shown as curves in blue) and cutting to cone singular points (shown as curves in red). Here, the red point and the blue point are cone singularities with cone angles $\frac{3\pi}{2}$ and $\frac{5\pi}{2}$, respectively. The quadrilateral layout itself is formed by mapping lines of constant $u$ and $v$ coordinates in an immersed representation of the surface (on the right) to the original cut surface (on the left). The quadrilateral layout is completed after invoking the quotient space topology on the cut surface to return the topology to the uncut version. \label{fig:annulus}}
\end{figure}

In this work, a quad mesh metric (hereafter called a \textquadmeshmetric{} to emphasize that it is typically not \pl{}) on a surface is shown to be equivalent to a special kind of immersion on a cut representation of the surface. A representative immersion on an annulus is depicted in Figure \ref{fig:annulus}.

The isometric immersion proposed in this paper generalizes the concept of an ``integer grid map,'' which can currently be considered the state-of-the-art in high quality all-quadrilateral mesh generation and computation of quadrilateral layouts \cite{Bommes:2013,Bommes:2009,Campen:2015,Ebke:2013}. Compared with the proposed immersion, integer grid maps feature additional, extensive integer constraints. Computational techniques involve mixed-integer programming and are typically computationally intensive. Furthermore, the integer-valued constraints at singularities can cause undesirable distortion particularly when the target mesh element sizes are large. We show that integer grid maps are a subset of the potential class of quadrilateral layout-generating immersions.
This alternative characterization unites existing integer grid map theory with parameterization techniques applying topological path constraints between singularities (e.g. \cite{Campen:2012,Campen:2014,Hiemstra:2020,Tong:2006}). Furthermore, it generalizes the potential framework in which researchers may operate to extract \quadlayout{s} and, possibly, mitigates some of the disadvantages of integer grid maps.

The outline of the paper is as follows. First, a \quadlayout{} is shown to be equivalent to a special immersion mapping in  Section \ref{sec:equivalence}. Afterwards, Section \ref{sec:computation} will show some simple computational results based on the mathematical theory. This layout can then be directly utilized for operations such as spline fitting, texture mapping, or \pl{} quadrilateral mesh extraction.
Finally, Section \ref{sec:conclusions} will summarize results and discuss future areas of research.

%% file: Equivalence.tex
\section{An Equivalent Representation}\label{sec:equivalence}
Here we describe how a quadrilateral layout (a.k.a. a \quadlayout{}) is equivalent to a special immersion, which we call a \quadmeshimmersion{}. 
We assume that the reader has a graduate-level understanding of material from algebraic topology and differential geometry.
The supplementary material to this article gives a brief primer for those desiring a high-level overview.
Furthermore, we assume that all surfaces are orientable and compact, but possibly with boundary.

\subsection{The \textQuadMeshMetric{}}

\begin{figure}
	\centering
	\includegraphics[trim=0cm 0cm 0cm 0cm, clip, width=.95\textwidth]{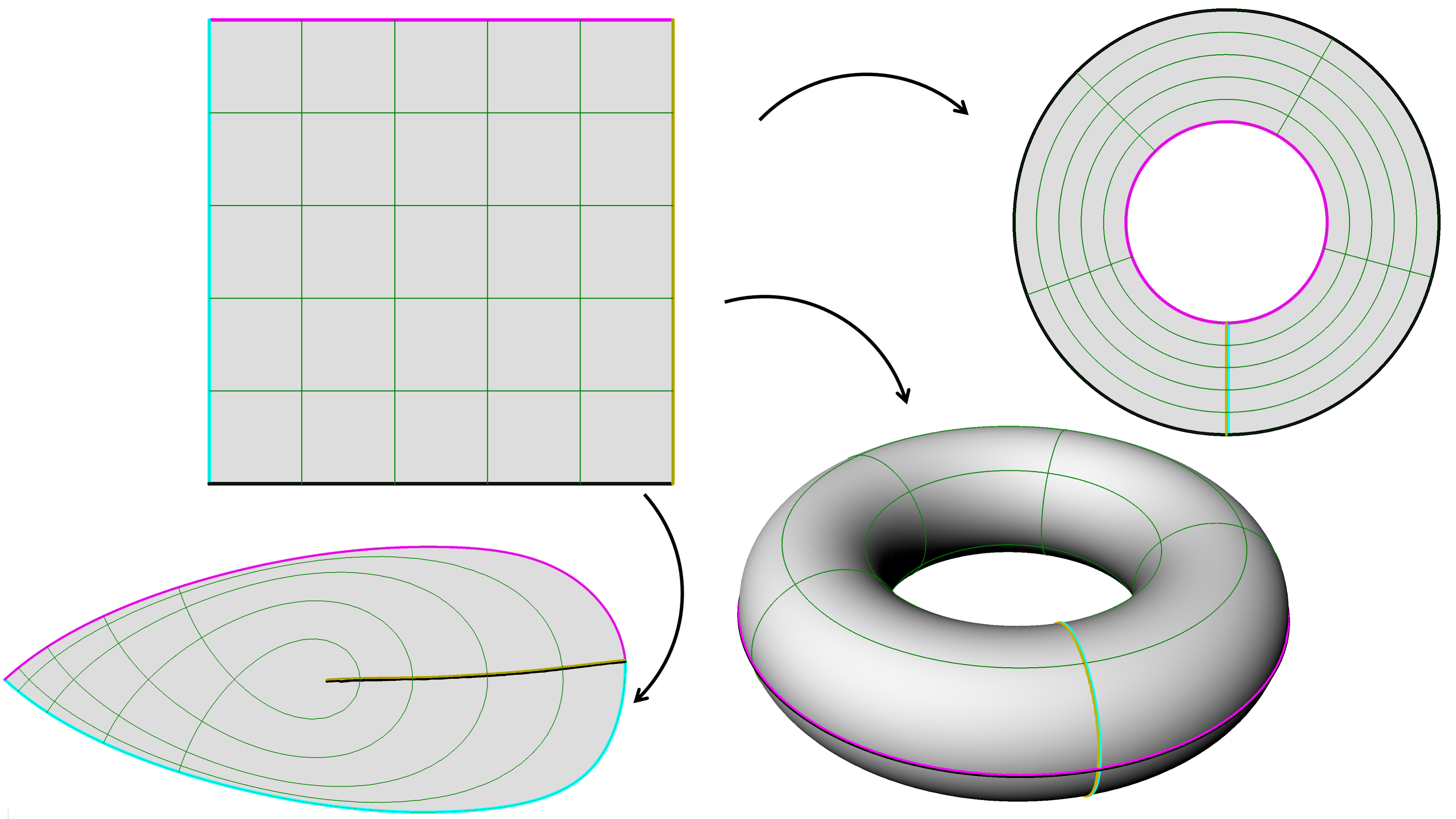}
	\caption{A parametric quadrilateral is mapped to various topological surfaces, with colored boundaries and contours indicating the precise mapping. Note that while each of these mappings is not injective on the boundary (or even locally injective at one \node{} for the teardrop shape), each defines  a single-patch \quadlayout{}. Additional examples of \quadlayout{s} are given in Section \ref{sec:computation}.\label{fig:single_patch_layouts}}
\end{figure}

For the sake of clarity, a \quadlayout{} is first defined.
\begin{definition}[{\Rd \QuadLayout}]\label{def:quad_layout}
	Let $\surf{}$ be a surface and take $\mathbb{I}^2 = [0,1]\times[0,1]$ to be the closed unit square.
	Accompanying $\mathbb{I}^2$ is a natural cell structure in which \node{s} are {\Bd points in $\mathbb{I}^2$ with integer-valued coordinates} in both $u$ and $v$, \arc{s} are line segments attached to adjacent \node{s} with constant $u$ or $v$ coordinate, and a \patch{} is attached to the square skeleton in the basic manner.
	Let $\phi:\mathbb{I}^2 \rightarrow \surf{}$ be a homeomorphism on $\mathring{\mathbb{I}}^2$ and a continuous local injection on each \arc{}. 
	Given a finite set of disjoint closed unit squares $\mathbb{I}^2_i$ with mappings $\phi_i:\mathbb{I}^2\rightarrow \surf{}$, the cellular structure induced on the image space is defined to be a \textbf{\quadlayout{}} if 
	\begin{itemize}
		\item Each point in the interior of a patch has a unique preimage defined by one (and only one) $\phi_i$.
		\item No cell of higher dimension is mapped to the same domain as the interior of a cell of lower dimension. (Here, the interior of a node is taken to be itself.)
	\end{itemize}
\end{definition}
As such, a T-mesh is not a \quadlayout{} because a \node{} of a \patch{} is mapped to the interior of an \arc{} of another.
By construction, all bilinear quadrilateral meshes are \quadlayout{s}.
Furthermore, representations such as an annulus with a single patch, a torus with a single patch, and a patch with adjacent arcs identified (each depicted in Figure \ref{fig:single_patch_layouts}) are \quadlayout{s} (with a single patch), despite having points on their boundaries mapped to the same locations.
Additional examples of \quadlayout{s} are presented in Figure \ref{fig:annulus} and in Section \ref{sec:computation}.

For the purposes of this paper, a quad mesh will be a \quadlayout{} in which each map $\phi$ is linear in $u$ and in $v$.
A well-structured \quadlayout{} will typically have far fewer \node{s}, \arc{s}, and \patch{es} than a quad mesh because
the objects of interest are, in general, curvilinearly mapped.

In \cite{Chen:2019}, a (curvilinear) \quadlayout{} on a smooth surface was shown to be equivalent to a \textquadmeshmetric{} (called the quad mesh metric in \cite{Chen:2019}, but renamed here to emphasize that a layout is generally curvilinear).
To describe this representation, the following definitions are necessary (see e.g. \cite{Charitos:2016,Cooper:2000,Viertel:2020}).

\begin{definition}[{\Rd Boundary Cone Singularity}]\label{def:boundary_cone}

	Let $0 < \curvaturemap \in \mathbb{R}$ be fixed.
	Take $F:(0,\epsilon) \times [0,\curvaturemap] \rightarrow \mathbb{R}^2$ be a smooth, bounded immersion with positive Jacobian  determinant  bounded from above and below, with $\lim_{t\rightarrow \epsilon} F(t,\phi) = (0,0) := v$ for all $\phi \in [0,\curvaturemap]$.
	Furthermore, take $F_\phi := F\Big|_{(0,\epsilon)\times\{\phi\}}$ with $\lim_{t\rightarrow \epsilon} F_\phi^\prime(t) = a(\phi) e^{i\phi}$ for $a:[0,\curvaturemap]\rightarrow \mathbb{R}$ positive and smooth.

	Define a metric on $(0,\epsilon) \times [0,\curvaturemap]$ by the pull-back of the Euclidean metric via the map $F$, given by $F^*(\emptymetric_{\mathbb{R}^2})$.
	Then the completion of $(0,\epsilon) \times [0,\curvaturemap]$ under this metric minus the subspace $\{0\}\times[0,\curvaturemap]$ is a \textbf{boundary cone} of angle $\curvaturemap,$ written $H\mathcal{C}(v,\curvaturemap)$, with $v$ as the \textbf{cone singularity}.
	The {\Bd point $v$ is singular in the following sense}: if $B_\delta(v)$ is a ball of radius $\delta$ about $v$, then
	\[
		\lim_{\delta \rightarrow 0} \int_{B_\delta(v) \cap \partial\big( H\mathcal{C}(v,\curvaturemap)\big)} \kappa_g  = \pi - \theta,
	\]
	{\Bd where $\kappa_g$ is geodesic curvature.} A \textbf{linear boundary cone} is defined when $F(t,0) = \big(r(t-\epsilon),0\big)$ for some $r > 0$ with $F(t,\phi) =  F(t,0) e^{i\phi}$.

\end{definition}

\begin{definition}[{\Rd Interior Cone Singularity}]\label{def:cone}
	A standard (surface) \textbf{cone} $\mathcal{C}(v,\theta)$ is a set with vertex $v$ and angle $\theta > 0$ described in coordinates as 
	\[
		\mathcal{C}(v,\curvaturemap) = \left\{(r,\phi): 0 \leq r \in \mathbb{R}, \phi \in \mathbb{R}/(\curvaturemap \mathbb{Z})\right\}
	\]
	with a metric locally of the form $ds^2 = dr^2 + r^2 d\phi^2$. 
	The vertex $v$ is called an  \textbf{interior cone singularity}.
	The singularity is represented in the following sense: for any neighborhood $N(v) \subset \mathcal{C}(v,\curvaturemap)$ containing $v$
	\[
		\int_{N(v)} \kappa = 2\pi - \curvaturemap,
	\]
	{\Bd where $\kappa$ is Gaussian curvature.}
\end{definition}

\begin{figure}
	\centering
	\includegraphics[trim=0cm 0cm 0cm 3cm, clip, width=.95\textwidth]{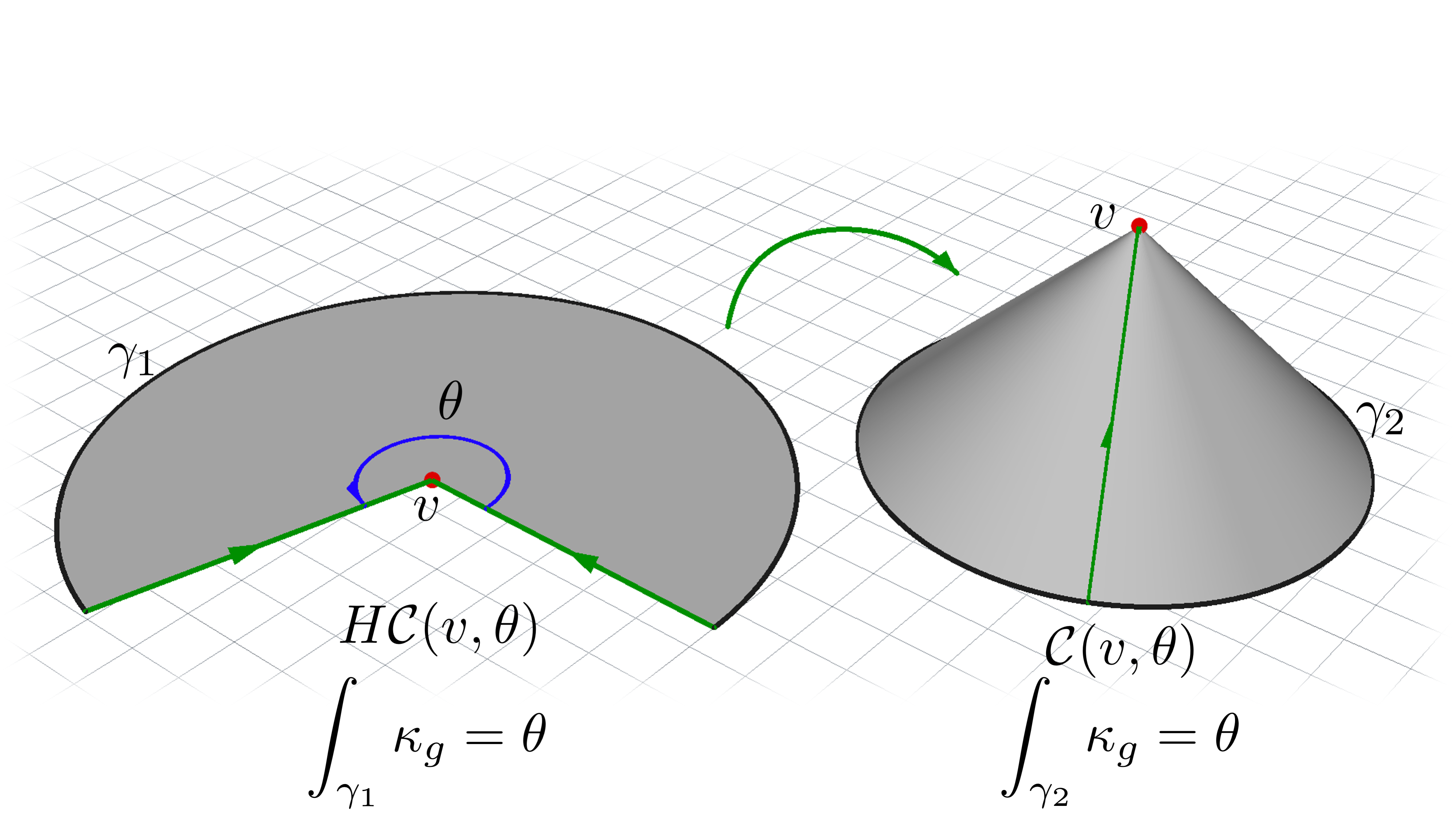}
	\caption{A cone $\mathcal{C}(v,\curvaturemap)$ can be thought of as the boundary cone $H\mathcal{C}(v,\curvaturemap)$ with the boundary marked in green with arrows glued. After this gluing, the cone is realized. To obey Gauss-Bonnet, the Gaussian curvature of any neighborhood of $v$ in the cone is taken to be $2\pi - \theta$ (recalling that the rest of the cone has zero Gaussian curvature).}\label{fig:cone_singularity}
\end{figure}

The integrals of Definitions \ref{def:boundary_cone} and \ref{def:cone} represent the contributions of cone singularities to a surface's geodesic and Gaussian curvatures, respectively.
Sometimes these are simply referred to as the discrete geodesic and Gaussian curvature of the cones.

For $0 < \curvaturemap{}  < 2\pi$, both the boundary cone and the cone can be visualized as depicted in Figure \ref{fig:cone_singularity}. 
The left is a linear boundary cone. 
Under a gluing operation of the edges marked with arrows, the boundary cone becomes a cone, which has no Gaussian curvature except at its singularity. 
For a cone singularity of angle greater than $2\pi$, the same general idea holds, but now the boundary cone {\Bd should} be thought of as an object in the complex plane with a branch cut. 
{\Bd Alternatively, these high-angle cones can be embedded in three dimensions by exploiting the vertical dimension.
Both representations are given in Figure \ref{fig:high_valence_cone}.}

\begin{figure}
\centering
	\begin{subfigure}{0.45\textwidth}
		\includegraphics[trim=0cm 0cm 0cm 0cm, clip, width=1\textwidth]{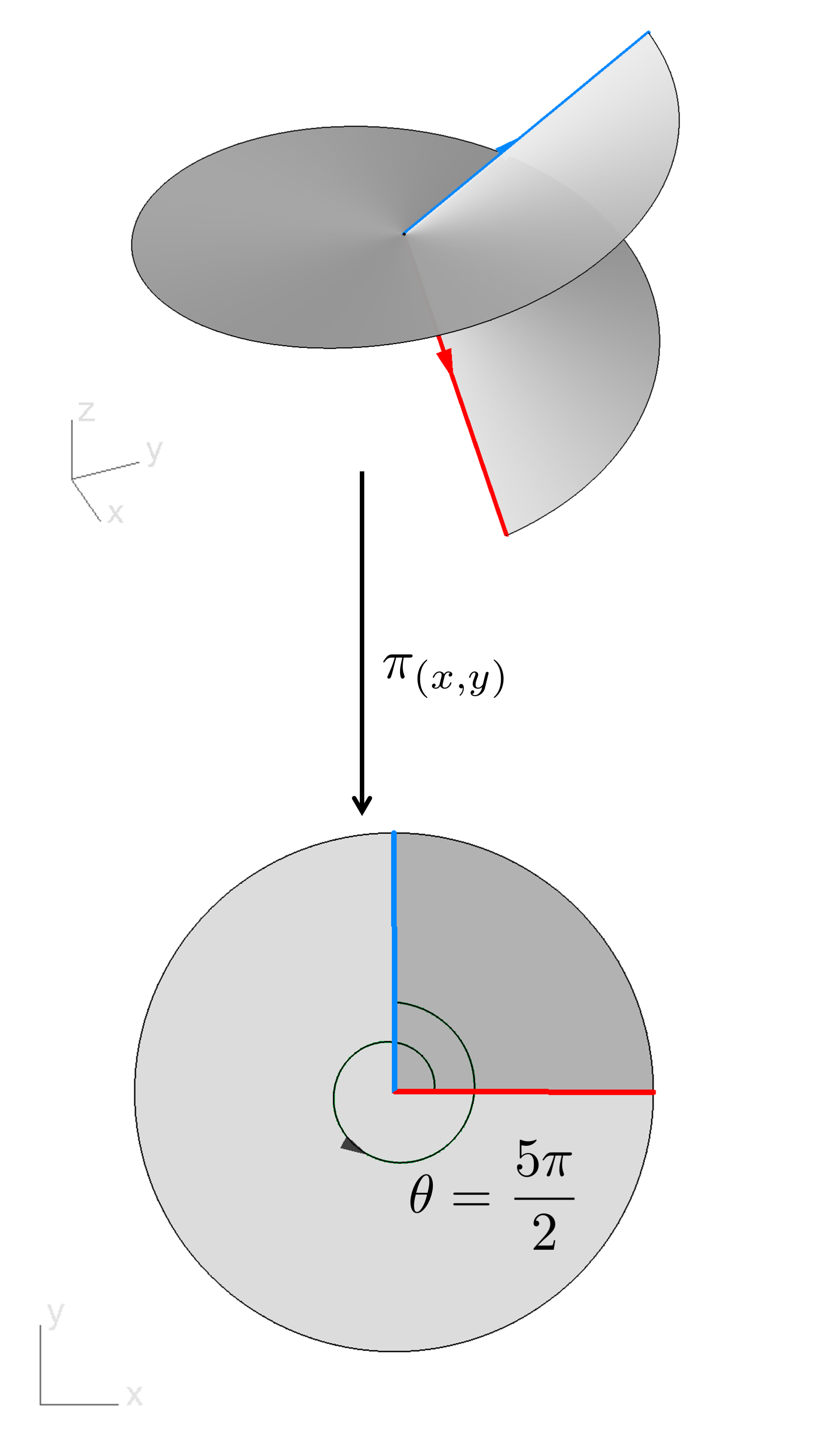}
		\caption{}
	\end{subfigure}
	\begin{subfigure}{0.45\textwidth}
		\includegraphics[trim=0cm 0cm 0cm 0cm, clip, width=1\textwidth]{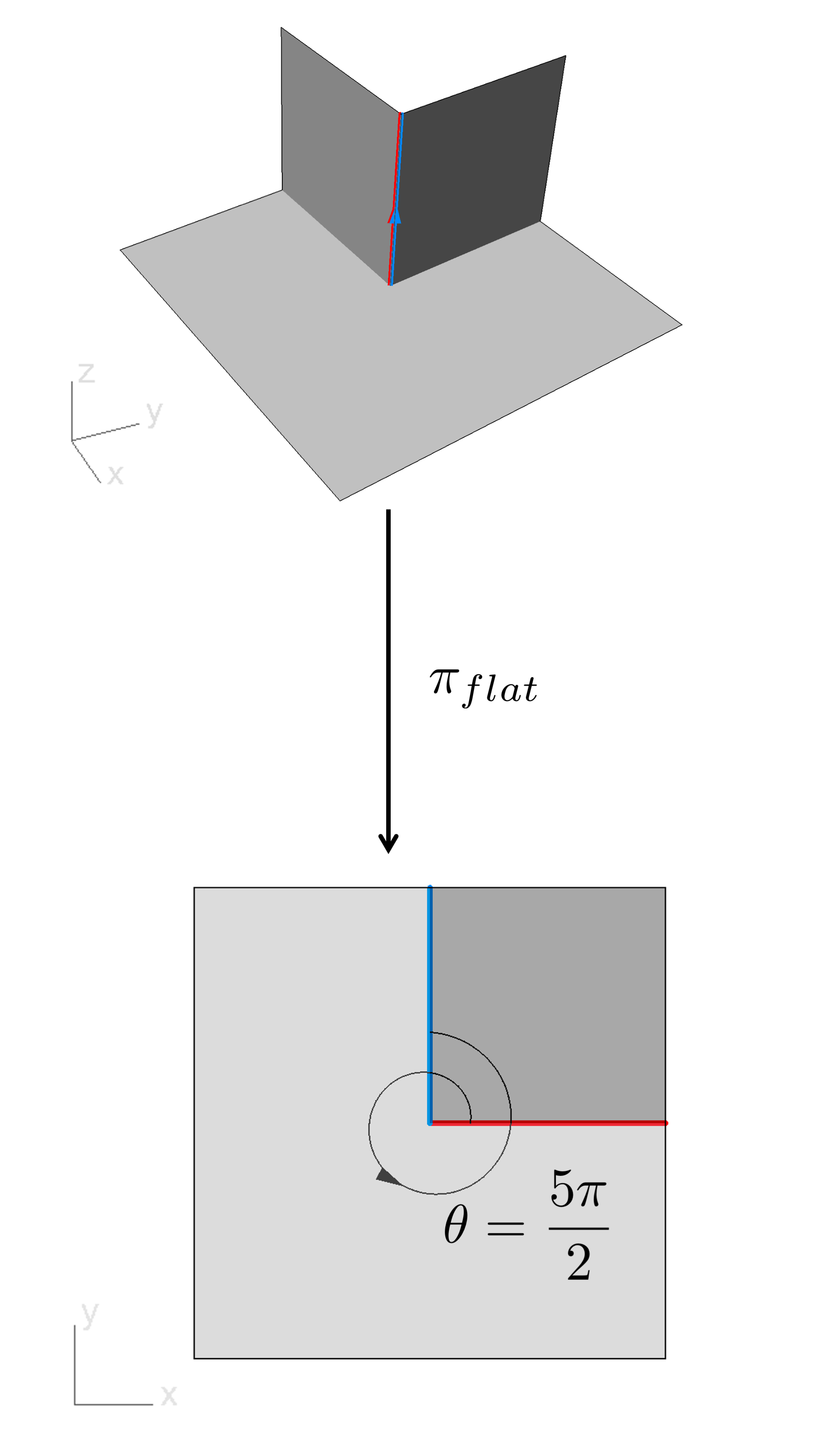}
		\caption{}
	\end{subfigure}	
	\caption{{\Bd 
	Two alternative representations of cones with angle $\frac{5\pi}{2}$ are presented. On the left, a heuristic representation of a cone is presented as a vertical spiral with red and blue edges glued which, when projected onto the $(x,y)$-plane, is locally injective everywhere except in the region of the cone. Alternatively, a sufficiently small neighborhood of this cone can be exactly embedded into three-space by leveraging the third dimension. When a cut is made and the vertical panels are folded onto the plane, the immersed structure is similar to that in the hueristic representation.
	}
	}\label{fig:high_valence_cone}
\end{figure}

Notice that the singularity definitions are both consistent for regular points on and off the boundary.
These are given by $\mathcal{C}(v,2\pi)$ and $H\mathcal{C}(v,\pi)$, respectively.
As such, for a surface $\surf$ with cone singularities, define $\curvaturemap: \surf \rightarrow \mathbb{R}$ by the $\curvaturemap$ of the (boundary) cone to which the point has an isometric neighborhood.

\begin{definition}[{\Rd Flat Metric with Cone Singularities}]\label{def:flat_w_cone}
A \textbf{flat metric with cone singularities} $\singpts$ on a surface $\surf$ (denoted $\emptyflatmetric$) is a Riemannian metric on $\surf{} - \singpts$ such that
\begin{itemize}
	\item Each point $p \in (\mathring{\surf} - \singpts)$ has a neighborhood isometric to an open disk in $\mathbb{R}^2$.
	\item Each point $q \in (\partial \surf - \singpts)$ has a neighborhood isometric to the regular boundary cone $H\mathcal{C}\big(v,\pi\big)$.
\end{itemize}
Furthermore, the Cauchy completion of the distance metric induced by the Riemannian metric is all of $\surf$.
Around the cone singularities, the following isometries hold:
\begin{itemize}
	\item Each point $p_i \in (\mathring{\surf} \cap \singpts) =:  \intsingpts$ has a neighborhood isometric to a neighborhood of the vertex $v$ of the standard cone $\mathcal{C}\big(v,\curvaturemap_{p_i}\big), \curvaturemap_{p_i}\neq 2 \pi.$
	\item Each point $q \in (\partial \surf \cap \singpts) =: \bdrysingpts$ has a neighborhood isometric to a neighborhood of vertex $v$ of a boundary cone $H\mathcal{C}\big(v,\curvaturemap_{q_i}\big), \curvaturemap_{q_i}\neq \pi.$
\end{itemize}
\end{definition}

In addition to these topology- and geometry-related concepts, the notion of a \textcrossfield{} will also be necessary.
Following the description of \cite{Ray:2008}, let a Riemannian metric $\emptymetric$ on a surface $\surf$ be given.
For $p \in \surf$, a unit tangent vector in $T_p\surf{}$ is a vector with norm $1$.
A $4$-symmetry direction is a set of $4$ unit tangent vectors at $p$ in which each vector differs by a rotation of $\frac{\pi}{2}$.
A \textcrossfield{} is a mapping associating with all but a finite number of  $p \in \surf$ a $4$-symmetry direction in a smooth manner.
 A \textcrossfield{} on a closed surface will have singularities which obey a Poincar\'e-Hopf-type theorem whose index is given heuristically by the number of rotations a unit directional field makes when traveling about a Darboux frame.
A \textframefield{} is a \textcrossfield{} represented under a different Riemannian metric.
A picture of a simple boundary-aligned \textframefield{} with singularities is given in Figure \ref{fig:trapezoidal_ff}.
More precise discussions are given in \cite{Ray:2008,Viertel:2020}, with definitions extending to boundary singularities given in \cite{Viertel:2020}.

\begin{figure}
	\centering
	\includegraphics[trim=0cm 0cm 0cm 0cm, clip, width=.95\textwidth]{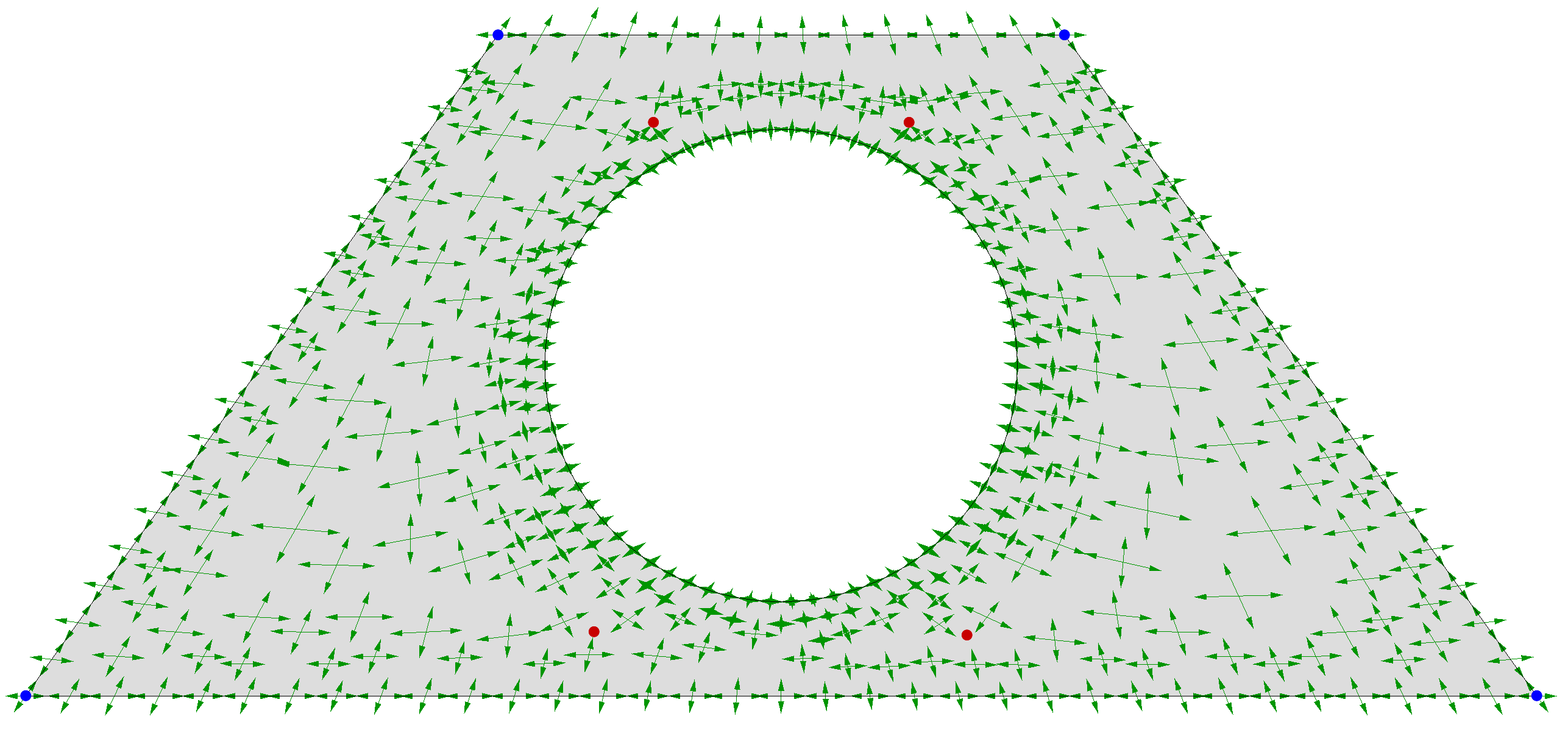}
	\caption{A boundary-aligned \textframefield{} with  four internal singularities (in red) on a surface with four sharp corners (in blue) is depicted.}\label{fig:trapezoidal_ff}
\end{figure}

With these definitions, we are prepared to describe a \textquadmeshmetric{}. Henceforth, we will assume the only connections used are the Levi-Cevita connections of the specified metric, written $\connection{}$. As such, the connection used for parallel translation and in defining the holonomy group is fixed.

\begin{definition}[{\Rd \textQuadMeshMetric{}} \cite{Chen:2019}]\label{def:qmetric}
A \textbf{\textquadmeshmetric{}} on a surface $\surf$ is a Riemannian metric $\emptyquadmeshmetric$ with cone singularities $\singpts$ with the following properties:
\begin{description}
\item[P1] 
	$\emptyquadmeshmetric$ is a flat metric with a finite number of cone singularities, $\singpts$. 
	The total curvature of the singularities obeys Gauss-Bonnet:
	\[
		\sum_{q \in \bdrysingpts}\big(\pi - \curvaturemap(q)\big) + \sum_{p \in \intsingpts}\big( 2\pi -  \curvaturemap(p)\big) = 2\pi \chi(\surf)
	\]
\item[P2] 
	The holonomy group of the surface is a subgroup of 
	$\mathcal{R} = \{\textrm{exp}(i \frac{k\pi}{2}), k \in \mathbb{Z}\}$, denoted
	\[
		\holonomy(\connection{}) \leq \mathcal{R}
	\]
\item[P3] 
	A boundary-aligned \textcrossfield{} defined on $\emptyquadmeshmetric{}$, $\quadcrossfield$,
	is obtained by parallel {\Bd transport} of a unit cross on a point 
	$p \in \surf-\singpts$ to all of $\surf - \singpts$.
\item[P4] 
	The integral curves of the \textcrossfield{} are geodesics of $\emptyquadmeshmetric{}$.
\item[P5]
	Integral curves of the \textcrossfield{} are periodic or of finite length.
\end{description}
\end{definition}
Note that Properties \textbf{P1} and \textbf{P3} imply that boundary singularities have a neighborhood isometric to a linear boundary cone, and that regular boundary points have a neighborhood isometric to an open half-disk, $HB_\epsilon$, with
	\[
		HB_\epsilon := \{ y = (y_1,y_2) \in \mathbb{R}^2: ||y||_{\mathbb{R}^2} < \epsilon, y_2 \geq 0\} ,
	\]
	for some $0 < \epsilon \in \mathbb{R}$ and $y = (y_1,y_2)$ under the standard Euclidean coordinates.
	A metric on $\surf$ obeying properties \textbf{P1}--\textbf{P4} of Definition \ref{def:qmetric} will be called a \textbf{\affineflatmetric{}}.

In general, a \textcrossfield{} can only be locally decomposed into four rotationally symmetric unit vector fields $\{X_j\}_{j=0}^3$ in which $RX_j = X_{(j+k)\; \mathrm{mod}\; 4}$ for $R = e^{i\frac{k\pi}{2}} \in \mathcal{R}$. 
 Parallel translation of a locally-defined component of the field about a loop may yield a possibly different component of the \textcrossfield{}, particularly if the loop bounds a topological disk with a cone singularity. 
 Here, \integralcurve{s} of a  \textcrossfield{} are defined  locally on a simply connected neighborhood and continued in a manner similar to analytic continuation {\Bd in} complex analysis.
Integrability of these fields, which is assumed by Properties \textbf{P4} and \textbf{P5}, is proved in the supplementary material.
This paper describes another equivalent representation which is a special type of immersion of the surface into $\mathbb{R}^2$.

Throughout this paper (as with other papers on the matter such as \cite{Bommes:2013,Bommes:2009,Campen:2019}), the following assumption holds on the curvature of the singularities:
\begin{assumption}\label{assume:sing_curvature}
	For $p \in \singpts$, $\curvaturemap{}_p > 0 $.

\end{assumption}
When $\curvaturemap{}_p = 0$, the definition of the cone singularity breaks down and it cannot be represented in Euclidean geometry; instead, it is a hyperbolic cusp if in the interior of the surface and a half-cusp if on the boundary \cite[pp.~54--55]{Cooper:2000}.
Such a point corresponds to a ``polar'' singularity, which of necessity would have an entire degenerate parametric edge.
These do not satisfy the above requirements of Definition \ref{def:quad_layout} and will not be further explored here.
However, it should be noted that a \quadlayout{} with polar singularities could be achieved by excising neighborhoods of the polar singularities, extracting a \quadlayout{} on the rest of the surface, and finally operating on the excised neighborhoods separately.

\subsection{Topological Preliminaries}

Under $\emptyquadmeshmetric{}$, the surface $\surf - \singpts$ has flat metric everywhere: such surfaces are often called \textbf{developable}, and can be ``flattened'' onto the plane (see \cite[pp.~66--72,91]{Struik:1988}).
The following discussion mimics theory related to the developing map of a manifold (see e.g. \cite[pp.~9]{Cooper:2000}).

We are interested in immersing the surface into the plane, and the ideal objects by which to do so are coordinate charts defined by the \textquadmeshmetric{'s} \textcrossfield{}, $\quadcrossfield{}.$
However, the holonomy about a singular point precludes this \textcrossfield{} from being separated into four globally-defined vector fields.
As such, the surface is first cut into a topologically simpler representation.

\begin{definition}[{\Rd \CuttingGraph}]\label{def:cutting_graph}
Let $\surf$ be a surface of genus $g$ with $\boundarycomp$ boundary components. 
Furthermore, let $\singpts \subset \surf$ be a finite set of discrete points in $\surf$. 
A \textbf{\cuttinggraph{}} $\graph(\surf,\singpts)$ is a piecewise smooth finite graph embedded in $\surf$ such that $(\singpts \cap \interior{\surf}) \subset \graph(\surf,\singpts)$ and $\surf - \graph(\surf,\singpts)$ is a set of simply-connected surfaces.

A \cuttinggraph{} is called \textbf{simple}, if, in addition, $\surf-\graph(\surf,\singpts)$ is a single simply-connected component  in which $\intsingpts \subset \partial \graph(\surf,\singpts), \bdrysingpts \ni q_i \not \in \graph(\surf,\singpts)$, and $\partial \graph(\surf,\singpts) \cap \partial \surf$ is discrete. 
\end{definition}

A simple \cuttinggraph{} is particularly convenient because it only intersects the boundary of $\surf$ transversely and discretely; it guarantees that each member of $\intsingpts$ is only cut to, but not through; and it ensures that each member of $\bdrysingpts$ is uncut.

When the surface and set $P$ are clear from context, we will use the notation $\graph := \graph(\surf,\singpts)$. 
The following assures the existence of cutting graphs on surfaces with cone singularities.

\begin{lemma}\label{lem:existence}
For a surface $\surf$ of genus $\genus$ with $\boundarycomp$ boundary components and a set of discrete points $\singpts \subset \surf$, there exists a simple cutting graph $\graph(\surf,\singpts)$.
\end{lemma}

{\Bd
The proof of Lemma \ref{lem:existence} is a slight extension of a fundamental result from algebraic topology.
The unfamiliar reader is referred to the Appendix (in the supplementary material) for additional details}.

\begin{remark}\label{rem:singularities}
Note that for a \textquadmeshmetric{}, the only surfaces in which $\graph = \emptyset$ are genus $0$ surfaces with $1$ boundary component and cone singularities on boundaries. For genus $\genus>0$ surfaces and surfaces with $\boundarycomp >1$, {\Bd $\graph \neq \emptyset$} is because $\pi_1(\surf) \neq 0$. For $\genus=0$ surfaces with $\boundarycomp = 0, \#\singpts > 2$ by Assumption \ref{assume:sing_curvature} and the Gauss-Bonnet Theorem. 
\end{remark}

{\Rd Generally a \cuttinggraph{} will not be a one-dimensional manifold with boundary because of the presence of splitting junctions. 
Nonetheless, we define the boundary of the \cuttinggraph{} to be the set of all points $p \in \graph$ possessing an open neighborhood in $\graph$ homeomorphic to $[0,\epsilon)$ in which $p \mapsto 0$. The boundary of the graph will be denoted as $\partial \graph$.}

Let each \cuttinggraph{} $\graph(\surf, \singpts)$ be given the following cellular structure.
First, take \node{s} of the graph to be the set $\mathcal{N}$ corresponding to {\Rd points in the \cuttinggraph{'s} boundary, splitting junctions, singularities of $\surf$ contained in $\graph$, and locations at which the graph is not smoothly embedded in $\surf$}.
Next, let the \arc{s} of $\graph$ be the set of (open) $1$-cells bijectively connecting zero-cells, written as $\mathcal{A}$.
{\Bd A simple cutting graph on a double torus with two cones of angle $4\pi$, together with the resulting cellular structure on $\graph$, is shown in Figure \ref{fig:cutting_graph}.
}

\begin{figure}
	\includegraphics[trim=0cm 0cm 0cm 0cm, clip, width=1\textwidth]{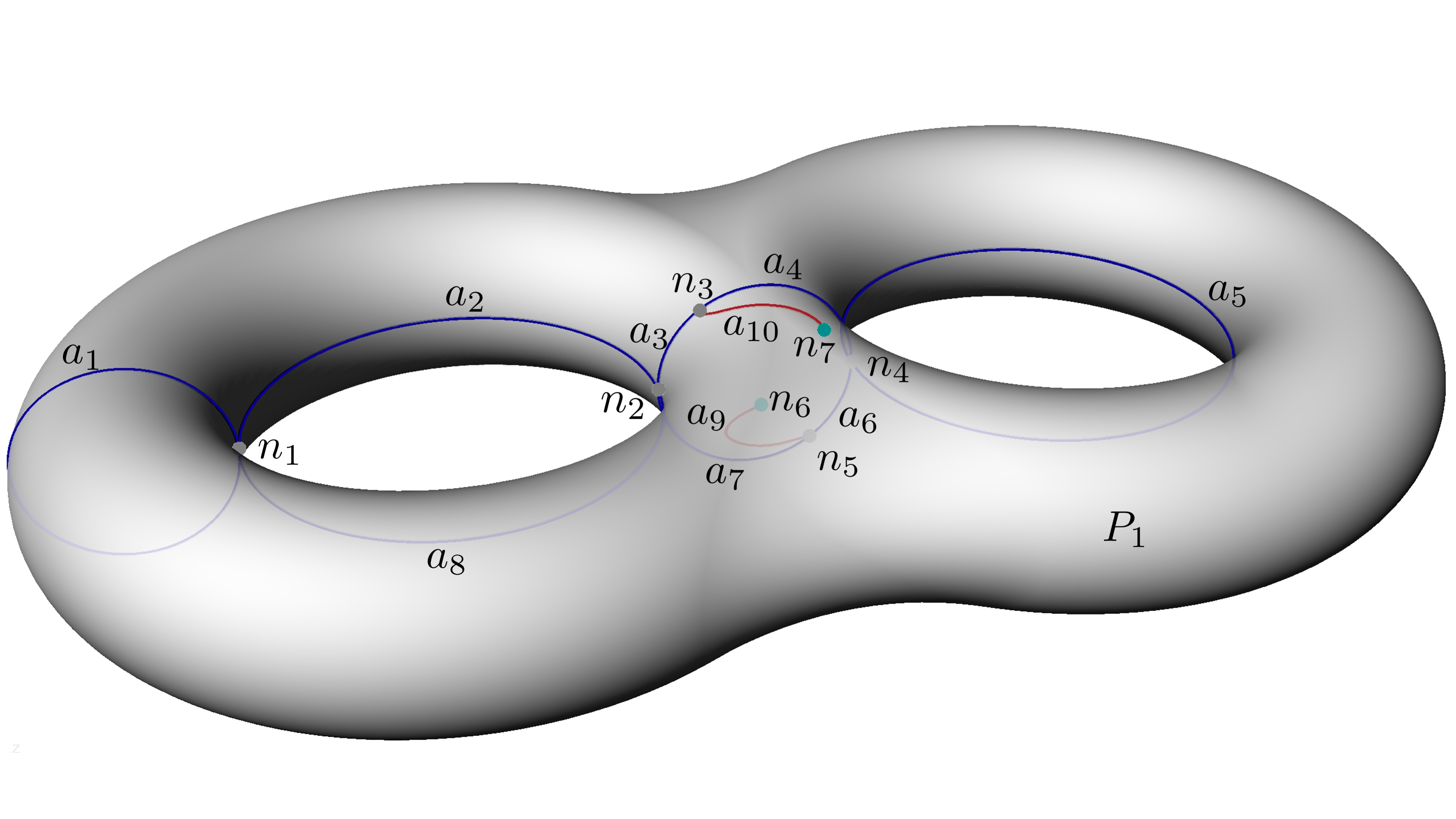}
	\caption{
	{\Bd A simple \cuttinggraph{} on a double torus is displayed.
	Curves in blue represent cuts to take the geometry to a topological disk, while curves in red cut to cone singularities of angle $4\pi$ (depicted as teal points).
	Here, \node{s} $n_i$ of the \cuttinggraph{} are locations in which the graph either terminates or is non-manifold (splits).
	\Arc{s} $a_j$ are maximal portions of the graph bounded by two \node{s} and with no \node{s} in their interior.
	Because this \cuttinggraph{} is simple, it cuts the surface into a single topological disk (called a \patch{}), $P_1$.
	}
	}\label{fig:cutting_graph}
\end{figure}

{\Bd Recall that any Riemannian metric, $\emptymetric$, on a manifold $M$ induces distance metric in the following manner.
If $\gamma:\mathbb{I}\rightarrow M$ is a curve in $M$, the length of the curve under the metric is defined as
\begin{equation}\label{eq:length_curve_manifold}
	L(\gamma) = \int_0^1 \Big(\metric{\gamma^\prime(t)}{\gamma^\prime(t)} \Big)^{\frac{1}{2}} dt.
\end{equation}
Define the induced distance metric, $\distancemetric{}:M\times M \rightarrow \mathbb{R}$ as 
\begin{equation}\label{eq:distance_metric_manifold}
	\distancemetric(p,q) = \inf_{\gamma \in \{\gamma_\iota\}_\iota} L(\gamma).
\end{equation}
where for some $p, q \in M$, $\{\gamma_\iota\}_\iota$ is the set of all curves in which $\gamma_\iota(0) = p, \gamma_\iota(1) = q.$
}

We are interested in the Cauchy completion of metric spaces on the surface.
First, it is shown that for any Riemannian metric on all of $\surf$, the {\Bd topology induced by its distance metric is equivalent to that induced by the \textquadmeshmetric{}}.

\begin{lemma}\label{lem:completion_same_nograph}
Let $\emptymetric$ be any Riemannian metric tensor on a compact surface $\surf$. 
Denote by $\distancemetric$ its induced {\Bd distance} metric on the surface $\surf$.
Let $\emptyquadmeshmetric$ be a \textquadmeshmetric{} on $\surf$, with $ \quaddistancemetric$ its induced {\Bd distance} metric on $\surf-\singpts$.
Then the {\Bd topologies} induced by $\distancemetric$ and $\quaddistancemetric$ on the domain $\surf - \singpts$  {\Bd are}  equivalent, and their Cauchy completions induce the same topology on $\surf$.
\end{lemma}
\begin{proof}
First, recall that all Riemannian metrics are Lipschitz equivalent on a compact surface in the following sense:
for metrics $\emptymetric_1,\emptymetric_2$, there exist $0 < \alpha,\beta < \infty$ constants such that for any curve $\gamma$,
\[
	\alpha \metric{\gamma^\prime}{\gamma^\prime}_1 \leq \metric{\gamma^\prime}{\gamma^\prime}_2 \leq \beta \metric{\gamma^\prime}{\gamma^\prime}_1.
\]
Then by definition (Equations \ref{eq:length_curve_manifold} and \ref{eq:distance_metric_manifold}),
\[
	\alpha^{\frac{1}{2}} \distancemetric_1(p,q) \leq \distancemetric_2(p,q) \leq \beta^{\frac{1}{2}} \distancemetric_1(p,q).
\]
Thus the metrics are strongly equivalent and the topologies induced by $\distancemetric_1$ and $\distancemetric_2$ are equivalent.
Furthermore, if a set of discrete points $\singpts$ are removed from $\surf$, the completion of $\surf - \singpts$ in both metrics is the same, and can simply be written as the topology on $\surf.$

Now, for $\surf - \singpts$, $\emptyquadmeshmetric$ is a Riemannian metric, which induces a well-defined metric on the surface, $\quaddistancemetric:(\surf-\singpts)\times(\surf-\singpts)\rightarrow \mathbb{R}$.
By construction, {\Bd the completion of the metric space induced by the distance metric $\quaddistancemetric$ is the entire domain} $\surf$.
 Furthermore, the topology on $\surf$ induced by $\quaddistancemetric$ coincides with the topology on $\surf$ induced by any Riemannian metric: a ball in one topology contains a ball in the other.
But all metrics on a compact domain  inducing the same topology are equivalent in the sense that the identity mapping from $\surf$ in one metric space to $\surf$ under a different topology is uniformly continuous.
Specifically, Cauchy sequences converging to any point in one will converge to the same point in the other, and thus the completions are identical.
\end{proof}

By a similar argument on $\surfmsinggraph$ the following holds.
\begin{corollary}\label{corol:completion_same_graph}
Let $\emptymetric$ be a Riemannian metric on $\surf$, with induced Riemannian metric on $\surfmsinggraph$ yielding a distance metric $\distancemetric:\big(\surfmsinggraph\big)\times\big(\surfmsinggraph\big) \rightarrow \mathbb{R}$.
Similarly, take the induced Riemannian metric $\emptyquadmeshmetric$ on $\surfmsinggraph$ with induced distance metric $\quaddistancemetric:\big(\surfmsinggraph\big)\times\big(\surfmsinggraph\big) \rightarrow \mathbb{R}$.
Then the topologies induced by $\distancemetric$ and $\quaddistancemetric$ are both the same, and have identical completion.
\end{corollary}
\begin{proof}
The notions used in \ref{lem:completion_same_nograph} are entirely local, and thus also apply to subspaces.
\end{proof}
Thus we can canonically define the topology of the surface $\surfmsinggraph$ for a Riemannian metric or a \textquadmeshmetric, which will simply be denoted as $\surfmsinggraph$. 
Furthermore, the completion is canonically defined and denoted as $\completion{\surfmsinggraph}$, where the double line is used {\Bd to emphasize} that this is not a closure operation.

To better understand the completion, let $p \in \surf$ with $U(p)$ a simply connected closed neighborhood in $S$ with the following properties.
\begin{enumerate}
	\item If $p \not \in  \singpts \cup \graph, U(p) \cap (\singpts \cup \graph) = \emptyset$.
	\item If $p \in \singpts \cup \graph$, then $U(p)$ contains no \node{s} of $\graph$ and no members of $\singpts$ other than (possibly) $p$ itself.
	\item $\graph \cap U(p)$ has at most one connected component.
\end{enumerate}
Then $U(p)- \graph$ divides $U(p)$ into $\ell$ connected components, written $\{B_i(p)\}_{i=1}^\ell$.
If $p \in \graph,$ in the completion $\completion{\toposurfmsinggraph}$, $p$ will be represented by $\ell$ distinct points, $\{p_i\}_{i=1}^\ell$.
Because the inclusion map $\iota:\surfmsinggraph \rightarrow \surf$ is Cauchy-continuous, it has a unique extension $\quotientmap:\completion{\surfmsinggraph} \rightarrow \surf$ in which $\quotientmap(p_i) = p$ for each $i = 1,\dots,\ell$.

Alternatively, this can be visualized as in Figure \ref{fig:completion}.
Here, level sets of geodesic distances from a point in red are given, with the cuts in dark green removed from the surface.
Notice that cutting curves locally divide otherwise connected domains, as seen by the geodesic distance colors. 
Under this metric, a Cauchy sequence entirely on one side of the cutting curve will converge to a point in $\completion{\surfmsinggraph}$ which differs from a Cauchy sequence converging to the same point in Euclidean space but defined on the other side of the cutting graph.

\begin{figure}
	\centering
	\includegraphics[trim=0cm 0cm 0cm 0.1cm, clip, width=.95\textwidth]{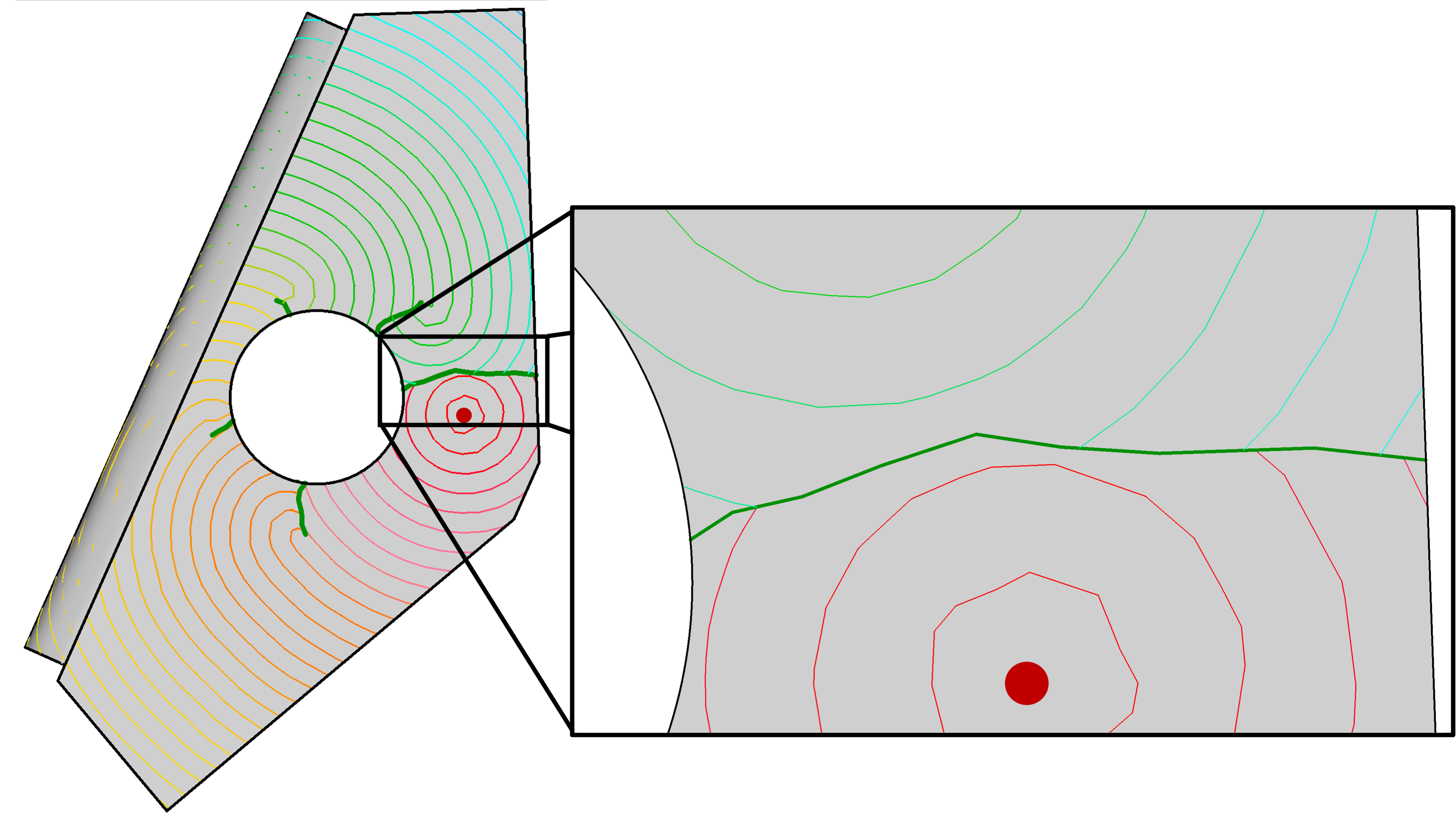}
	\caption{Level sets of distances on a surface from the point in red are displayed. Cuts, displayed in dark green, break continuity and alter distances (and thus neighborhoods) across a cut.  Under the topology defined by this metric, a Cauchy sequence on one side of the \cuttinggraph{} converging to a point in Euclidean space will converge to a different point than a Cauchy sequence converging to the same point in Euclidean space on the other side of the \cuttinggraph{}. }\label{fig:completion}
\end{figure}

Here, note that $\quotientmap$ is not only an inclusion operator, but also a quotient map from $\completion{\surfmsinggraph}$ to $\surf$.
The identification is $p \sim q \iff \quotientmap(p) = \quotientmap(q)$.
Throughout the remainder of this work, for $p \in \surf$ (respectively $A \subset \surf$) we take $\quotientmap^{-1}(p)$ (respectively $\quotientmap^{-1}(A)$) to mean the preimage of $p$ (respectively $A$) under the quotient map.

\subsection{Definition of the Immersion}

After removal of the \cuttinggraph{} from the surface, the local vector field representation of a \textcrossfield{} is both well-defined globally and integrable.

\begin{lemma}\label{lem:integrable_fields}
Given a \textquadmeshmetric{} $\emptyquadmeshmetric$ on $\surf$, the induced \textcrossfield{} on $\surfmsinggraph$ of $\quadcrossfield$ decomposes into four distinct, rotationally symmetric integrable vector fields $\{X_i\}_{i=0}^3$ which are well-defined over each connected component of $\surfmsinggraph$.
\end{lemma}
\begin{proof}
After cutting $\surf$ into a set of simply connected components, the holonomy group on each connected component is trivial.
Then parallel translation of any component of the \textcrossfield{} on a connected component of $\surfmsinggraph$ yields a well-defined vector field.
Integrability holds by construction.
\end{proof}


{\Bd Figure \ref{fig:frame_field_cuts} shows pictorially how, after removal of the \cuttinggraph{}, a \textframefield{} on a surface in Euclidean space can be decomposed into four well-defined vector fields.
Without introduction of these cuts, the smooth vector fields would not be well-defined, as seen by the {\Rd change of direction of the fields across the \cuttinggraph{}.}
(Recall that a \textframefield{} is a \textcrossfield{} under a non-Euclidean metric).
}

\begin{figure}
		\centering
		\begin{subfigure}{1\textwidth}
			\includegraphics[trim=0cm 0cm 0cm 0cm, clip, width=1.0\textwidth]{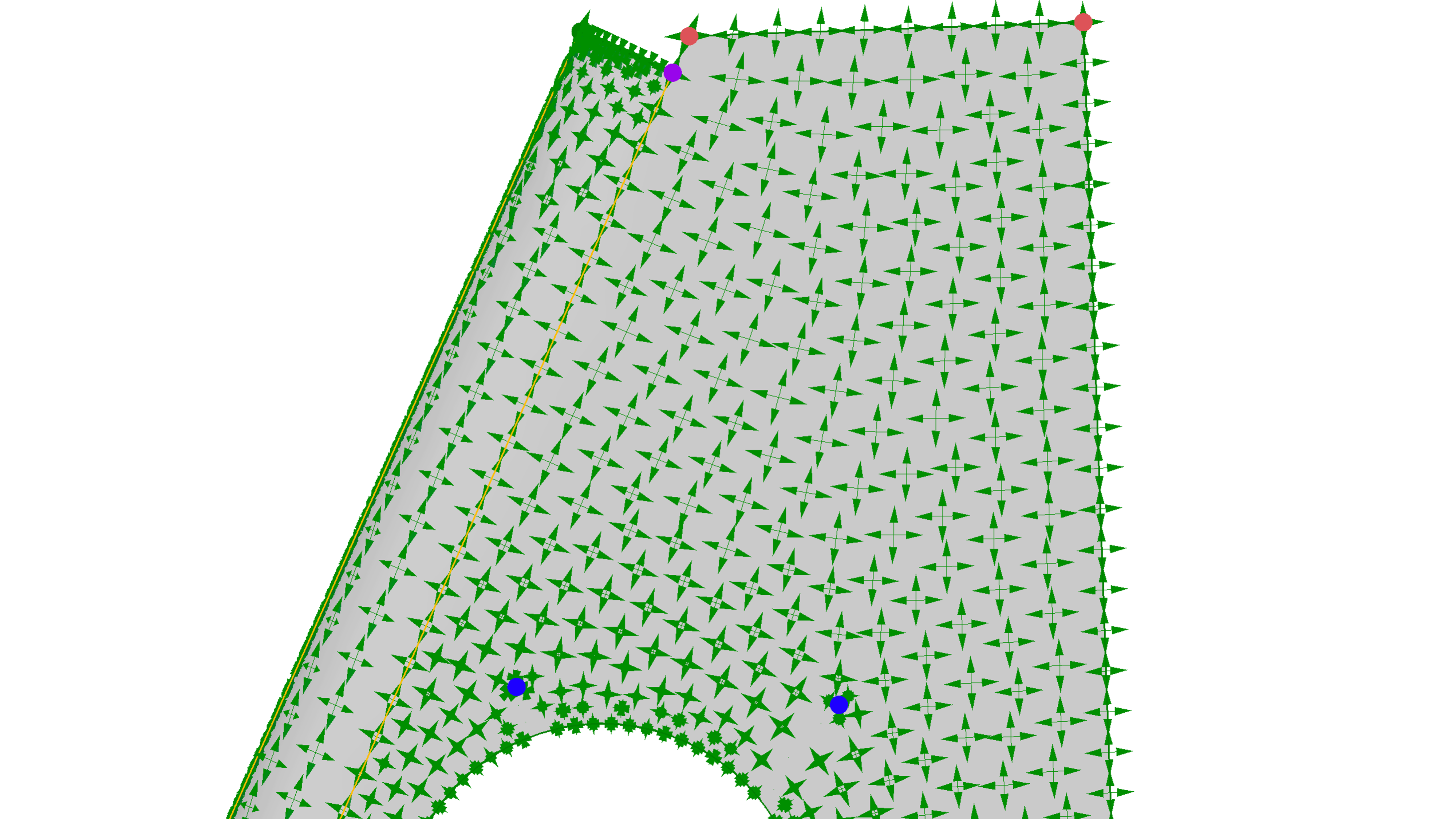}
			\caption{\Rd \textFramefield{} and singularities defined on the surface without cuts}
		\end{subfigure}
		\\
		\begin{subfigure}{0.485\textwidth}
			\includegraphics[trim=5cm 0cm 5cm 0cm, clip, width=1.0\textwidth]{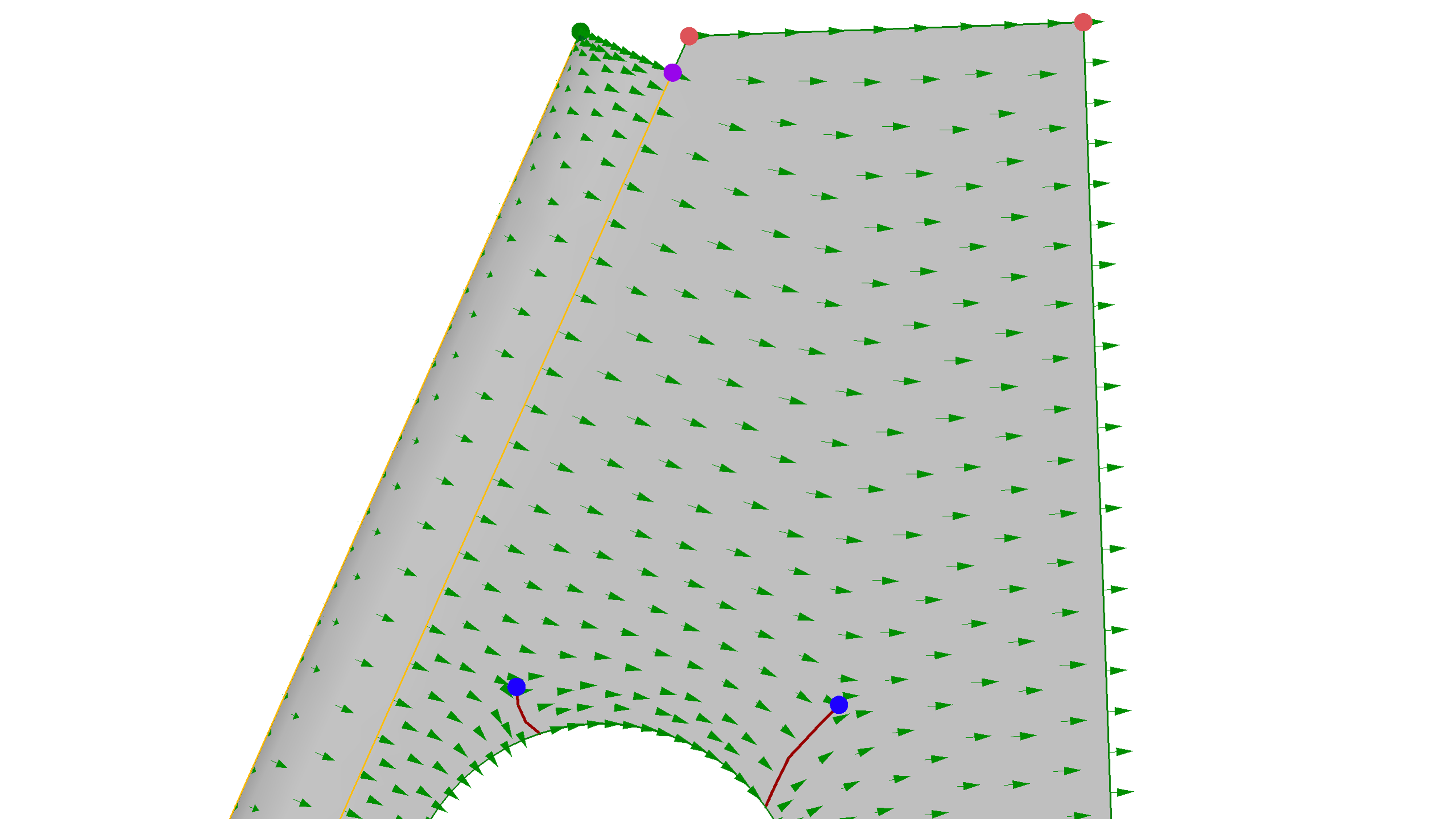}
			\caption{\Rd $1^{\mathrm{st}}$ vector field component after cutting }
		\end{subfigure}
		\;
		\begin{subfigure}{0.485\textwidth}
			\includegraphics[trim=5cm 0cm 5cm 0cm, clip, width=1.0\textwidth]{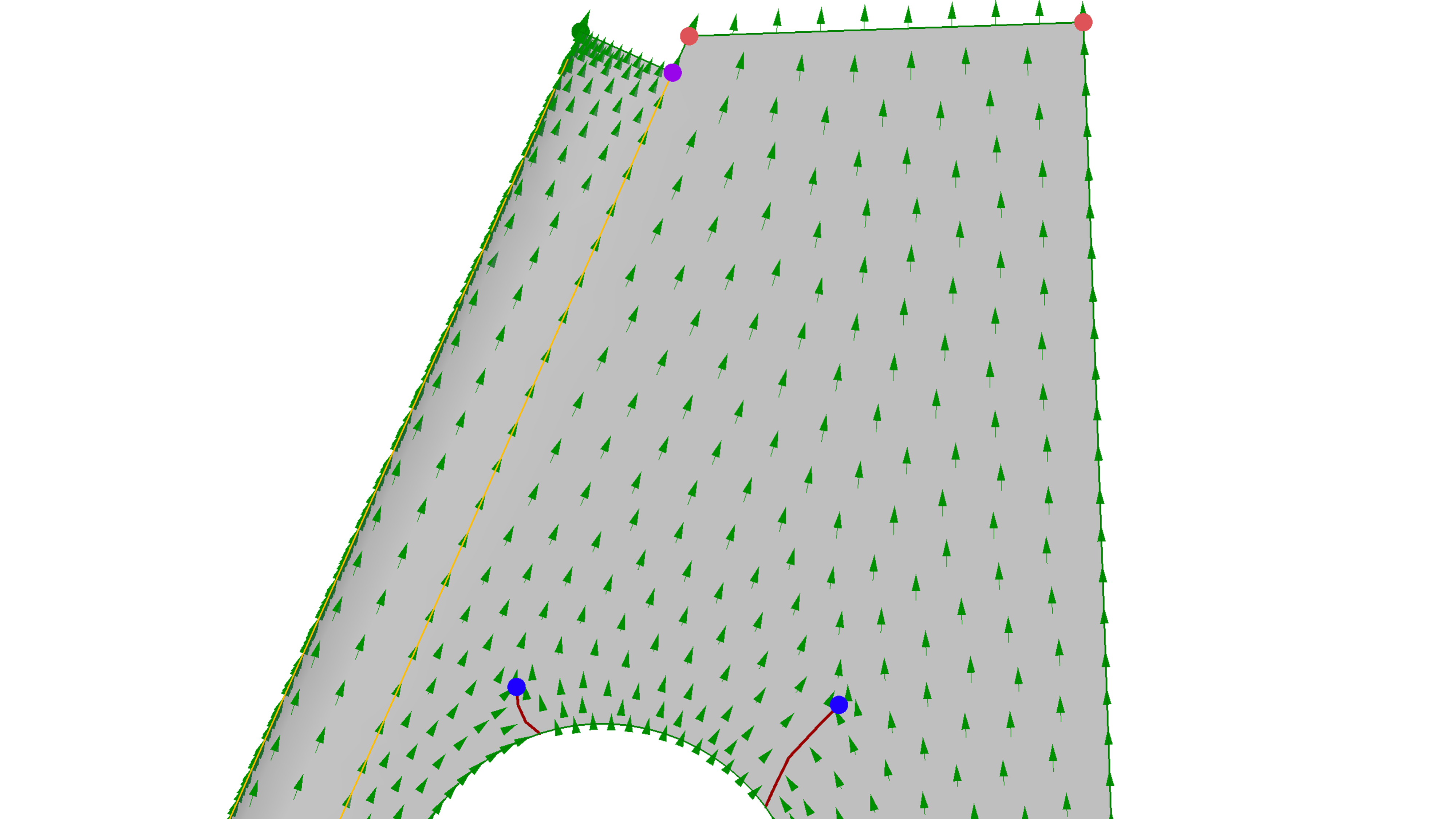}
			\caption{\Rd $2^{\mathrm{nd}}$ vector field component after cutting}
		\end{subfigure}
	\caption{{\Rd After removing the \cuttinggraph{} (in red) from the surface, a \textframefield{} can be decomposed into a set of four well-defined vector fields, two of which are linearly independent and shown above.
	Without introduction of these cuts, these smooth vector fields would generally not be well-defined, as is seen by the change in direction of the vector field along the boundary on either side of the cut.
	}}\label{fig:frame_field_cuts}
\end{figure}

With the results of Lemma \ref{lem:integrable_fields}, we are prepared to discuss an isometric immersion (via the \textquadmeshmetric{}) of $\surf-\left(\singpts\cup\graph\right)$ into the plane.

\begin{proposition}\label{prop:quad_to_immersion}
The \textquadmeshmetric{} $\emptyquadmeshmetric$ on $\surfmsinggraph$ induces a map 
\[
	\bar{\immersion}:\completion{\surfmsinggraph}\rightarrow \mathbb{R}^2
\]
which is an isometric immersion on $\surfmsinggraph$ and locally injective for each $p$ in which $\quotientmap(p) \not \in \singpts \cup \partial\graph$. 
This can be taken to have coordinate functions as integral curves of \textcrossfield{} $\quadcrossfield$ vectors.
\end{proposition}
\begin{proof}
We proceed under the assumption that $\completion{\surfmsinggraph}$ is a single connected component; if not, operate on each individually.

By construction, each point $p \in \completion{\surfmsinggraph}$ is of one of four types:
\begin{enumerate}
	\item $p \in \big(\surfmsinggraph\big) - \partial \surf$, in which case it has a neighborhood which is isometric to a disk in $\mathbb{R}^2$.
	\item $p \in \big(\surfmsinggraph\big) \cap \partial \surf$, in which case it has a neighborhood which is isometric to an open half-disk $HB_{\epsilon_p}$.
	\item $\quotientmap(p) \in \singpts \cup \mathcal{N}$, in which case it has a neighborhood isometric to a boundary cone $H\mathcal{C}(p, \curvaturemap_p), \curvaturemap_p \leq \curvaturemap\big(\quotientmap(p)\big)$.
	\item $\quotientmap(p) \in \graph - (\singpts \cup \mathcal{N})$, in which case it has a neighborhood which is isometric to a boundary cone $H\mathcal{C}(p, \pi)$.
\end{enumerate}

Let $\mathcal{U}$ be an open cover of $\completion{\surfmsinggraph}$ using the above neighbhorhoods.
Specifically, for $p_i$ in which $\quotientmap(p_i) \in \graph \cap (\singpts \cup \mathcal{N})$, take the neighborhood to be a subset of a small open ball of radius $\epsilon$ about $p_i$.
Define this set to be $\{B_{\epsilon}(p_i)\}_{i=1}^N$ for some finite $N$.
For each member of $\mathcal{U},$ subtract the closed balls $\{\bar{B}_{\frac{\epsilon}{2}}(p_i)\}_{i=1}^N$; call this $\tilde{\mathcal{U}}$.
Then $\tilde{\mathcal{U}} \cup \{B_{\epsilon}(p_i)\}_{i=1}^N$ is an open cover of $\completion{\surfmsinggraph}$ with only one (small) neighborhood on each $p_i$. Because $\completion{\surfmsinggraph}$ is compact, this set has a finite subcover, $\mathcal{V}$.

Let $U \in \mathcal{V} - \{B_{\epsilon}(p_i)\}_{i=1}^N$.
Then by construction, it is an isometric immersion (actually an embedding) on $U \cap \big(\surfmsinggraph\big)$, and locally injective on $U \cap \big(\completion{\surfmsinggraph}\big)$.
By Lemma \ref{lem:integrable_fields}, we can make integral curves of $X_0$ parallel to the $u$-axis and integral curves of $X_1$ parallel to the $v$-axis, potentially after a rotation and/or a reflection (both of which preserve isometries).
Define this map to be $\immersion_1$.

Next, let $V \in \mathcal{V} - \{B_{\epsilon}(p_i)\}_{i=1}^N, V \cap U \neq \emptyset, V \neq U$.
Because of the isometry $\varphi_V:V\rightarrow \mathbb{R}^2$, $\varphi_{V}$ is also an immersion and locally injective as described.
Furthermore, because lengths and angles are preserved, after a rotation, reflection, and/or translation, the map of $\varphi_V$ exactly equals $\immersion_1$ on $U \cap V$.
Then define $\immersion_2$ as $\varphi_V$ if $p\in V, \immersion_1$ if $p \in U$, which is well-defined and has the advertised properties of the final immersion map.

Proceeding inductively, we get an immersion $\tilde{\immersion}$ defined by $V \in \mathcal{V} - \{B_{\epsilon}(p_i)\}_{i=1}^N$.
Finally, each $B_\epsilon(p_i)$ is isometric to a boundary cone singularity.
Under this isometry's image (after a potential translation, rotation, or reflection), $\varphi_{B_\epsilon(p_i)}\big(B_\epsilon(p_i)\big)$ must exactly align with $\tilde{\immersion}$ where intersections with $B_\epsilon(p_i)$ and another $V \in \mathcal{V}$ are non-empty.
These are isometries away from the actual cone point, and locally injective on all boundary points.
Then the desired immersion is the union of the map $\tilde{\immersion}$ with maps $\varphi_{B_\epsilon(p_i)}$ over all of $\completion{\surfmsinggraph}$.

Finally, note that for boundary cones with angle $\curvaturemap < 2\pi$, local injectivity holds. 
Thus, the result can be relaxed to guarantee local injectivity on all members of $\mathcal{N} \cup \singpts$ for which this is the case.
\end{proof}

Beyond being an immersion, the map $\bar{\immersion}$ has additional important qualities that will be used to show an equivalence between it and a \affineflatmetric{}.
The following definitions will be used to make some of these qualities clear.

\begin{definition}[{\Rd Conical Function}]\label{def:conical_function}
A \textbf{conical function} $\phi:\surf\rightarrow \mathbb{R}$ is a (discontinuous) function, together with a discrete set $\singpts \subset \surf,$ such that 
	\[
		\phi(p) =
		\begin{cases}
			2\pi & \text{ if } p \in \interior{\surf} - \singpts\\
			\pi & \text{ if } p \in \partial \surf - \singpts\\
			\phi_p \in \frac{k\pi}{2}, k \in \mathbb{Z}_{>0} & \text{ otherwise. }
		\end{cases}
	\]
Furthermore, it obeys the following Gauss-Bonnet relationship:
	\begin{equation}
		\sum_{q \in \partial \surf} \left(\pi - \phi(q)\right) 
			+ \sum_{p \in \interior{\surf}} \left(2\pi - \phi(p)\right) 
			= 2\pi \chi(\surf)
	\end{equation}
\end{definition}

\begin{definition}[{\Rd \AffineFlatImmersion{}}]\label{def:affine_flat_immersion}
Let $\emptymetric$ be a Riemannian metric defined on all of surface $\surf$.
	Let $\phi$ be a conical function (see Definition \ref{def:conical_function}), together with its discrete set $\singpts,$ and a \cuttinggraph{} $\graph(S,\singpts)$ (see Definition \ref{def:cutting_graph}).
	Write $\distancemetric:\surfmsinggraph\times\surfmsinggraph\rightarrow\mathbb{R}$ as the distance metric on $\surfmsinggraph$ induced by $\emptymetric$, and the induced topology denoted as 
	$\toposurfmsinggraph$.
	A function $\bar{\immersion}:\completion{\toposurfmsinggraph}\rightarrow \mathbb{R}^2$ is defined to be a \textbf{\affineflatimmersion{}} if it satisfies the following:
	\begin{description}
		\item[Q1] $\bar{\immersion}$ is locally injective in a neighborhood of each $p \in \completion{\toposurfmsinggraph} -  \quotientmap^{-1}  (\singpts \cup \mathcal{N})$, and is an orientation-preserving smooth immersion  on $\toposurfmsinggraph$ whose Jacobians are bounded from above and below by constants independent of location.
		\item[Q2] 
		For any $p \in \surf$  with $\{\bar{p}_i\}_{i=1}^n = 
		 \quotientmap^{-1}  (p)$, there is some simply-connected open neighborhood of $U_p \subset \surf$ with $U_{\bar{p}_i}$ being connected components of the completion of $U_p-\graph$, written $\completion{(U_p - \graph)_d}$ such that $\bar{\immersion}(\completion{(U_p-\graph})_d)$ is either
		\begin{itemize}
			\item Isometric to a set of boundary cones, $H\mathcal{C}\big(\bar{\immersion}(U_{\bar{p}_i}),\phi_{\bar{p}_i}\big)$ with $\sum_{i=1}^n \phi_{\bar{p}_i} = \phi(p).$ (These correspond to points that are either in $ \graph$ or $\partial \surf$).
			\item A single connected component which is isometric to a two-dimensional Euclidean ball. (These correspond to points of $\surf$ that are not in $\graph$ or $\partial \surf$.)
		\end{itemize}
		\item[Q3] 
		Let $E \in \mathcal{A}$ be an oriented arc of the \cuttinggraph{} $\graph$. 
		Define $\omega_-,\omega_+:(0,1)\rightarrow \completion{\toposurfmsinggraph}$ as the curves associated with $E$ in the completion.
		Then for any $t \in (0,1), \bar{\immersion}\big(\omega_-(t)\big) = T\Big(\bar{\immersion}\big(\omega_+(t)\big)\Big)$ for $T:\mathbb{R}^2\rightarrow\mathbb{R}^2$ a translation and rotation by $\frac{j\pi}{2}, j \in \mathbb{Z}.$
		\item[Q4] 
		Under $\bar{\immersion}$, each connected component of $\quotientmap^{-1}\big(\partial{\surf} - (\singpts\cup\graph)\big)$ is an open arc with constant Euclidean coordinates in $u$ or $v$.
	\end{description}

\end{definition}

\begin{proposition}\label{prop:partial_metric_to_partial_immersion}
A \affineflatmetric{}, together with a set of (boundary) cone singularities $\singpts$ and a \textbf{cutting graph} $\graph$, induce a \affineflatimmersion{}.
\end{proposition}
\begin{proof}

	Any orientable surface can be embedded in $\mathbb{R}^3$, and as such inherits the Euclidean metric of $\mathbb{R}^3$.
	Take $\phi:\surf\rightarrow\mathbb{R}$ by $\phi(p) = \curvaturemap(p)$ as the conical function and $\graph(\surf,\singpts)$ as a \cuttinggraph{}.
	Take $\bar{\immersion}:\completion{\surfmsinggraph}\rightarrow \mathbb{R}^2$ by Proposition \ref{prop:quad_to_immersion}. 
	Because the topologies on $\surf$ and $\completion{\surfmsinggraph}$ are strongly equivalent via Lemma \ref{lem:completion_same_nograph} and Corollary \ref{corol:completion_same_graph}, all maps are equivalently represented on the topology of the \textquadmeshmetric{} or on the Euclidean metric.
	But $\completion{\surfmsinggraph}$ is compact, so $\bar{\immersion}$  has Lipschitz constants globally bounded from above and below by positive constants. Then Property \textbf{Q1} holds.
	
	Next, because $\bar{\immersion}$ is isometric, cones will be split into sets of boundary cones and boundary cones into boundary cones of smaller angle whose sum add to the original cone.
	Then Property \textbf{Q2} holds.
		
	Now, let $E \subset \mathcal{A}$ be an arc in the \cuttinggraph{} $\graph$, and $\omega_-,\omega_+$ its associated curves in the completion of $\completion{\surfmsinggraph}$. 
	Because $\bar{\immersion}$ is isometric, the length and (geodesic) curvature of $\bar{\immersion}(\omega_-)$ pointwise equals that of $\bar{\immersion}(\omega_+)$.
	Then they are just a rotation and translation of each other.
	But some component of $\quadcrossfield$ was chosen to align to the $u$ axis for each connected component of $\completion{\surfmsinggraph}$ by construction of $\bar{\immersion}$.
	Then the component vector of $\quadcrossfield$ represented in $\bar{\immersion}(\omega_-)$ can only be represented by a vector which has rotated by $\frac{k\pi}{2}$ radians in $\bar{\immersion}(\omega_+)$.
	Hence, Property \textbf{Q3} holds.
	
	Again, by definition of $\bar{\immersion}$ there are unique the vector fields $X_0, X_1$ on $\completion{\surfmsinggraph}$ induced by the \textcrossfield{} $\quadcrossfield$ of the \textquadmeshmetric{}.
	Because the original \textcrossfield{} is boundary-aligned, so each component of $\partial\surf - (\singpts \cup \graph)$ is an arc with constant Euclidean coordinates in either $u$ or $v$ under the immersion mapping $\bar{\immersion}$(Property \textbf{Q4}).
	
\end{proof}

Next, a partial converse is shown.

\begin{proposition}\label{prop:partial_converse}
A \affineflatimmersion{} induces a flat metric on $\surf$ with cone singularities $\singpts$ obeying Properties \textbf{P1}, \textbf{P2}, and \textbf{P3} of Definition \ref{def:qmetric}
\end{proposition}
\begin{proof}

	First, note that $\bar{\immersion}$ is an immersion on $\toposurfmsinggraph$, and as such, $\bar{\immersion}^*(\emptymetric_{\mathbb{R}^2})$ pulls back to a flat Riemannian metric $\emptymetric_{\mathbb{R}^2}$.
	This in turn yields another metric on the domain, denoted $\euclideandistancemetric$ with its associated topology $\toposurfmsinggraphrtwo$.
	The completion of this space is then denoted as $\completion{\toposurfmsinggraphrtwo}$.
	However, by boundedness of the Jacobians of $\immersion$, it is bi-Lipshitz locally with global bounds on the Lipschitz constants.
	Thus the metrics $\euclideandistancemetric$ and $\distancemetric$ are strongly equivalent, so the topology of $\toposurfmsinggraphrtwo$ is equivalent to the topology of $\toposurfmsinggraph$, and the completions are identical.
	With this, we may transfer assumptions defined on the topology of $\toposurfmsinggraph$ to the topology $\toposurfmsinggraphrtwo$.
	
	For an edge $E \subset \mathcal{A}$, let $\omega_-,\omega_+$ boundary segments of $\completion{\toposurfmsinggraphrtwo}$, with $q_- \in \omega_-, q_+ \in \omega_+, \quotientmap_{\mathbb{R}^2}(q_-) = \quotientmap_{\mathbb{R}^2}(q_+)$. 
	Let $Y_0,Y_1$ be unit vectors locally parallel and orthogonal, respectively, to $\bar{\immersion}(\omega_-)$ at $\bar{\immersion}(q_-)$. 
	Then by Property \textbf{Q3}, the corresponding vectors in the tangent space of $\bar{\immersion}(q_+)$ are $R(Y_0),R(Y_1)$, where $R$ is a rotation by $\frac{k\pi}{2}, k \in \mathbb{Z}$. 
	Then as in Lemma \ref{lem:integrable_fields}, these locally correspond to coordinate systems $w^1,w^2$ and $Rw^1,Rw^2$. 
	In these coordinates, the metric tensor near $\bar{\immersion}(q_-)$ takes the form $\emptymetric_{\mathbb{R}^2} = \sum_{i,j=1,2} \delta_{ij} dw^i dw^j$, and near $\bar{\immersion}(q_+)$ it is $\emptymetric_{\mathbb{R}^2} = \sum_{i,j=1,2} \delta_{ij} d(Rw^i) d(Rw^j)$. 
	Under the quotient map these basis vectors for the tangent spaces are equivalent.
	Because the coordinates functions of the metric tensor are identical for both, the metric tensor has a well-defined extension in the quotient topology.
	This identification is locally consistent by Property \textbf{Q3}.
	Then under the quotient map there is a well-defined flat metric on $\surf - (\singpts \cup \mathcal{N})$.

	Now, let  $p \in \surf \cap (\singpts \cup \mathcal{N})$ and $\quotientmap^{-1}(p) = \{p_i\}_{i=1}^n $.
	Then by Property \textbf{Q2}, the quotient of all of their neighborhoods will be a boundary cone of angle $\phi(p) = \sum_{i=1}^n \phi_{p_i}$ if on the boundary of $\surf$ or a cone of the same angle if in the interior of $S$.
	Then the Riemannian metric can be extended to all $p \in \mathcal{N} - \singpts$ because the tangent space on these nodes is well-defined and consistent with the rest of the manifold.
	Then the surface is a flat manifold with cone singularities (see Definition \ref{def:flat_w_cone}), with map $\curvaturemap(p) := \phi(p)$.
	Then by $\textbf{Q1}$, condition \textbf{P1} holds.
	Denote this flat metric as $\emptyflatmetric$.
	
	Now, let $X_p$ be a vector in $T_p\surf$ for $p \in \surfmsinggraph$.
	Choose any loop.
	For a flat manifold, the homotopy class of the loop preserves the holonomy.
	Thus we can choose the loop to be transverse to $\graph$ by a smooth homotopic deformation of the loop.
	Call this deformed, transverse loop $\gamma$.
	Let $\connection{t}(X)$ {\Bd be} the Levi-Cevita connection map by parallel translation on this loop.
	Because intersections of compact sets are compact, there are a finite number of intersection points between the image of $\gamma$ and $\graph$.
	Taking the restriction of $\connection{t}(X)$ onto $\surfmsinggraph$ and then its pushforward (and extension) to $\bar{\immersion}\big(\completion{\surfmsinggraph}\big)$, we find that $\connection{0}(X)$ corresponds exactly to some vector $(u_X,v_X)$ at $\bar{\immersion}(p)$.
	Because $\mathbb{R}^2$ is flat, parallel translation of the vector between points of intersection with $\graph$ will keep the immersed vector having the same representation.
	However, across cuts, the vector will rotate by $\frac{j\pi}{2},j\in \mathbb{Z}$ according to Property \textbf{Q3}. 
	The ultimate rotation in $T_p\surf$ induced by parallel translation along this loop must, then, be the sum of $\frac{j\pi}{2}, j \in \mathbb{Z}$.
	But both the point and loop were arbitrary up to homotopy, so the total holonomy group must be a subset of $\mathcal{R}$, giving Property \textbf{P2}.
	
	Now, let $X_1 = (1,0), X_2 = (0,1)$ be orthogonal unit vectors at $\bar{\immersion}(p) \in \mathbb{R}^2$ for some $p \in \surfmsinggraph$. 
	Pull back both to $\surfmsinggraph$ using the immersion map; view these as a two vectors in $\surf - \singpts$.
	Define a unit cross in $T_p\surf{}$ by the vectors $\{X_1,X_2,-X_1,-X_2\}$, which are orthonormal under $\emptyflatmetric$.
	Because Property \textbf{P2} holds, this cross can be parallel translated over all of $\surf-\singpts$ to yield a global \textcrossfield{}.
	Property \textbf{Q4} ensures that the \textcrossfield{} obeys \textbf{P3}.

\end{proof}

To establish a notion of equivalence between a \affineflatmetric{} and a \affineflatimmersion{}, we need to describe the notion of an integral curve on the immersion.
To accomplish this, we use the following {\Bd notations and} definitions.

{\Bd First, when two curves, $\gamma_1, \gamma_2 : (-\epsilon, \epsilon) \rightarrow \surf$ are transversal at $p = \gamma_1(0) = \gamma_2(0)$, this is denoted as $\gamma_1 \transverse_p \gamma_2$.
}
Let $\bar{\immersion}_u, \bar{\immersion}_v:\completion{\surfmsinggraph}\rightarrow \mathbb{R}$ be functions defined by $\bar{\immersion}(p) = \big(\bar{\immersion}_u(p),\bar{\immersion}_v(p)\big)$ in Euclidean coordinates.

\begin{definition}[{\Rd Coordinate Lines}]\label{def:coord_line}
Let $p \in \completion{\surfmsinggraph} - \singpts.$  Define
\begin{equation}
	\coordline{p}{u} = \left\{ q \in \completion{\surfmsinggraph} - \quotientmap^{-1}(\singpts): \bar{\immersion}_u(q) = \bar{\immersion}_u(p), q \text{ connected to } p \text{ in } \bar{\immersion}^{-1}_u(p)\right\}. 
\end{equation}
Similarly, let
\begin{equation}
	 \coordline{p}{v} = \left\{ q \in \completion{\surfmsinggraph} - \quotientmap^{-1}(\singpts): \bar{\immersion}_v(q) = \bar{\immersion}_v(p), q \text{ connected to } p \text{ in } \bar{\immersion}^{-1}_v(p)\right\}. 
\end{equation}
Then $\coordline{p}{u}, \coordline{p}{v}$ are called the \textbf{coordinate lines} of $p$ under $\bar{\immersion}$.
\end{definition}

\begin{lemma}\label{lem:extension}
For a \affineflatimmersion{}, fix $p \in \completion{\surfmsinggraph} - \quotientmap^{-1}(\singpts)$ with coordinate line $\coordline{p}{u}$ (respectively $\coordline{p}{v}$) non-discrete and intersecting $\partial \big(\completion{\surfmsinggraph{}}\big)$ at $q$. 
If  $\quotientmap(q) \not \in \singpts$ and $\quotientmap(\coordline{p}{u})$ (respectively $\quotientmap(\coordline{p}{v})$) {\Bd $\not \transverse_{\quotientmap(q)} \partial \surf$}, the quotient map of the coordinate line has a extension across the \cuttinggraph{} that is geodesic for $\emptyflatmetric{}$.
\end{lemma}
\begin{proof}
First, under the quotient topology, the coordinate line is a  curve in $\surf$.
Using the flat metric of Proposition \ref{prop:partial_converse}, the coordinate line is a geodesic (being a line in the immersion).
Because the metric is flat, there is a neighborhood of $\quotientmap(q), U_{\quotientmap(q)} \subset \surf$ with an isometry $\varphi: U_{\quotientmap(q)} \rightarrow \mathbb{R}^2$.
	Then $\varphi\big(\quotientmap(\coordline{p}{})\big)$ is a line segment, which can be extended uniquely to a line segment dividing $U_{\quotientmap(q)}$.
	Call this $\gamma$.
	The set $\quotientmap(\coordline{p}{}) \cup \varphi^{-1}(\gamma)$ is such an extension.
\end{proof}
Such an extension will be called a $\textbf{geodesic extension}$ of the coordinate line at $q$.

Now, for $A \in \completion{\surfmsinggraph}$, define the sets 
\[
	u_+(A) = \left\{q \in \bar{A}: \bar{\immersion}_u^{-1}\big(\max_{p \in \bar{A}} \bar{\immersion}_u(p)\big)\right\},
\]
and
\[
	u_-(A) = \left\{q \in \bar{A}: \bar{\immersion}_u^{-1}\big(\min_{p \in \bar{A}} \bar{\immersion}_u(p)\big)\right\}.
\]
Let $v_+(A), v_-(A)$ be defined analogously.

\begin{definition}[{\Rd Quotient Curve}]
Fix $p \in \completion{\surfmsinggraph} - \quotientmap^{-1}(\singpts)$ with $\coordline{p}{u} =: \coordline{p_0}{} \neq \{p_0\}$.
Take $q_{0+} \in v_+(\coordline{p_0}{})$.
Inductively define $\coordline{p_i}{}$ for $i > 0$ as follows:
\begin{itemize}
	\item If $\quotientmap(q_{i-1}) \in \singpts$, terminate the induction.
	\item {\Bd If $\quotientmap(q_{i-1}) \in \partial \surf, \quotientmap(\coordline{p_{i-1}}{}) \transverse_{\quotientmap(q_{i-1})} \partial \surf,$ terminate the induction.}
	\item Define $\gamma_1 := \quotientmap(\coordline{p_{i-1}}{})$ and $\gamma_2$ as its geodesic extension at $q_{i-1}$ (see Lemma \ref{lem:extension}) in a neighborhood $U_{\quotientmap(q_{i-1})}$ isometric to a ball in $\mathbb{R}^2$ via $\varphi$.
	Then $\omega := \overline{\varphi^{-1}(\gamma_2 - \gamma_1)}$ is a closed arc in $\surf$.
	Take a non-discrete connected component $A$ of $\quotientmap^{-1}(\omega)$ with $A \ni \tilde{q} \in \quotientmap^{-1}(q_{i-1}).$
	Define $\coordline{p_i}{}$ as the unique coordinate curve containing $A$.
	Define $q_{i}$ by the following
	\[
		q_i \in \begin{cases}
			v_+(\coordline{p_i}{}) & \text{ if } \bar{\immersion}_u(\coordline{p_i}{}) \text{ is discrete}, \tilde{q} \in v_-(A)\\
			v_-(\coordline{p_i}{}) & \text{ if } \bar{\immersion}_u(\coordline{p_i}{}) \text{ is discrete}, \tilde{q} \in v_+(A)\\
			u_+(\coordline{p_i}{}) & \text{ if } \bar{\immersion}_v(\coordline{p_i}{}) \text{ is discrete}, \tilde{q} \in u_-(A)\\
			u_-(\coordline{p_i}{}) & \text{ otherwise.}
		\end{cases}
	\]
\end{itemize}
Define $\coordline{p_{i}}{}$ for $i < 0$ analogously, with $q_{0-} \in v_-(\coordline{p_0}{})$.
Then the \textbf{quotient curve} of $p$ in the $u$ direction is defined to be the set
\begin{equation}
	\quotientcurve{p}{u} = \barcoordline{p_0}{u} \cup \Big(\bigcup_{i=1}^{N+} \barcoordline{p_i}{} \Big)\cup\bigcup_{i=1}^{N_-} \Big(\barcoordline{p_{-i}}{}\Big)
\end{equation}
where $0 \leq N\pm \leq \infty$. 

Similarly, the \textbf{quotient curve} of $p$ in the $v$ direction is defined to be the set
\begin{equation}
	\quotientcurve{p}{u} = \barcoordline{p_0}{v} \cup \Big(\bigcup_{i=1}^{N+} \barcoordline{p_i}{} \Big)\cup\bigcup_{i=1}^{N_-} \Big(\barcoordline{p_{-i}}{}\Big)
\end{equation}
where, as before, $\barcoordline{p_0}{v}$ is not discrete.
\end{definition}
A \textbf{finite} quotient curve is one in which $N\pm < \infty$.
A \textbf{periodic} quotient curve is one in which $\coordline{p_i}{} = \coordline{p_j}{}$ for some $i,j \in \mathbb{Z}$ both contained in the quotient curve.
{\Bd A finite, self-intersecting quotient curve comprised of two coordinate lines on a partial quad layout immersion is displayed on the right side of Figure \ref{fig:finite_quotient_curve}.
The integral curves of the pullback of the coordinate one-forms onto the original spatial geometry with cuts are depicted on the left of the same figure.}

\begin{figure}
		\centering
	\includegraphics[trim=0cm 0cm 0cm 0cm, clip, width=1.0\textwidth]{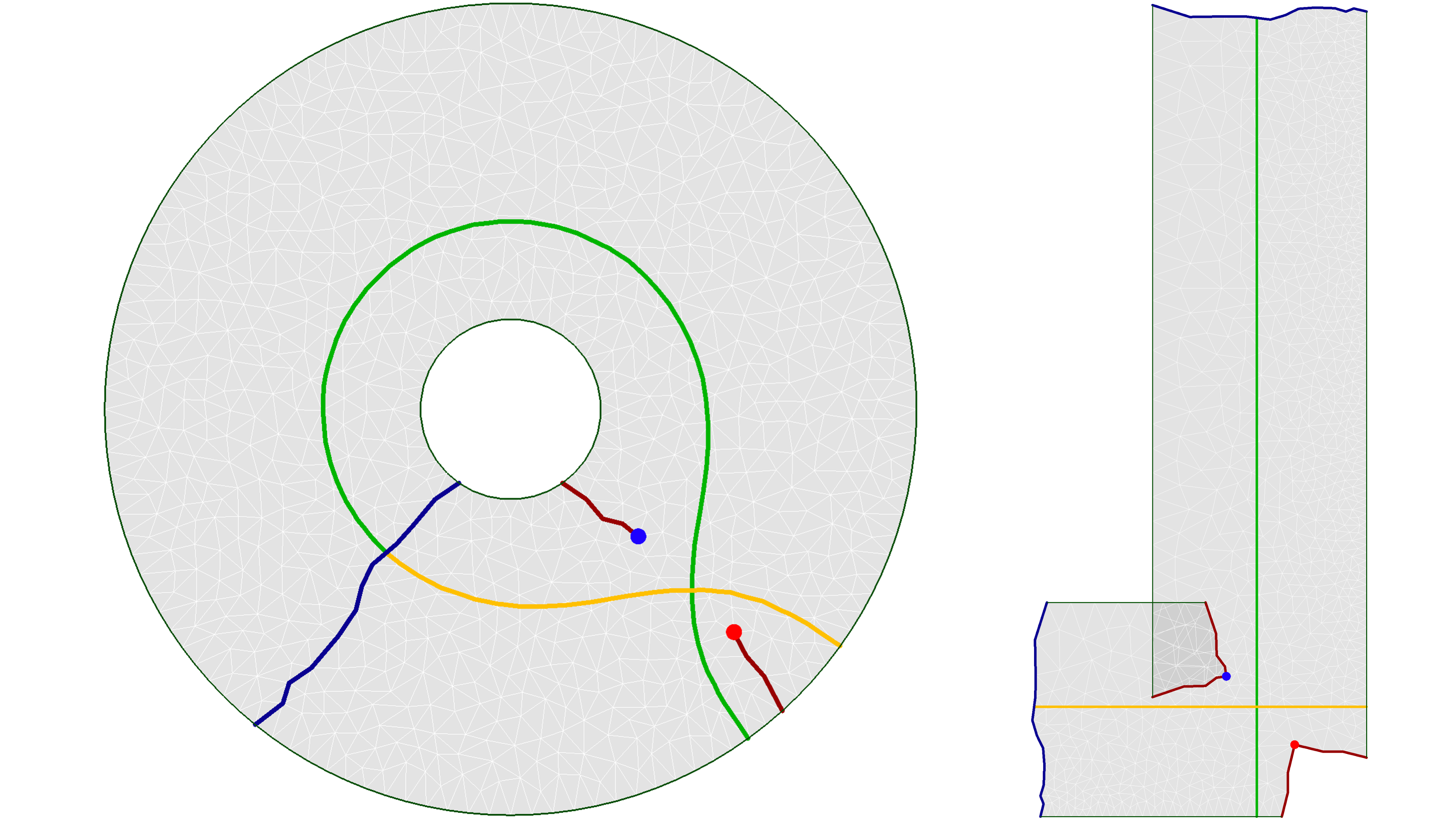}
	\caption{{\Bd A cut annulus with two cone singularities (on the left) is mapped via a (partial) quad layout immersion to the domain on the right.
	Two coordinate lines, shown in green and gold, comprise a single finite quotient curve in this immersion.
	The curves in green and gold on the left geometry are the pullback of the coordinate one-forms tracing the green and gold (respectively) coordinate lines in the immersed geometry.
	Notice that the quotient curve may intersect itself, but can only do so with orthogonal coordinate lines in the immersed geometry. 
	}}\label{fig:finite_quotient_curve}
\end{figure}

\begin{lemma}\label{lem:quotient_curve}
Given a \affineflatimmersion{} $\bar{\immersion}$ on $\surf$, quotient curves of $\bar{\immersion}$ are geodesic integral curves of the induced \textcrossfield{} on $\surf$ from $\bar{\immersion}$.

Conversely, {\Bd let $\omega$ be a \textcrossfield{} integral curve of a \affineflatmetric{} on $\surf$. 
Then the preimage $\quotientmap^{-1}(\omega)$ contains a quotient curve on its induced \affineflatimmersion{}.}
\end{lemma}
\begin{proof}
These follow by definition of the quotient curve and by construction of $\bar{\immersion}$ in Proposition \ref{prop:quad_to_immersion}.
\end{proof}

{\Bd Combining the results of Lemma \ref{lem:quotient_curve} with Propositions \ref{prop:partial_metric_to_partial_immersion} and \ref{prop:partial_converse}, the following equivalence holds.}

{\Bd
\begin{theorem}\label{thrm:equivalence_of_partials}
	A \affineflatmetric{} on $\surf$, together with a \cuttinggraph{} induces a \affineflatimmersion{}. 
	Similarly, an \affineflatimmersion{} induces a unique \affineflatmetric{}.
\end{theorem}
}

Finally, we are ready to define a \quadmeshimmersion{,} the object which induces a \textquadmeshmetric{} on $\surf$.
\begin{definition}[{\Rd \QuadMeshImmersion{}}]\label{def:quad_mesh_immersion}
	A \textbf{\quadmeshimmersion{}} $\bar{\immersion}$ is a \affineflatimmersion{} with the following additional property:
	\begin{description}
		\item[Q5] 
			All quotient curves emanating from singularities are finite.
			For a torus or annulus without singularities, two transverse quotient curves are finite or periodic.
	\end{description}
\end{definition}

\begin{theorem}\label{thrm:equivalence}
	A \quadmeshimmersion{} induces a \textquadmeshmetric{}. 
	Conversely, a \textquadmeshmetric{} on $\surf$ with cone singularities $\singpts$, accompanied by a \cuttinggraph{} $\graph(\surf,\singpts)$, induces a \quadmeshimmersion{}.
\end{theorem}
\begin{proof}
	($\Leftarrow$)
	{\Bd This result holds from Theorem \ref{thrm:equivalence_of_partials}.}
	
	($\Rightarrow$)
	{\Bd
	Properties \textbf{P1}, \textbf{P2},  \textbf{P3}, and \textbf{P4} hold by Theorem \ref{thrm:equivalence_of_partials}.}
	Then {\Bd it is only necessary to show} that the assumptions on singular quotient curves guarantee that all integral curves of the induced \affineflatmetric{} are finite.
	
	By property \textbf{Q3}, boundaries are coordinate curves, which combine to form quotient curves.
	Note that quotient curves of boundaries must have finite length.
	If a singularity lies on the boundary, this holds by property \textbf{Q5}.
	Otherwise, the quotient curve traces the entire boundary component and must be periodic.
	
	A finite length quotient curve containing a singularity must terminate at either two (possibly identical) singularities or a singularity and a boundary.
	Combined with finiteness or periodicity of boundary curves, these curves will segment the surface into a cellular structure.
	(For a torus with no singularities or an annulus without singularities, property \textbf{Q5} guarantees finite-length quotient curves, which also split into cellular structure).
	Each of these cells must necessarily be a quadrilateral to ensure that the cellular structure aligns with integral curves of the induced \textcrossfield{} without {\Bd addition of} cone singularities.
	
	Pick some quadrilateral cell on the surface and some $p$ in its interior.
	Construct an isometric immersion $\tilde{\immersion}$ of the quadrilateral cell containing $p$ under the newly-defined metric on $\surf, \emptyquadmeshmetric$, as was done in Proposition \ref{prop:quad_to_immersion}.
	This is a parameterization of the cell.
	But the cell was chosen arbitrarily, so the skeletal structure of the separatrices, combined with the new metric, induces a quadrilateral layout on $\surf$.
	Then by the equivalence of a quadrilateral layout to a \textquadmeshmetric{}, $\emptyquadmeshmetric{}$ must be a \textquadmeshmetric{}, which has periodic or finite-length integral curves.
\end{proof}

\begin{remark}
A \quadmeshimmersion{} is a generalization of both a translation surface and  a half-translation surface, which are used in Riemannian surface theory and are equivalent to holomorphic one forms and holomorphic quadratic forms, respectively. A generalized translation surface was introduced in \cite{Chen:2019} to motivate \textquadmeshmetric{s}, but criteria were not given which would make them a \quadmeshimmersion{}. Furthermore, the immersed space of a surface under the image of a \quadmeshimmersion{} need not be a polygon.
\end{remark}

{\Bd Thus far, the \quadmeshimmersion{} has relied on principals of smooth topology to define a bijective, well-defined immersion into the plane.
It is worth noting, however, that the immersion mapping} on the completion just needs to be locally bi-Lipschitz with globally bounded Lipschitz constants for all theory to hold.
As a result, the theory for extracting a \quadlayout{} from a \quadmeshimmersion{} explicitly generalizes to compact Lipschitz (and thus two-dimensional piecewise-linear \cite{Rosenberg:1987}) surfaces.
As such, the following comparison with surface triangulations and integer grid maps can be made.

\begin{proposition}\label{prop:integer_grid_map}
	Assume $\surf$ is a piecewise-linear triangulation. A boundary-aligned integer grid map \cite{Bommes:2009} that is also an immersion is a \quadmeshimmersion{}. However, there exist \quadmeshimmersion{s} which are not integer grid maps.
\end{proposition}

\begin{proof}
	An integer grid map cuts a surface into disk topology and cuts to all singularities, after which it seeks a planar immersion (written embedding in \cite{Bommes:2009, Ebke:2013}) of the surface.
	By computation of a viable \textcrossfield{}, a discrete Poincare-Hopf formula is met \cite{Knoppel:2013} which is analogous to the discrete Gauss-Bonnet formula of \textbf{Q1}.
	By boundary alignment, \textbf{Q3} is met.
	Furthermore, by rotational constraints on cut edges, \textbf{Q4} is met.
	These two, in conjunction with the singularity index of nodes, implies \textbf{Q2}.
	
	Then the only other object of interest is \textbf{Q5}.
	But singular vertices are constrained to lie in $\mathbb{Z}^2$. 
	Furthermore, opposite sides of an arc $E \in \mathcal{A}$ of the surface are constrained to be rotations and integer translations of each other.
	As such, a quotient curve with coordinate line lying on an integer value continues to another coordinate line lying on an integer value.
	Then quotient curves of singularities must lie entirely on integers.
	But $\completion{\surfmsinggraph}$ is compact, so there are only a finite number of members of $\mathbb{Z}^2$ in $\bar{\immersion}(\surfmsinggraph)$.
	Thus quotient curves must be finite length or periodic.
	
	Next, as a simple converse, take the Euclidean axis-aligned rectangle with diagonal bisecting line from vertices $(0,0)$ and $(a,b)$, with $a,b$ irrational.
	This is a \quadlayout{}.
	Identifying pairs of sides will also lead to a \quadlayout{} on an annulus or a torus.
\end{proof}

Note that the above example of a non-integer coordinate layout is in some sense trivial, and can easily be rescaled.
In general, after a set of rescalings, any parametric description of a \quadlayout{} can be given parametric integer coordinates.
However, this rescaling must generally be globally anisotropic.
Without prior knowledge of the entire cell structure, producing such a globally anisotropic rescaling is highly nontrivial.
Furthermore, without care it will distort the original metric.

%% file: Results.tex
\section{Computational Results}\label{sec:computation}

In this section, some basic results are shown to computationally verify the above theory.
As suggested by Proposition \ref{prop:integer_grid_map}, no computation makes use of integer grid maps.
{\Rd Instead, linear constraints enforcing that the \quadlayout{} align to surface features and boundaries are prescribed as  \cite{Hiemstra:2020};
topological constraints enforcing quotient curves to be of finite length are enforced as proposed in \cite{Campen:2014b,Hiemstra:2020}.}

\subsection{Annulus with Singularities}

Perhaps the simplest non-trivial example which demonstrates the above theory is presented by an annulus with a cone singularity on a cone of angle $\frac{3\pi}{2}$ and a cone singularity on a cone of angle $\frac{5\pi}{2}$.
Such a configuration is displayed in Figure \ref{fig:annulus}, with the cone of larger angle shown as a point in blue, the lesser angle cone in red.
After computation using discrete surface Ricci flow \cite{Jin:2008,Yang:2009} with Neumann boundary conditions, the surface is cut to a topological disk (curves in blue) with additional cuts to the singularities (cuts in red).
After cutting, vertices and edges of the mesh along the cuts are multiply defined, as expected from the completion topology.
Using this information, the completion of the cut surface can then be immersed as seen on the right of Figure \ref{fig:annulus}.

Notice that the image on the right is not an embedding, but rather simply an immersion.
This can be seen by the domain overlap near the point in blue on the left.
Furthermore, note that the boundary curves (shown in green) are represented by curves with constant $u$ and $v$ coordinates.
Additionally, positive and negative sides of cut curves under the immersion are given by rotations of $\frac{k\pi}{2}, k \in \mathbb{Z}$ and translations.
Finally, the \quadlayout{} is given as unions of curves which are constant in $u$ or $v$, and is depicted by the curves given in black and green (the minimal set of  quotient curves).
Though discontinuous in the cut topology, these curves are continuous and smooth in the original surface topology which is rebuilt as the quotient space of the cut mesh topology.
Notice that all quotient curves are either periodic (e.g. green boundary curves) or finite (e.g. all curves emanating from singularities).

\subsection{Curved L-Shaped Domain}

\begin{figure}
	\centering
	\includegraphics[trim=0cm 0cm 0cm 0cm, clip, width=.95\textwidth]{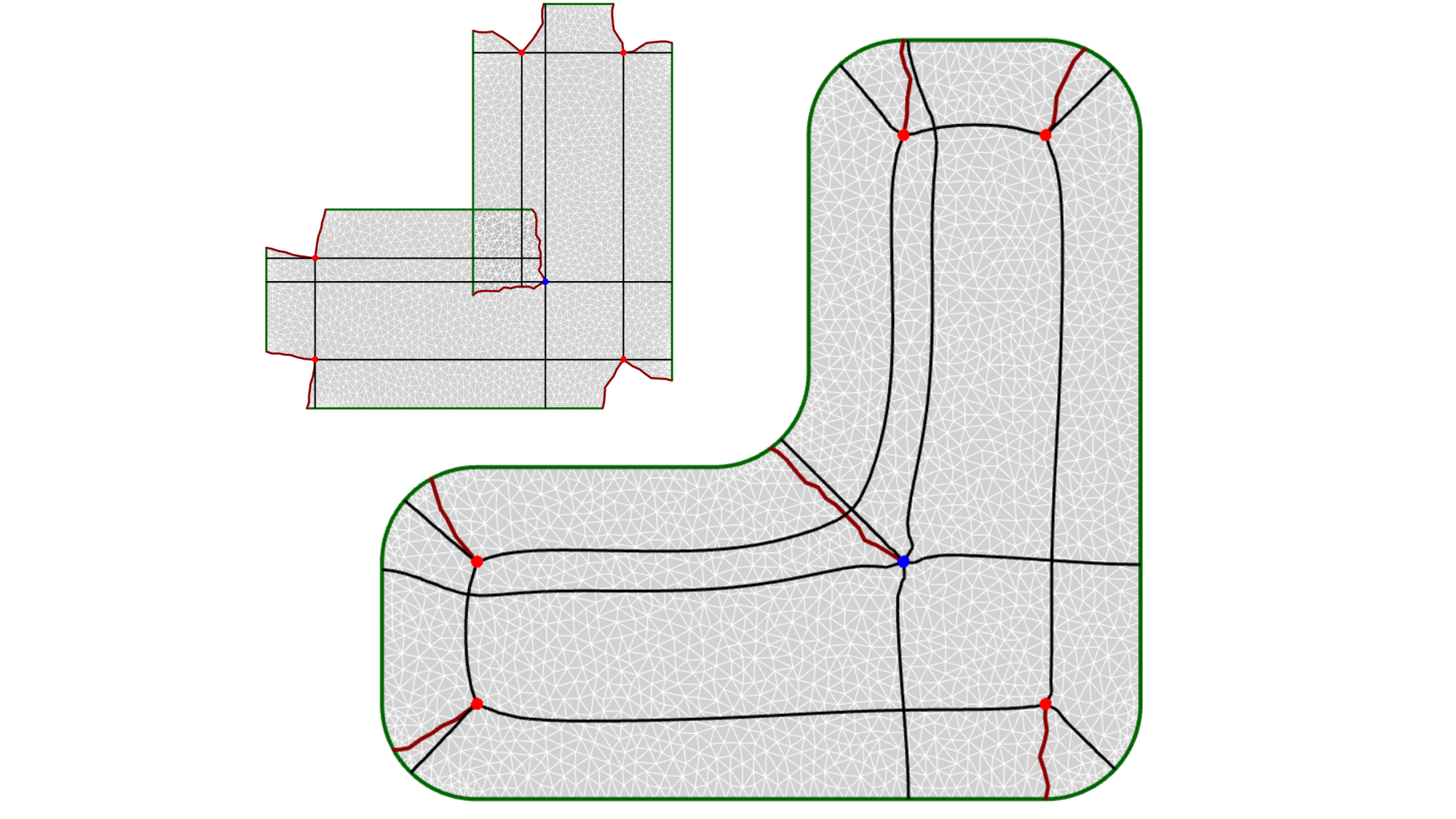}
	\caption{A \quadlayout{} (below-right) and its \quadmeshimmersion{} (above-left) of a curved L-shaped domain is shown.
	Here, an internal cone singularity with cone of angle $\frac{5\pi}{2}$ is shown in blue, while internal cone singularities with cones of angles $\frac{3\pi}{2}$ red. 
	Note that curvilinear arcs on the original surface geometry emanating from singularities (black) and tracing boundaries (green) are combinations of curves on the immersed geometry which have constant $u$ or $v$ coordinate.
	Because this is a topological disk, no cuts are necessary to simplify the surface homology, though jagged curves in red are cuts made to cone singularities.
	Note that in the immersed topology, opposite sides of the same curves have parametric coordinates that are simple rotations of each other by $\frac{k\pi}{2}, k \in \mathbb{Z}$.
	 \label{fig:l_shaped}}
\end{figure}

Another basic geometry illustrating the proposed theory is the curved L-shaped domain of Figure \ref{fig:l_shaped}.
This surface was generated by imposing constraints as in \cite{Hiemstra:2020} on boundaries, holonomy, and connectivity between singularities while  minimizing a symmetric Dirichlet energy \cite{SmithSchaefer:2015} using composite majorization \cite{Shtengel:2017} via Progressive Parameterizations \cite{Liul:2018}. 
Again, note that curvilinear objects in the quadrilateral layout are straight lines in the immersion which have been glued together.
Here, the boundary quotient curve is periodic, while all other curves displayed in the layout are finite and terminate between a singularity and a boundary or between two singularities.

\subsection{Double Torus}

{\Bd The subtle difference between two partial \quadmeshimmersion{s}\textemdash{}one with finite length quotient curves and one without\textemdash{}is depicted in Figure \ref{fig:double_torus}.
Here, though both geometries have very similar \cuttinggraph{s} and cone singularity structures, one has finite-length quotient curves while the other heuristically represents an immersion that does not (for which extraction of quotient curves is instead terminated prematurely).
Both immersions are initially computed by integrating holomorphic one-forms \cite{Gu:2003}.
}

\ifthenelse{\boolean{isELS}}
{
\begin{figure}
		\centering
	\begin{subfigure}{1\textwidth}
	\centering
		\includegraphics[trim=0cm 0cm 0cm 0cm, clip, width=0.9\textwidth]{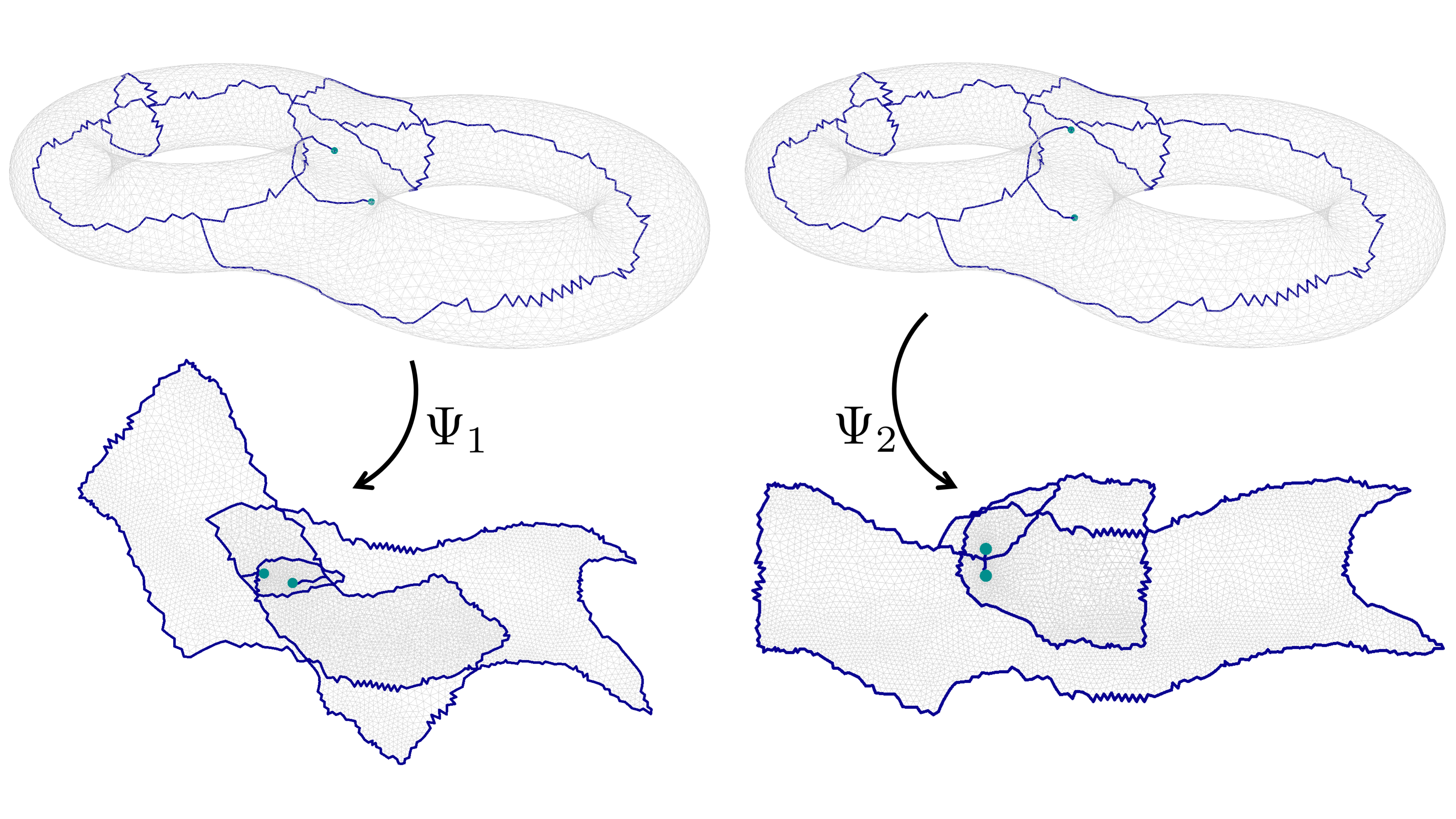}
		\caption{Two partial \quadmeshimmersion{s} on a double torus}
	\end{subfigure}
	
	\begin{subfigure}{1\textwidth}
		\centering
		\includegraphics[trim=0cm 0cm 0cm 0cm, clip, width=0.9\textwidth]{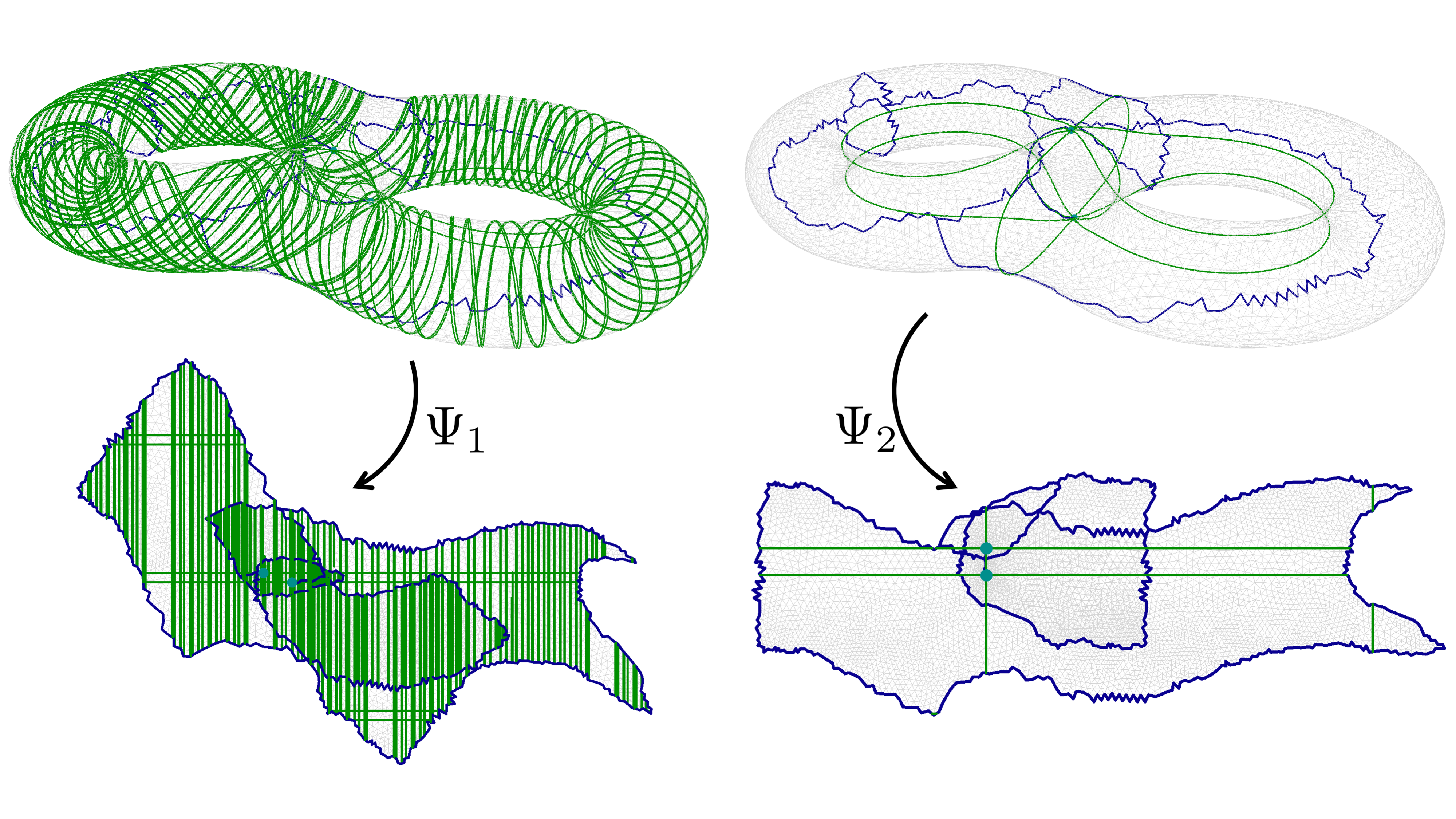}
		\caption{The set of quotient curves emanating from singularities of each immersion (in green)}
	\end{subfigure}
	\caption{{\Bd Two partial \quadmeshimmersion{s} are computed as the integrals of different holomorphic one-forms \cite{Gu:2003} with slightly different singularity locations.
	Above, \cuttinggraph{s} (in blue) and immersions for both structures are shown. 
	Below, all quotient curves emanating from singularities on both immersions are extracted and represented on the spatial domain.
	While the quotient curves of $\Psi_2$ are finite, the quotient curves of $\Psi_1$ are terminated prematurely.
	Heuristically, this represents a partial \quadmeshimmersion{} that does not have finite length quotient curves.
	}}\label{fig:double_torus}
\end{figure}
}
{
\begin{figure}
		\centering
	\begin{subfigure}{1\textwidth}
		\includegraphics[trim=0cm 0cm 0cm 0cm, clip, width=1.0\textwidth]{Partial_vs_Full_Quad_Layout_img1}
		\caption{Two partial \quadmeshimmersion{s} on a double torus}
	\end{subfigure}
	
	\begin{subfigure}{1\textwidth}
		\includegraphics[trim=0cm 0cm 0cm 0cm, clip, width=1.0\textwidth]{Partial_vs_Full_Quad_Layout_img2}
		\caption{The set of quotient curves emanating from singularities of each immersion (in green)}
	\end{subfigure}
	\caption{{\Bd Two partial \quadmeshimmersion{s} are computed as the integrals of different holomorphic one-forms \cite{Gu:2003} with slightly different singularity locations.
	Above, \cuttinggraph{s} (in blue) and immersions for both structures are shown. 
	Below, all quotient curves emanating from singularities on both immersions are extracted and represented on the spatial domain.
	While the quotient curves of $\Psi_2$ are finite, the quotient curves of $\Psi_1$ are terminated prematurely.
	Heuristically, this represents a partial \quadmeshimmersion{} that does not have finite length quotient curves.
	}}\label{fig:double_torus}
\end{figure}
}

\subsection{DEVCOM Generic Hull Bracket}

In addition to simple geometries, the theory of the \quadmeshimmersion{} also applies to more complicated geometries {\Bd of engineering interest}. 
In Figure \ref{fig:sym_bracket}, a \quadlayout{} and its respective feature-aligned \quadmeshimmersion{} on the symmetric part of a bracket of the DEVCOM Generic Hull vehicle are displayed.
The immersion here was computed using the approach outlined in \cite{Hiemstra:2020}, with the additional caveat that face-based singularities of the \textframefield{} were transferred to vertex-based singularities by subdividing the face and attaching the singularity to the vertex at the original face's barycenter.

\begin{remark}
In general, the cone singularity structure on a surface (e.g. embedded in $\mathbb{R}^3$) will not simply be a subset of that on the \textquadmeshmetric{}.
Particularly for this bracket, the surface embedded in three-dimensional Euclidean space has boundary cone singularities at all points with sharp angles, including all of those shown in magenta and red, as well as two displayed in green.
These points are typically called ``features.''
In this metric, there are no internal singularities.
Under the flat metric induced by the \quadmeshimmersion{}, each red and magenta point gets mapped to a boundary singularity.
However, the two feature points in green are mapped to boundaries in the immersion with constant $u$ coordinate.
As a result, these points are not boundary cones when viewed under the flat metric.
{\Rd Similarly, the \textquadmeshmetric{} introduces four interior singularities (shown in blue) that are not present in the original geometry. 
Here, boundary alignment constraints require cut segments of the bracket hole to be mapped to (geodesic) lines of constant $u$ or $v$ coordinates under the \textquadmeshimmersion{}.
The geodesic curvature lost by the surface in straightening these boundary segments in the \textquadmeshmetric{} is instead concentrated in these four additional interior singularities to ensure that the Gauss-Bonnet condition is still obeyed.}
By sampling all cone singularities---in both the flat metric as well as the three-dimensional Euclidean setting---a geometry-aware, feature preserving \quadlayout{} is attained.
Finally, note that feature curves in the original geometry (shown in gold) must also be represented by a set of lines in $u$ or $v$ only on the immersed geometry to remain features in the final parameterization.
\end{remark}

\begin{figure}
	\centering
	\includegraphics[trim=0cm 0cm 0cm 0cm, clip, width=.95\textwidth]{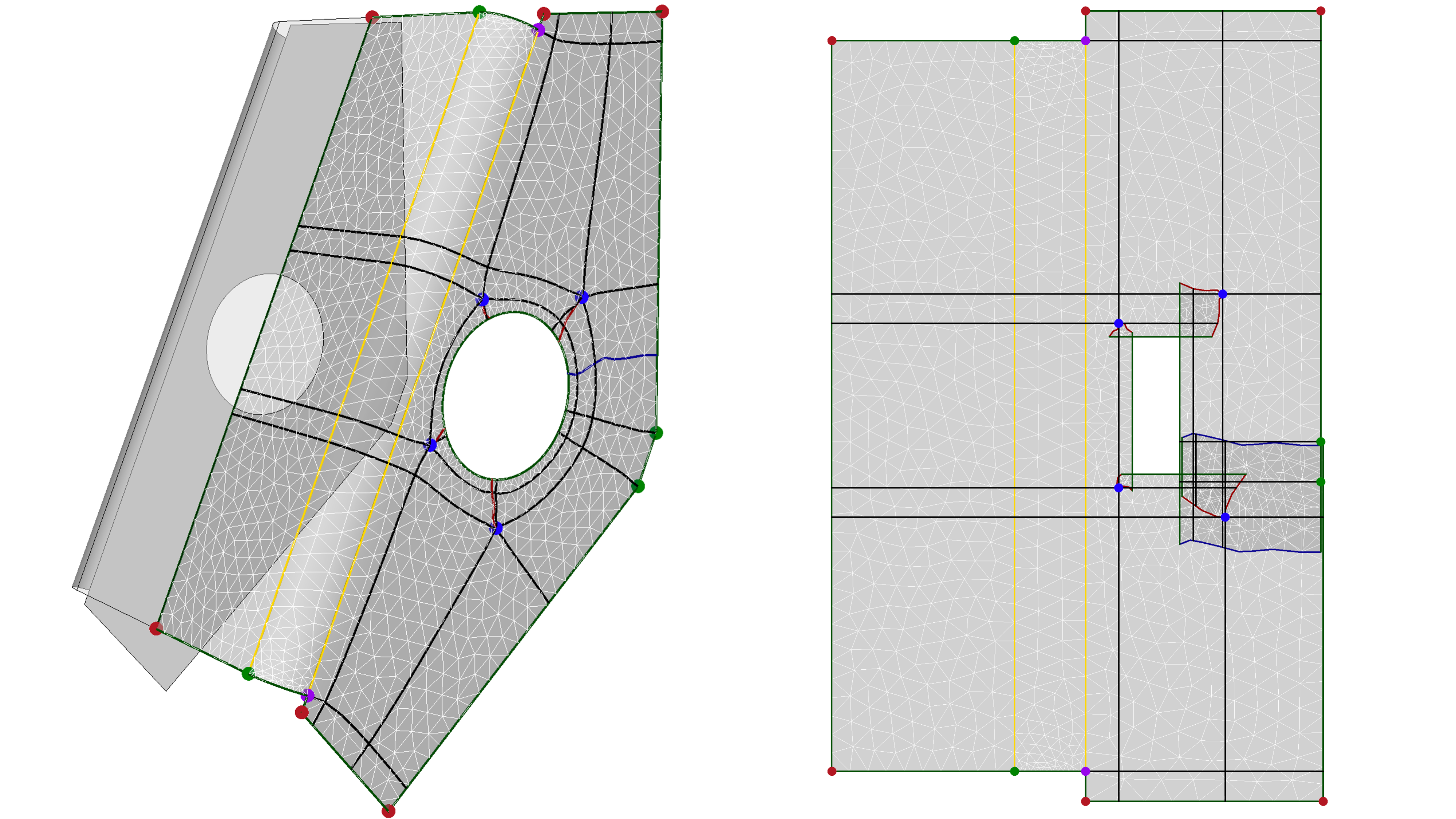}
	\caption{A \quadlayout{} (left) and its \quadmeshimmersion{} (right) of the symmetric part of a bracket of the DEVCOM Generic Hull is shown.
	Here, internal cones of angle $\frac{5\pi}{2}$ are shown in blue, boundary cones of angles $\frac{3\pi}{2}$ and $\frac{\pi}{2}$ are shown in magenta and red, respectively, and feature nodes which are not parametric cones are displayed in green. 
	Note that curvilinear arcs on the original surface geometry emanating from singularities (black), tracing boundaries (green), and following geometric features (gold) are combinations of curves on the immersed geometry which have constant $u$ or $v$ coordinate.
	Curves in blue represent cuts made on the mesh to make the bracket a topological disk, while curves in red are cuts made to cone singularities.
	Note that in the immersed topology, opposite sides of the same curves have parametric coordinates that are simple translations and rotations of each other by $\frac{k\pi}{2}, k \in \mathbb{Z}$.
	 \label{fig:sym_bracket}}
\end{figure}

%% file: Conclusion.tex
\section{Conclusion}\label{sec:conclusions}

This paper mathematically establishes an equivalence between a special type of immersion---a \quadmeshimmersion{}---and a quadrilateral layout on a surface.
The theory for the \quadmeshimmersion{} holds for both smooth and \pl{} topology, making this immersion a strong computational tool.
Furthermore, the \quadmeshimmersion{} generalizes the notion of an integer grid map.
{\Bd Simple tests} were run to show the computational viability of the theory.

While the \quadmeshimmersion{} provides a general paradigm in which to extract quadrilateral layouts of surfaces, generation of \quadmeshimmersion{s} is still a non-trivial operation.
Many techniques for generation of \quadlayout{s} seek a set of $u$ and $v$ coordinates on surface triangulations that  satisfy the above constraints (e.g. \cite{Bommes:2013,Bommes:2009,Hiemstra:2020}.) 
However, these techniques cannot generally guarantee local invertibility of the parametetization (immersion) or finite length quotient curves (integral curves of \textframefield{s}).
Future research will explore computational methods with better guarantees.

{\Bd Much of the} current quadrilateral-only mesh generation literature exploits integer grid maps (see e.g. \cite{Bommes:2013,Bommes:2009,Campen:2015,Ebke:2013}).
While other frameworks (e.g. \cite{Campen:2012,Campen:2014,Hiemstra:2020,Viertel:2019,Viertel:2020}) for quadrilateral mesh generation exist, it is not clear that these have been explored fully.
Future work should look into methods which can yield quadrilateral layouts without the stipulations imposed by integer grid maps.

Finally, this work lays out mathematical theory, but it does not provide a comprehensive computational framework for how to extract a \quadlayout{} from a \quadmeshimmersion{}.
Subsequent work will describe data structures and algorithms necessary to extract a \quadlayout{} from a valid \quadmeshimmersion{}.

%% file: Mathematical_Background.tex
\ifthenelse{\boolean{isELS}}
{
\section{Supplementary Material}\label{sec:supplement}
\subsection{Mathematical Concept Review}\label{sec:math_background}
}
{
\section{Mathematical Concept Review}\label{sec:math_background}
}

The following section reviews a number of concepts from topology and geometry. Woven through this will be a discussion on concepts essential to the topological and geometric story of \quadlayout{s}. The reader more familiar with point set, algebraic, and differential topology, as well as differential geometry, may wish to simply skim this for relevant highlights and notation.


\ifthenelse{\boolean{isELS}}
{
\subsubsection{Topological Concepts}\label{sec:topology}
}
{
\subsection{Topological Concepts}\label{sec:topology}
}

Topology describes the notion of proximity between objects.
Of particular interest in topology is whether functions preserve proximity (called continuous functions) and whether these functions have nice properties such as reversibility.
A function which is continuous will be called a \textbf{map}.
An invertible continuous function whose inverse is also continuous is called a \textbf{homeomorphism}.
For our purposes, we say a neighborhood is ``connected'' if for every point in the neighborhood there is a path entirely in the neighborhood to every other point in the neighborhood.

For general domains  $X$ (not necessarily embedded in Euclidean space), a \textbf{metric} is a function $d:X\times X \rightarrow \mathbb{R}$ in which
\begin{enumerate}
	\item $d(x,y) \geq 0,$ with equality only if and only if $x = y$
	\item $d(x,y) = d(y,x)$
	\item $d(x,z) \leq d(x,y) + d(y,z)$ for all $y$
\end{enumerate}
A space with a metric is called a \textbf{metric space}, and it has an induced topology defined by open sets at each $x \in X$ whose distance from a base point is less than some value, and written
\[
	B_r(x) = \{ y \in X: d(x,y) < r\}.
\]
A \textbf{Cauchy sequence} is an infinite set of points $\{x_i\}_{i=1}^\infty\subset X$ such that for any $\epsilon > 0, \exists N \in \mathbb{N}$ such that $d(x_M,x_{M+1}) < \epsilon$ for each $M \geq N$.
A metric space $X$ is \textbf{complete} if every Cauchy sequence converges to a point in $X$.
The \textbf{Cauchy completion} of a metric space $X$ is the union of $X$ with the set of all limit points of Cauchy sequences in $X$, and is unique for the metric space.
If $\{x_i\}_{i=1}^\infty$ is a sequence in $X$ such that for any $\epsilon > 0, \exists N \in \mathbb{N}$ such that $d(x_M,x) < \epsilon$ for each $M \geq N$, the sequence limits to $x$, and $x$ is the \textbf{limit point}.
A set is \textbf{closed} if it contains all of its limit points.

In the Euclidean plane $\mathbb{R}^2$, an open ball of radius $r > 0$ at a point $p$, denoted $B_r(p)$ is
\[
	B_r(p) := \{ x \in \mathbb{R}^2: d_{\mathbb{R}^2}(p,x) < r \},
\]
where $d_{\mathbb{R}^2}:\mathbb{R}^2 \times \mathbb{R}^2 \rightarrow \mathbb{R}$ is the typical Euclidean metric.
Taken as a Banach space, this metric is written $d_{\mathbb{R}^2}(p,x) = ||p-x||_{\mathbb{R}^2}$.
An open \textbf{half-ball} at the origin of radius $r > 0$ is written as 
\[
	HB_r := \{ x = (x_1,x_2) \in \mathbb{R}^2: ||x||_{\mathbb{R}^2} < r, x_2 \geq 0 \}.
\]

Frequently, spaces are described as {\Bd an amalgamation} of local phenomena. 
An \textbf{open cover} of a domain $X$ is a family $\mathcal{U}$ of open sets $U_\iota$ such that $\bigcup_{\iota} U_\iota = X$.
By definition, a \textbf{compact} space is one in which every open cover has a subfamily which also covers $X$ and is finite.
This subfamily is called a \textbf{subcover}. 

\begin{figure}
	\centering
	\includegraphics[trim = 0cm  0cm 0cm 0cm, clip, width=.95\textwidth]{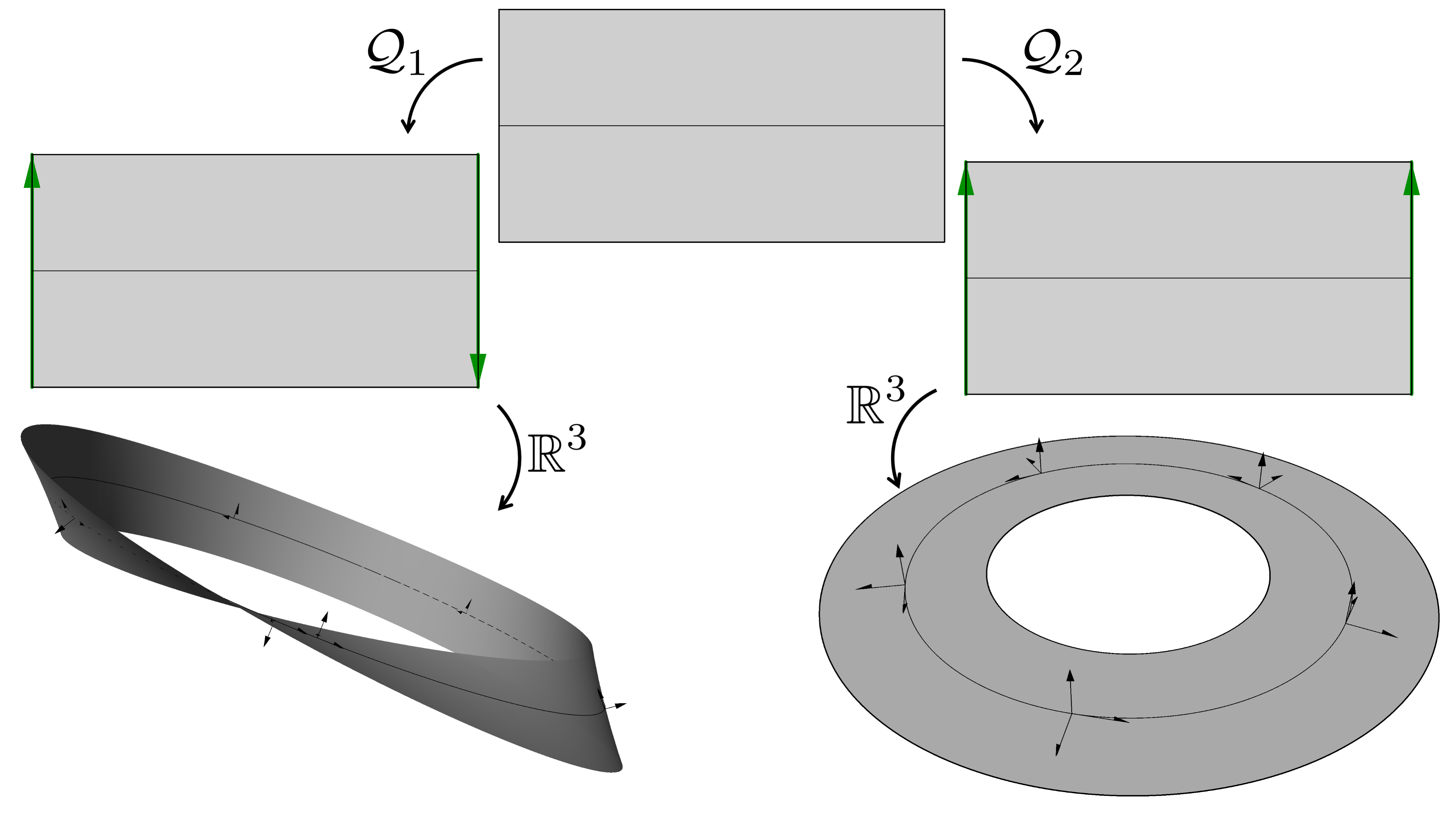}
	\caption{An initial rectangle is operated on by two different quotient maps, $\quotientmap_1$ and $\quotientmap_2$. Under both, opposite sides of the rectangle are glued to make the resulting arrow directions align. The first results in a M\"obius band, while the second gives an annulus. Note that when each is embedded in $\mathbb{R}^3$, the normals of the {\Bd annulus are well-defined} when transported about a loop, while the normals on the M\"obius band are not. }\label{fig:rectangle_quotient}
\end{figure}

A topology of particular interest in this work is the \textbf{quotient topology}, which is given as a mathematical gluing operation.
This is most easily expressed pictorially.
In Figure \ref{fig:rectangle_quotient}, two edges of a rectangle are glued in such a way that arrows on the glued sides, after gluing, point in the same direction.
The first of these results in a M\"obius band, while the second yields a cylinder.
Note that after gluing, a much shorter path can be used to travel between points on either side of the boundary than before. 

Another topology of interest in this paper is that of a two-manifold, also known as a surface. 

A \textbf{closed} surface is an object in which every point has a neighborhood which is homeomorphic to a two-dimensional open ball.
An \textbf{open surface} is an object in which every point has a neighborhood which is either homeomorphic to either a two-dimensional open ball or a two-dimensional open half-ball.
The \textbf{interior} of a surface is the set of all points with neighborhood homeomorphic to  an open ball; the \textbf{boundary} is the remainder of the surface.
When a surface's boundary is not empty, the union of all connected boundary points will be a set of simply closed curves  called the \textbf{boundary components}. 
Here, \textbf{simple} means a curve that does not intersect itself.

\begin{figure}
\centering
	\includegraphics[trim = 0cm  0cm 0cm 0cm, clip, width=.65\textwidth]{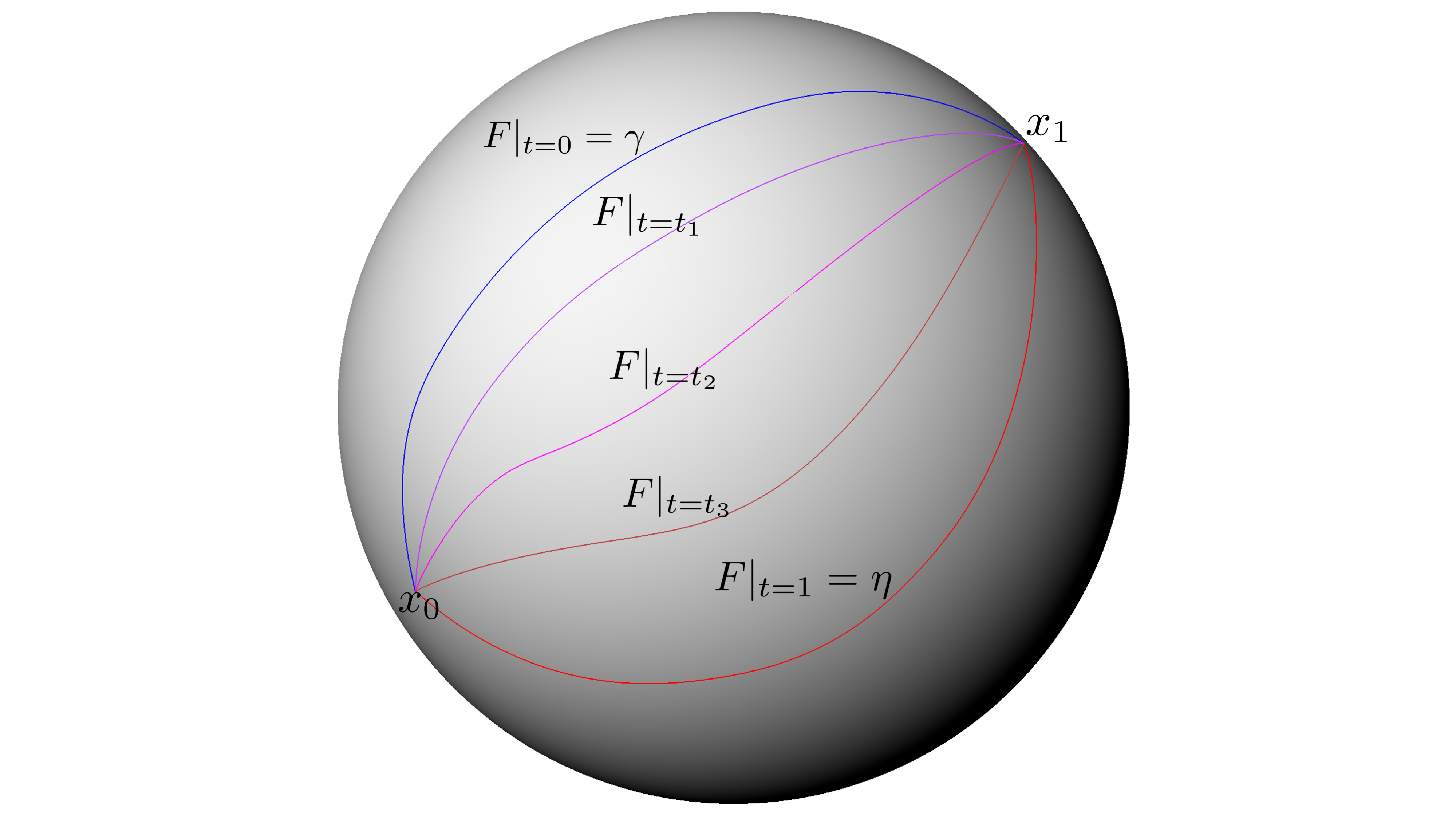}
	\caption{A homotopy between curves $\gamma_0$ and $\gamma_1$ on a sphere, relative to their endpoints. }\label{fig:homotopic_sphere_curves}
\end{figure}

A \textbf{path}, $\gamma:\mathbb{I}\rightarrow X$ is a continuous map from the unit interval to a space $X$.
If $\gamma(0) = \gamma(1)$, the path is \textbf{closed}; otherwise it is \textbf{open}.
Paths $\gamma_0, \gamma_1$ to a domain $X$ are \textbf{homotopic relative to their endpoints} if there is a continuous function $F:\mathbb{I}\times\mathbb{I}\rightarrow X$ in which $F\big|_{\{0\}\times\mathbb{I}} = \gamma_0, F\big|_{\{1\}\times\mathbb{I}} = \gamma_1$, and $F\big|_{\mathbb{I}\times\{0\}} = \gamma_0(0), F\big|_{\mathbb{I}\times\{1\}} = \gamma_0(1).$ 
Here, $F$ is called a \textbf{homotopy} between the curves.
One such homotopy between two curves on a sphere is depicted in Figure \ref{fig:homotopic_sphere_curves}. 
Two paths for which such a homotopy does not exist are called  \textbf{inequivalent}.

For two paths $\gamma_1,\gamma_2:\mathbb{I}\rightarrow X$, define path \textbf{composition} $\gamma_1 \cdot \gamma_2:\mathbb{I}\rightarrow M$ by 
\begin{equation}
\left(\gamma_1 \cdot \gamma_2\right)(t) =
\begin{cases}
\gamma_1(2 t) & 0 \leq  t \leq \frac{1}{2} \\
\gamma_2(2 t-1) & \frac{1}{2} \leq  t \leq 1
\end{cases}
\end{equation}
The \textbf{fundamental group} of space $X$ based at $p\in X$ is the group generated by equivalence classes of closed paths with start and end at $p$ in the domain $X$ which are homotopically inequivalent, and is written $\pi_1(X,p)$.
The group operation is given by path composition.
The zero element of this group is the set of paths which are homotopically equivalent to the constant path (a single point), and the inverse operation is traveling the same path in opposite direction.
Figure \ref{fig:plate_with_hole_fundamental_group} depicts two curves: $\eta$ is homotopy equivalent to the constant path, where $\gamma$ is homotopically non-trivial.
For a path-connected space, $\pi_1(X,p)$ is equivalent to $\pi_1(X,q)$ for any $p,q \in X$, and so this group is often abbreviated as $\pi_1(X)$.

\begin{figure}
	\centering
	\includegraphics[trim = 0cm  0cm 0cm 0cm, clip, width=.75\textwidth]{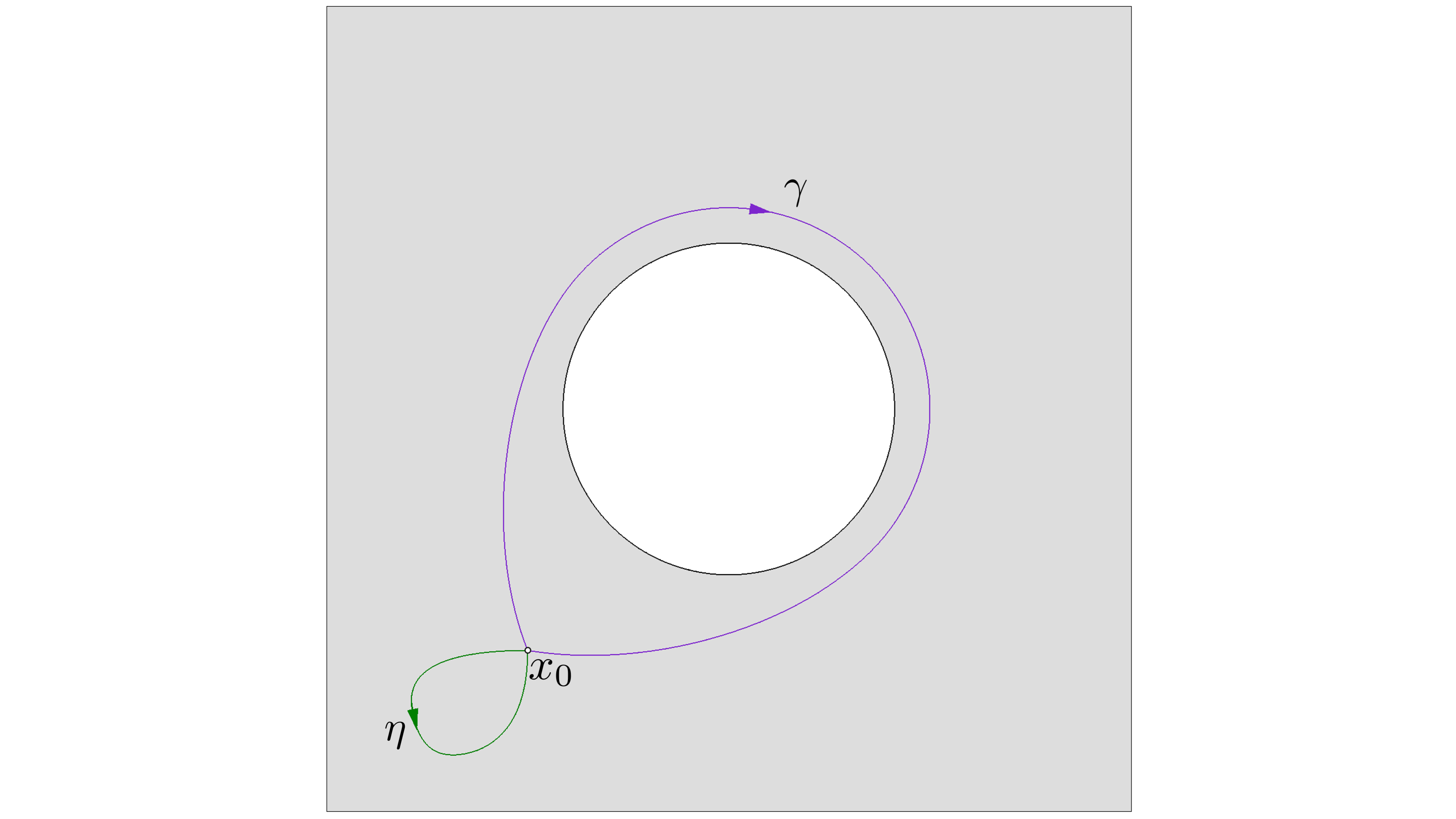}
	\caption{Two loops based at point $x_0$ are depicted in a plate with a hole. The loop $\gamma$ cannot be continuously deformed to a point without passing outside of the domain, and generates the fundamental group. The loop $\eta$ is homotopically equivalent to the constant path. }\label{fig:plate_with_hole_fundamental_group}
\end{figure}

 Let $f_0,f_1:X\rightarrow Y$  be two maps. $f_0$ is \textbf{homotopic} to $f_1$ if there is a map $F:X\times\mathbb{I}\rightarrow Y$ with  $F\big|_{X \times \{0\}} = f_0, F\big|_{X \times \{1\}} = f_1$, and is written $f_0 \simeq f_1$.
If $f:X\rightarrow Y$ and $g:Y\rightarrow X$ are maps with $g \circ f \simeq \mathrm{Id}_X$ and $f \circ g \simeq \mathrm{Id}_Y$, with $\mathrm{Id}_X$ being the identity function on $X$ (and similarly for $\mathrm{Id}_Y$), then the spaces $X$ and $Y$ are said to be \textbf{homotopy equivalent}.
A space which is homotopy equivalent to a single point is called \textbf{contractible}.
Figure \ref{fig:homotopy_contractible} pictorially displays how the space $\mathbb{R}^2-\{0\}$ is homotopy equivalent to the unit circle (which is not contractible) and how the L-shaped domain is contractible.

\begin{figure}
\begin{subfigure}{0.45\textwidth}
\centering
	\includegraphics[trim = 0cm  0cm 0cm 0cm, clip, width=1\textwidth]{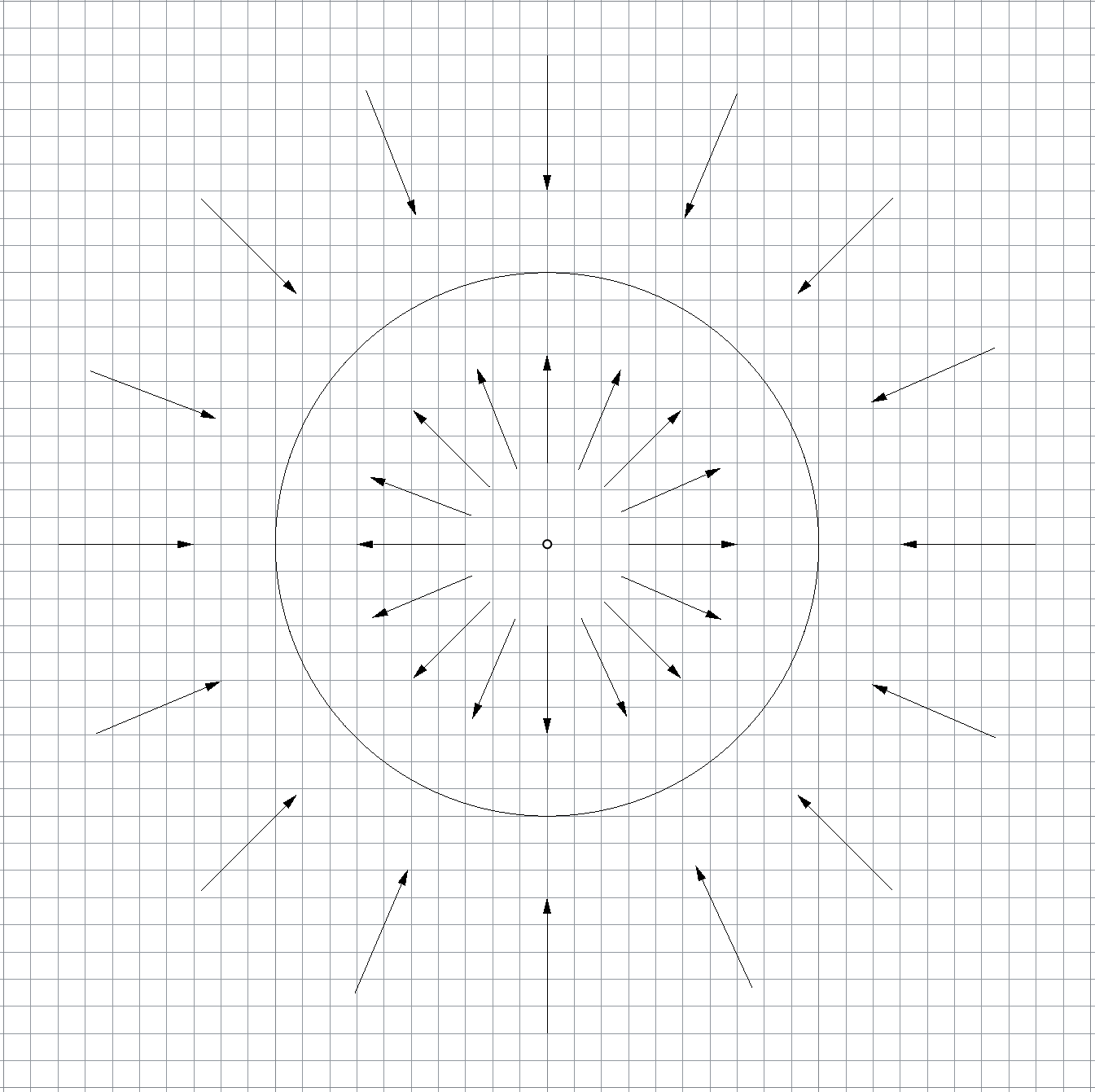}
	\caption{}
\end{subfigure}
\hspace{1em}
\begin{subfigure}{0.45\textwidth}
	\includegraphics[trim = 0cm  0cm 0cm 0cm, clip, width=1\textwidth]{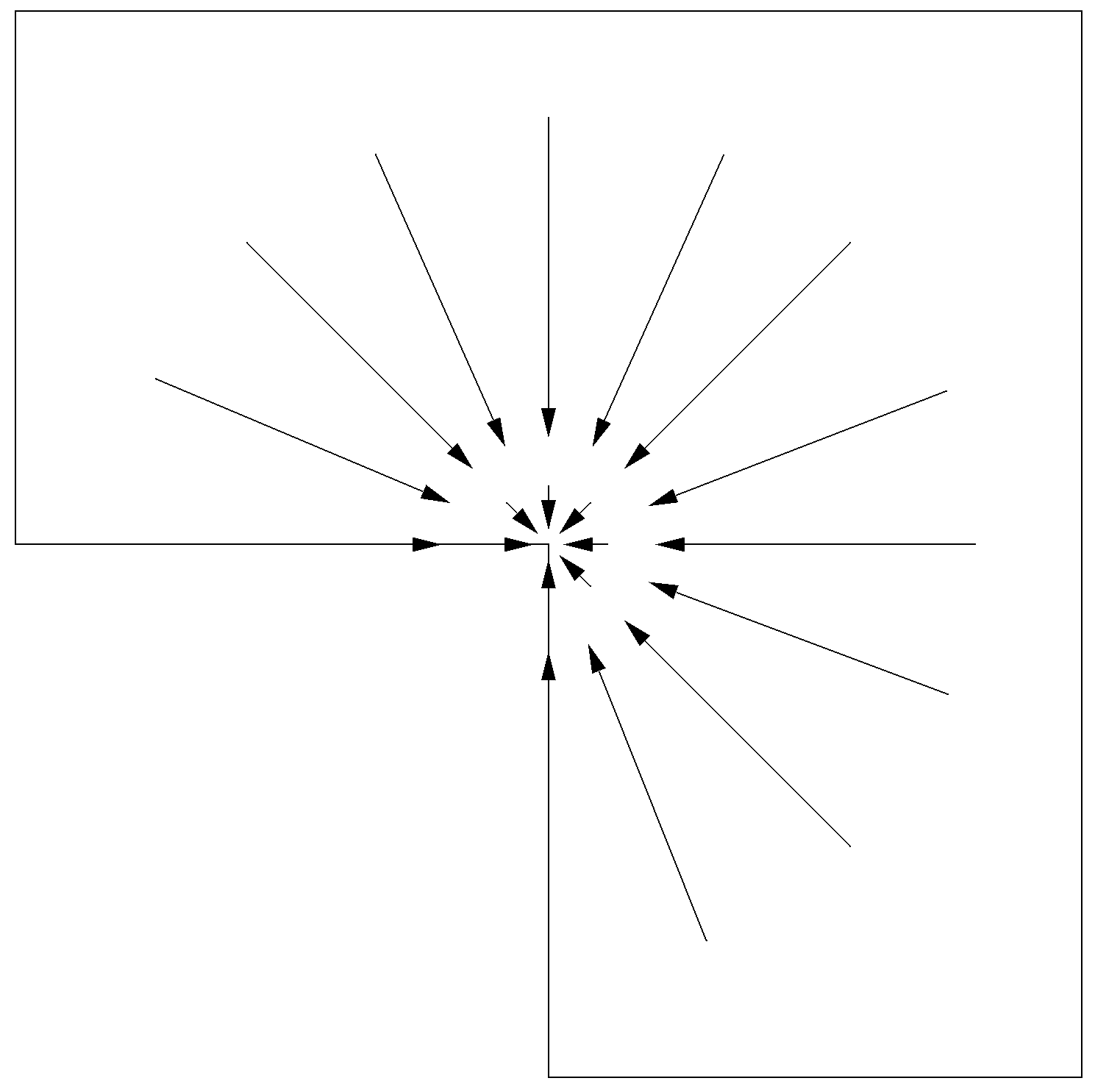}
	\caption{}
\end{subfigure}
\caption{The domain $\mathbb{R}^2-\{0\}$ is homotopy equivalent to the unit circle via the identity of the unit circle and the map $\frac{x}{||x||}$. Similarly, the L-shaped domain is contractible, via the identity on the reentrant point and the radial projection onto this point.}\label{fig:homotopy_contractible}
\end{figure}

For a connected open surface $\surf$, each boundary except one is homotopically inequivalent to concatenations of the other boundary components.
As a result, if $k$ denotes the number of boundary components in $\surf$, $\max\{0,k-1\}$ generators of $\pi_1(\surf{})$ are produced from boundary components.
The \textbf{genus} of a surface is half the number of generators of $\pi_1(\surf)$  that are not produced from a boundary component.
{\Bd (For a surface, the number of generators of $\pi_1(\surf)$ less the number of generators given by boundary components is always even.)}
Heuristically, it can be thought of as the number of ``holes'' in the object which are not boundary components.

An \textbf{orientable} surface is a surface into which a M\"obius band cannot be  injectively mapped.
Alternatively, one can think of a non-orientable surface as one that does not have a consistently defined normal, such as the M\"obius band in Figure \ref{fig:rectangle_quotient}.
All closed surfaces embedded in $\mathbb{R}^3$ are orientable.
The following is an important classical result of surface topology, originally the consequence of the combined theory from \cite{Brahana:1922, DehnHeegaard:1907, Rado:1925, Dyck:1888}, but which has since been presented more cohesively in works such as \cite{Thomassen:1992}.
\begin{theorem}[Classification of Surfaces:]\label{th:classification}
Every connected, compact, orientable surface is unique up to homeomorphism based on its genus $(\genus)$ and number of boundary components $(\boundarycomp)$.
\end{theorem}

Another useful object for surfaces is the \textbf{Euler characteristic}, which can be defined for connected, compact, orientable surfaces as
\[
	\chi(\surf) = 2 - 2g - k
\]
where $g$ is the genus of the surface and $k$ is its number of boundary components.
Some basic topologies of different genus and boundary component count are shown in Figure \ref{fig:surface_classification}.
Notice that the Euler characteristic can yield the same value for surfaces which are topologically distinct as per the Classification of Surfaces.

\begin{figure}
\captionsetup[subfigure]{justification=centering}
\centering
	\begin{subfigure}{0.26\textwidth}
	\includegraphics[trim = 0cm  0cm 0cm 0cm, clip, width=.95\textwidth]{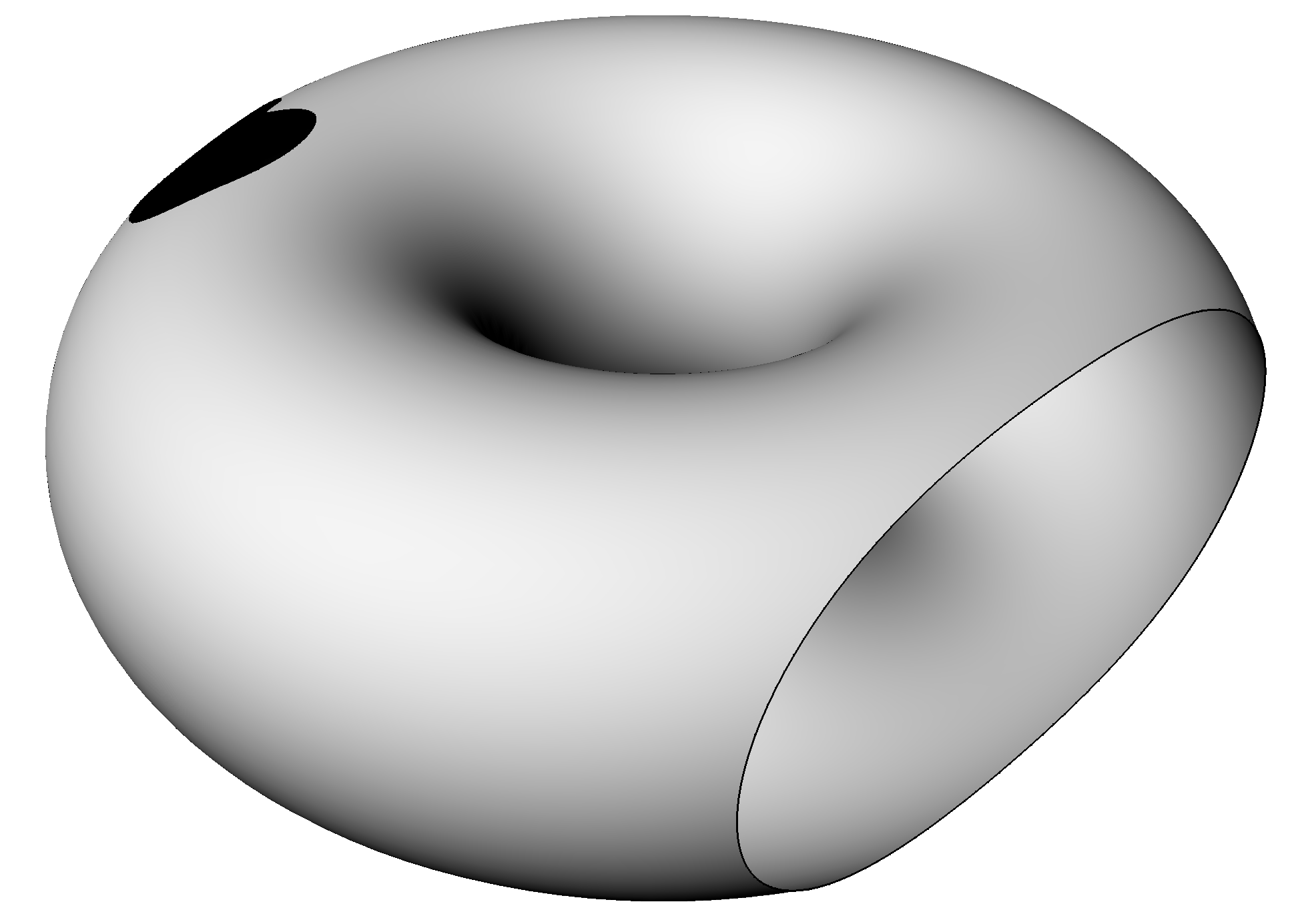}
	\caption{$g = 1, k = 2,$\\ $\chi(S) = -2$}	
	\end{subfigure}
\hspace{1em}
	\begin{subfigure}{0.235\textwidth}
	\includegraphics[trim = 0cm  0cm 0cm 0cm, clip, width=.95\textwidth]{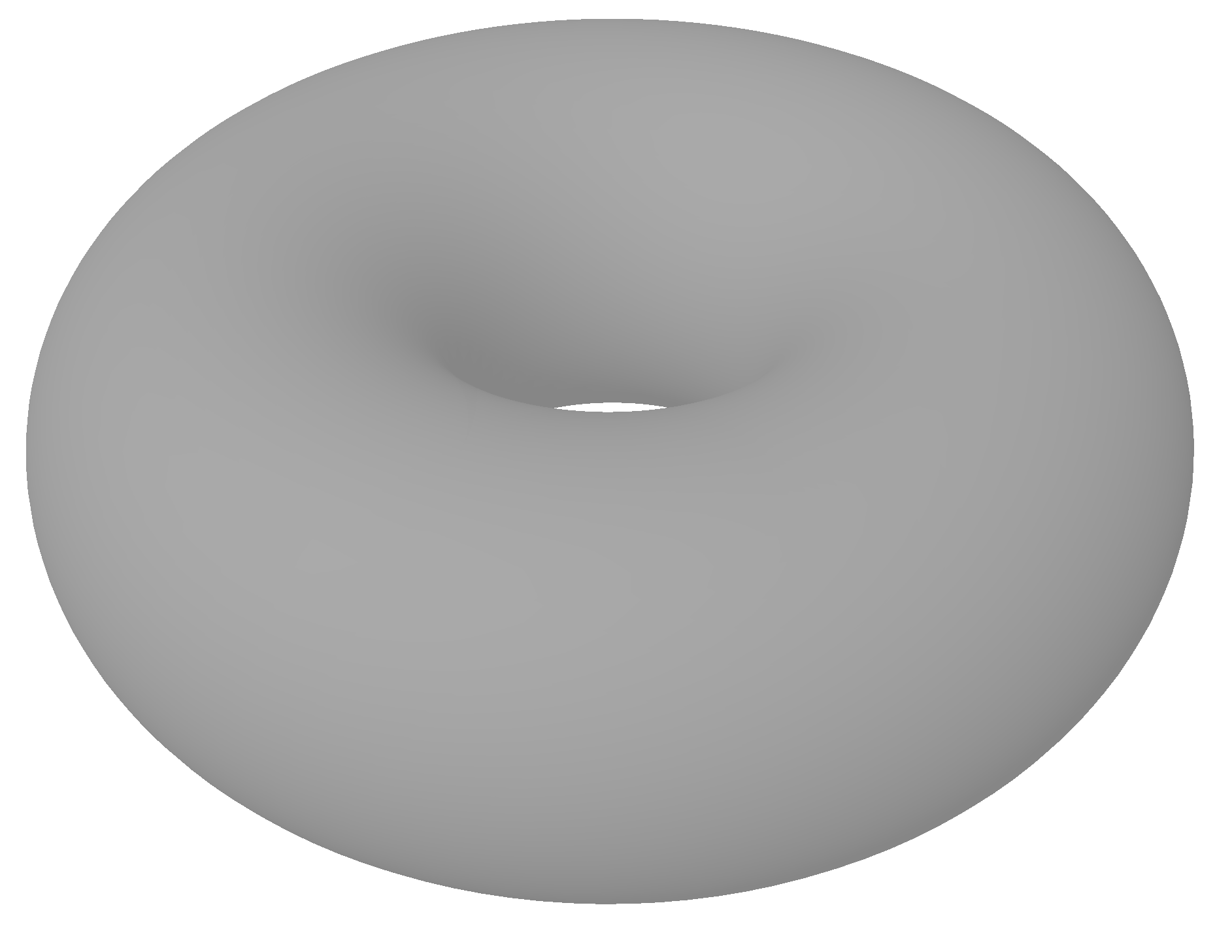}
	\caption{$g = 1, k = 0$\\ $\chi(S) = 0$}
	\end{subfigure}
\\
\hspace{1em}
	\begin{subfigure}{0.24\textwidth}
	\includegraphics[trim = 0cm  0cm 0cm 0cm, clip, width=.95\textwidth]{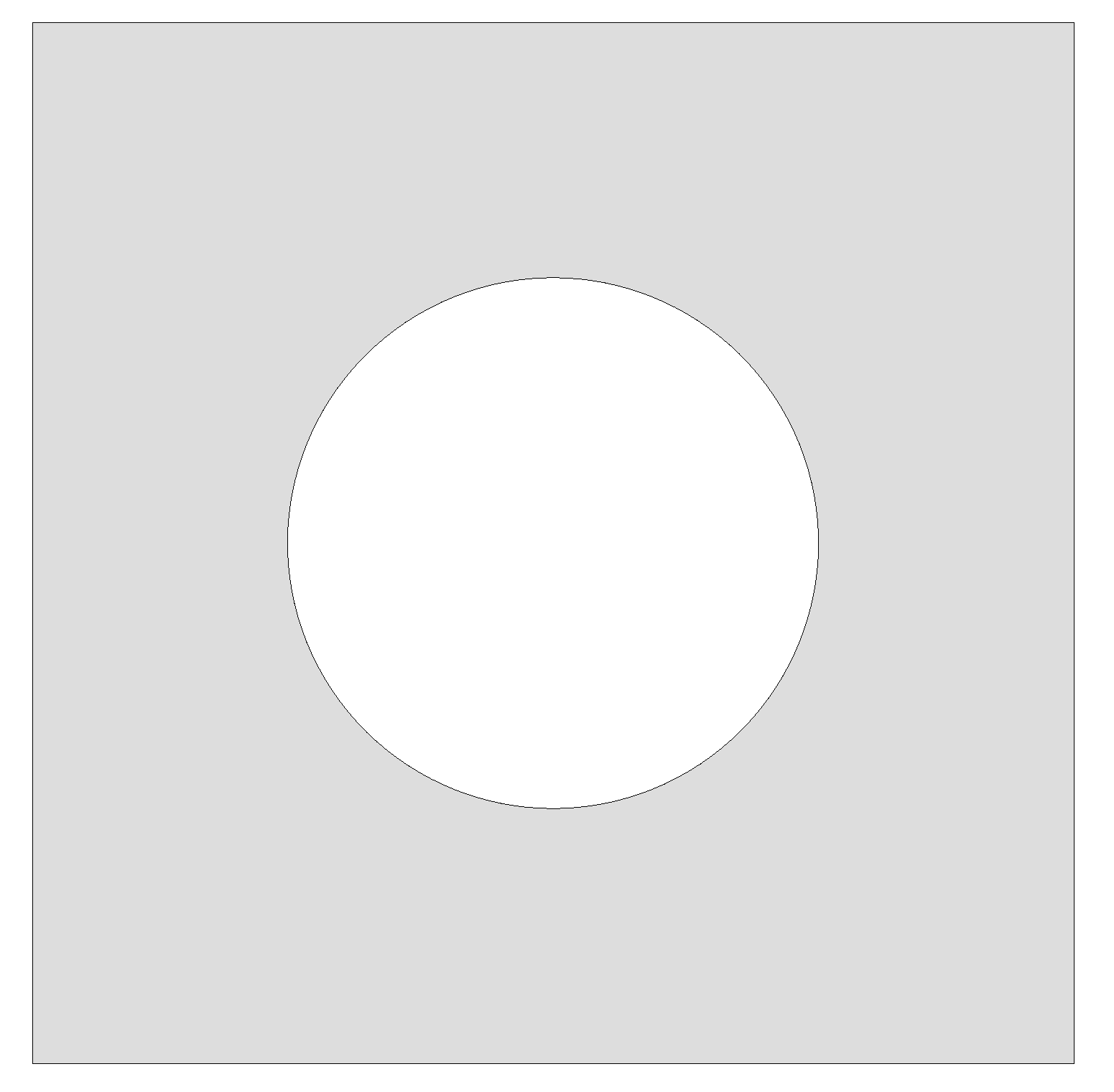}
	\caption{$g = 0, k = 2,$\\$ \chi(S) = 0$}
	\end{subfigure}
\hspace{1em}
	\begin{subfigure}{0.258\textwidth}
	\includegraphics[trim = 0cm  0cm 0cm 0cm, clip, width=.95\textwidth]{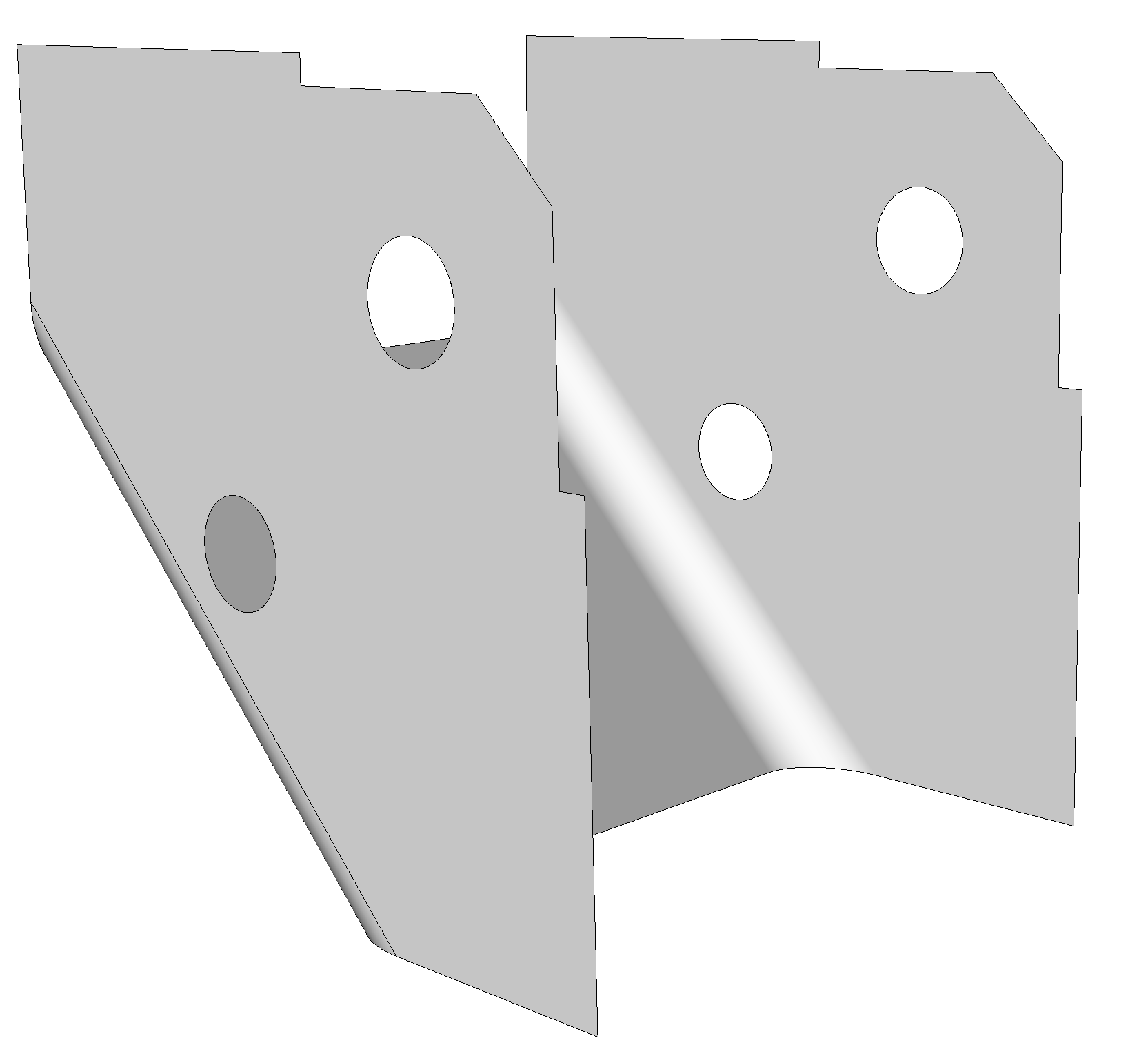}
	\caption{$g = 0, k = 5,$\\$ \chi(S) = -3$}
	\end{subfigure}
\caption{{\Bd Surfaces with different genus ($\genus$) and number of boundary components ($k$)} are depicted. Note that the Euler characteristic may be the same for topologically distinct objects.}\label{fig:surface_classification}
\end{figure}

\ifthenelse{\boolean{isELS}}
{
\subsubsection{Extensions of Basic Topology}
}
{
\subsection{Extensions of Basic Topology}
}

The only assumption made on maps in the previous section was that they are continuous. 
Often this representation is too general.
Smooth, conformal, and piecewise-linear topologies deal with maps that are additionally assumed to be differentiable, analytic (in the sense of complex analysis), and piecewise-linear, respectively.
With these additional assumptions come additional structure.
Of particular interest are constructs given in smooth topology.

A \textbf{smooth surface} $\surf$ is a surface in which every point $p \in \surf$ has a neighborhood $U(p) \subset \surf$ and an accompanying map $\phi_p:U(p) \rightarrow \mathbb{R}^2$ with the following structure.
\begin{enumerate}
	\item $U(p)$ is homeomorphic to $\phi_p\big(U(p)\big)$, which is homeomorphic  to an open ball in $\mathbb{R}^2$ if $p$ is in the interior of $\surf$, and it is homeomorphic to an open half-ball in $\mathbb{R}^2$ if it is on the boundary.
	\item If $p,q \in \surf$ with $V = U(p) \cap U(q)$, then $\phi_p \circ \phi_q^{-1} $ and $\phi_q \circ \phi_p^{-1}$ are differentiable (typically at least twice differentiable).
\end{enumerate}
A set $\big(\phi_{p_\iota},U_\iota\big)$ is called a \textbf{chart}, and a set of charts covering the surface is called an \textbf{atlas}.

A function, $f:M\rightarrow N$, between $k$-dimensional manifold $M$ and $\ell$-dimensional manifold $N$ is \textbf{differentiable} at $p\in M$ if for an atlas on $M,$ $(\phi_\iota:M\supset U_\iota \rightarrow \phi_\iota(U_\iota)\subset \mathbb{R}^k,U_\iota),$ and and atlas on $N$, $(\varphi_{\tilde{\iota}}:N\supset V_{\tilde{\iota}}\rightarrow \varphi_{\tilde{\iota}}(V_{\tilde{\iota}})\subset\mathbb{R}^\ell,V_{\tilde{\iota}}), \varphi_{\tilde{\iota}} \circ f \circ \phi_\iota^{-1}$ is differentiable at $\phi_\iota(p)$.
A smooth \textbf{immersion} is a differentiable function with a \textbf{Jacobian} (Frech\'et derivative) that is globally of full rank.
Under this representation, it has a well-defined inverse locally (by the Inverse Function Theorem).
An \textbf{embedding} is a bijective immersion.

Let $M$ denote an $n$-dimensional differentiable manifold, $p$ a point in $M$ and let $(\varphi:M\supset U\rightarrow \varphi(U)\subset\mathbb{R}^n,U)$ be a chart with $p \in U$. 
Let two curves $\gamma_1,\gamma_2:(-\epsilon,\epsilon)\rightarrow M$ be such that $\gamma_i(0) = p.$ 
The curves $\gamma_1$ and $\gamma_2$ are {\Bd defined to be} equivalent if and only if $\frac{d}{dt}(\varphi \circ \gamma_1) = \frac{d}{dt}(\varphi \circ \gamma_2),$ where $t \in (-\epsilon,\epsilon)$. 
The equivalence class of curve $\gamma$ is denoted $\gamma'(0)$. 
The \textbf{tangent space} of $M$ at $p$, denoted $T_pM$, is the set of all equivalence classes of curves passing through $p$, as seen in Figure \ref{fig:tangent_space}.
{\Bd In this sense, $\gamma^\prime(0)$ can be interpreted as a tangent vector at point $0$ on the curve $\gamma$.}
Two curves $\gamma_1,\gamma_2:(-\epsilon,\epsilon)\rightarrow \surf$ on a surface $\surf$ in which $\gamma_1(0) = \gamma_2(0) = p$ are \textbf{transversal} at $p$, written $\gamma_1 \transverse_p \gamma_2$, if  $\frac{d}{dt}(\varphi \circ \gamma_1) \neq \alpha \frac{d}{dt}(\varphi \circ \gamma_2)$ for any $\alpha \in \mathbb{R}$.

\begin{figure}
\centering
\includegraphics[trim = 0cm  0cm 0cm 0cm, clip, width=.75\textwidth]{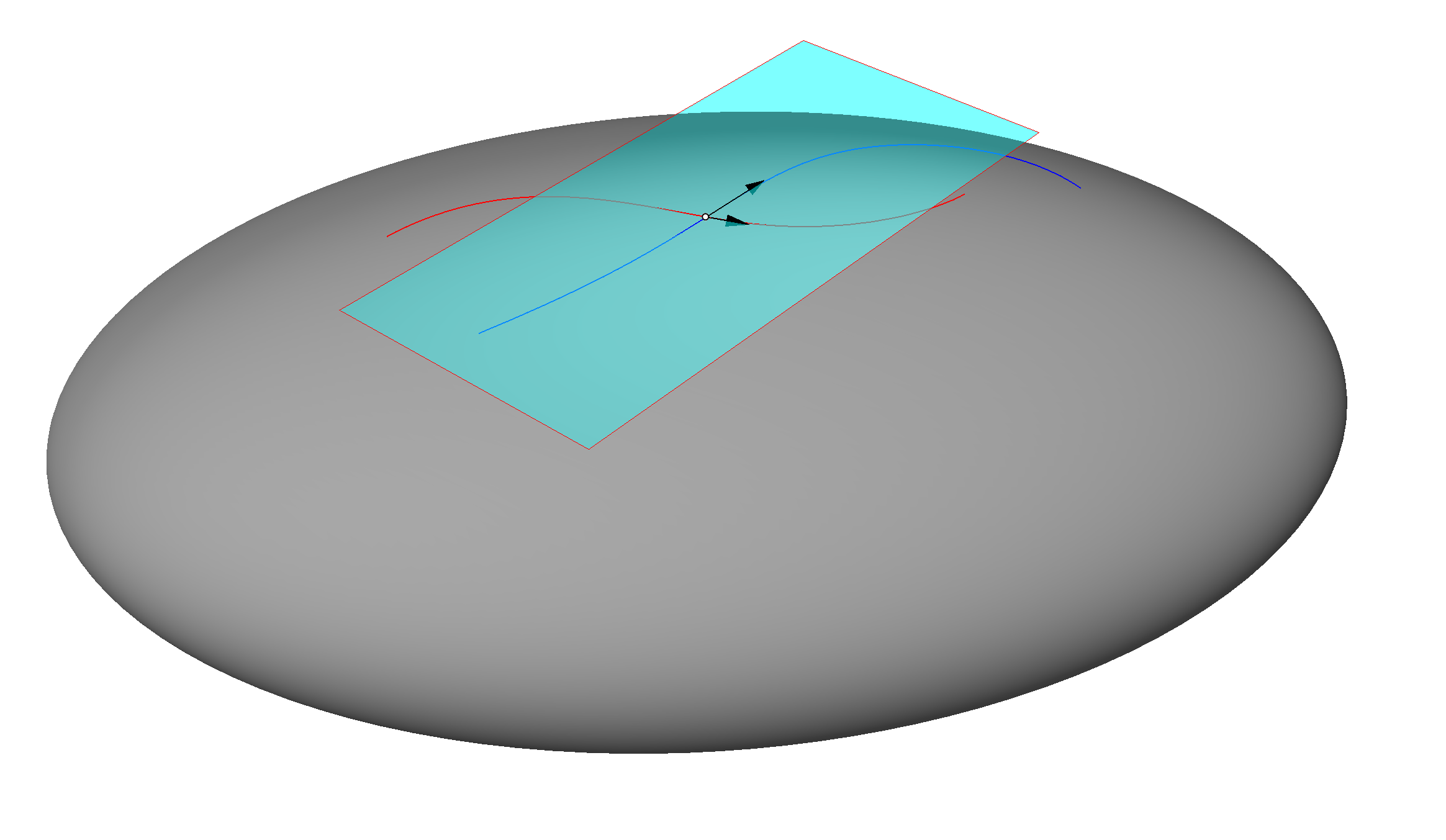}
\caption{The tangent {\Bd space} of a surface at a point is defined by the derivatives of smooth curves on the surface passing through the point. Here, two smooth, transversal curves have derivative vectors which form a basis for the tangent space.}\label{fig:tangent_space}
\end{figure}

Let $\varphi \circ \gamma(0) =: x \in \mathbb{R}^n$, and let $\eta_x := \frac{d}{dt}(\varphi \circ \gamma(t))|_{t=0}$ be the vector at $x$ of the derivative of the function composition. 
Then the $\textbf{push-forward}$ of $\eta_x$ onto $T_pM$ via $\varphi^{-1},$ denoted $(\varphi_*^{-1})_x(\eta_x):T_x\mathbb{R}^n\rightarrow T_pM,$ is given by $(\varphi_*^{-1})_x(\eta_x) = [D\varphi^{-1}(x)](\eta)_{\varphi^{-1}(x)}$, where $D\varphi^{-1}(x)$ is the Frech\'et derivative of $\varphi^{-1}$ at $x \in \mathbb{R}^n$ and the subscript $\varphi^{-1}(x) = \gamma(0) = p$ denotes that the vector is in the tangent space $T_pM$. 
Using the vector space structure in $\mathbb{R}^n$ and noting that the Frech\'et derivative is linear, $T_pM$ also inherits a vector space structure. 
More generally, if $f:M\rightarrow N$ is a differentiable function between $M$ and $N$, and $\gamma:(-\epsilon,\epsilon)\rightarrow M$ a curve with $\gamma(0) = p \in M$, then $(f_*)_{p}\big(\gamma'(0)\big) := (f \circ \gamma)'(0)$ is the \textbf{push-forward} of $\gamma'(0) \in T_pM$ to $(f_*)_{p}\big(\gamma'(0)\big) \in T_{f(p)}N$. 
Figure \ref{fig:push_forward} shows the push forward of a vector onto a manifold.

\begin{figure}
\centering
\includegraphics[trim = 0cm  0cm 0cm 0cm, clip, width=.95\textwidth]{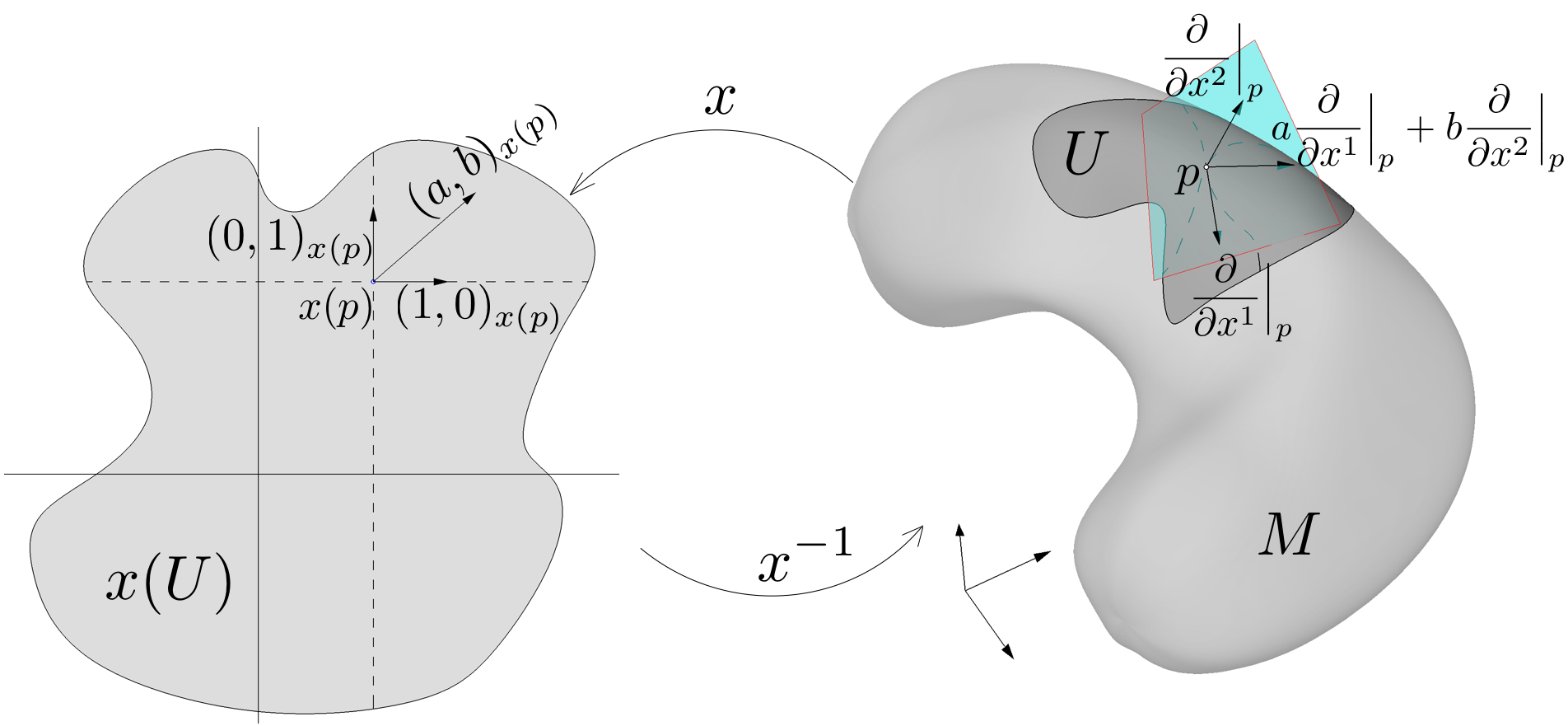}
\caption{The vector $(a,b)_x(p)$ in the tangent space at $p$ in $\mathbb{R}^2$ is pushed forward via the map $x^{-1}$ onto the manifold $M$, yielding the vector $a \frac{\partial}{\partial x^1}\Big|_p + b \frac{\partial}{\partial x^2}\Big|_p$.}\label{fig:push_forward}
\end{figure}

A continuous (or smooth) structure uniting each separate tangent space comes in the form of a tangent bundle. The \textbf{tangent bundle} of $M,$ written $TM,$ is given by a continuous projection from the disjoint union of each tangent space onto the manifold $\pi:\bigsqcup_{p \in M} T_pM \rightarrow M$ such that $\pi(\eta_p) = p$ for any $\eta_p \in T_pM$, such that addition and scalar multiplication between members in each individual tangent space  {\Bd(in the typical vector space manner)} is well-defined, and such that for every $p \in M,$ there is a neighborhood $U\ni p$ and a homeomorphism ${\Bd\hat{t}}:\bigsqcup_{{\Bd{q}} \in U}T_{\Bd q}M \rightarrow U\times\mathbb{R}^n$ which is also a vector space isomorphism from every $T_{\Bd{q}}M$ to ${\Bd{q}}\times\mathbb{R}^n$ for each ${\Bd{q}} \in U$ (called a \textbf{local trivialization}). 
If $f:M\rightarrow N$ is a differentiable map, then the \textbf{push-forward} map, $f_*:TM\rightarrow TN$ is the map defined by the union of each $(f_*)_p:T_pM\rightarrow T_{f(p)}N$ defined as in the previous paragraph.

A \textbf{vector field} is a continuous map $\eta:M\rightarrow \bigsqcup_{p \in M} T_pM$ such that $\pi(\eta(p)) = p$ for each $p \in M$, with $\pi$ the usual projection map from the tangent bundle to the base space. More generally, a \textbf{section} of a bundle is a continuous map from the manifold to the bundle---a vector field is simply a section of the tangent bundle. 
A \textbf{singularity} of a vector field is an isolated point at which the vector field is zero.
The Poincar\'e-Hopf Theorem implies that every vector field on a closed surface must have singularities if it is not a topological torus.

Accompanying each vector space, $V,$ is a dual space $V^*$ for which members are linear functionals on members of $V$ (i.e. $V^* \ni \alpha:V \rightarrow \mathbb{R}$). Using this construction, we write the dual to the tangent space $T_pM$ as $T_p^*M$. Following a similar construct as for the tangent bundle, the \textbf{dual bundle} is written as $T^*M$, and is given by the objects $(\bigsqcup_{p \in M} T_p^*M, \pi', M)$ in conjunction with vector addition and scalar multiplication on each $T_p^*M$. Here, $\pi':\bigsqcup_{p\in M} T_p^*M\rightarrow M$ acts as $\pi$, and ${\Bd\hat{t}}':\bigsqcup_{{\Bd{q}} \in U}T_{\Bd{q}}^*M \rightarrow U\times(\mathbb{R}^n)^*$ is a homeomorphism, and some isomorphism $w':(\mathbb{R}^n)^*\rightarrow \mathbb{R}^n$ is canonically chosen so that $(\text{Id}\times w') \circ t':U\times\mathbb{R}^n$ is a local trivialization. 
 A \textbf{covariant tensor} of order $m$ is a multilinear map $TM_1\times\dots\times TM_m=:(TM)^m\rightarrow \mathbb{R}$ defined on a tensor bundle.
Here, the tensor bundle is defined in a manner analogous to the definition of the covector bundle.
The space of covariant tensors is written as $\mathcal{T}^m(M)$.
Similarly, a \textbf{contravariant tensor} of order $m$ is a multilinear  map  map $T^{*}M_1\times\dots\times T^*M_m=:(T^*M)^m\rightarrow \mathbb{R}$ with space of contravariant tensors written $\mathcal{T}_m(M),$ and a \textbf{mixed tensor} a multilinear map $(TM)^k\times(T^*M)^\ell\rightarrow\mathbb{R},$ whose space is written $\mathcal{T}^k_\ell(M).$ 

Because the push-forward is a well-defined map, if $f:M\rightarrow N$ is differentiable with $(f_*)_p:T_pM\rightarrow T_{f(p)}N$ each a linear transformation, the dual on each tangent space may be defined and denoted as $(f^*)_p:T_{f(p)}^*N \rightarrow T_p^*M.$ If $\omega:N\rightarrow \bigsqcup_{q \in N} T^*_qN$ is a section of the cotangent bundle $T^*N$, then the \textbf{pull-back} of $\omega$ to $T^*M,$ denoted $f^*(\omega)$, is defined point-wise as $(f^*)_p\Big(\omega\big(f(p)\big)\Big)\circ (f_*)_p$, and is a section of $T^*M$. Similarly, for $A \in \mathcal{T}^k(N)$ is a covariant $k$-tensor, the \textbf{pull-back} $f^*(A) \in \mathcal{T}^k(M)$ is defined point-wise by $\big[\big(f^*(A)\big)(p)\big](X_1(p),\dots,X_k(p)\big) = A\big(f(p)\big)\Big((f_*)_p\big(X_1(p)\big),\cdots,(f_*)_p\big(X_k(p)\big)\Big),$ with $\{X_j\}_{j=1}^k$ sections of $TM$.

Finally, given two vector fields $X$ and $Y$, the \textbf{Lie bracket} $[X,Y]$ is the derivative of $Y$ along the vector field $X$.
If $[X,Y]$ is zero, then locally there is a well-defined coordinate system on the surface $\surf$ defined by the integral curves of $X$ and $Y$ \cite[pp.~158]{Spivak:1999v1}

\ifthenelse{\boolean{isELS}}
{
\subsubsection{Differential Geometry}
}
{
\subsection{Differential Geometry}
}

\begin{figure}
\centering
\begin{subfigure}{0.42\textwidth}
\includegraphics[trim = 0cm  0cm 0cm 0cm, clip, width=\textwidth]{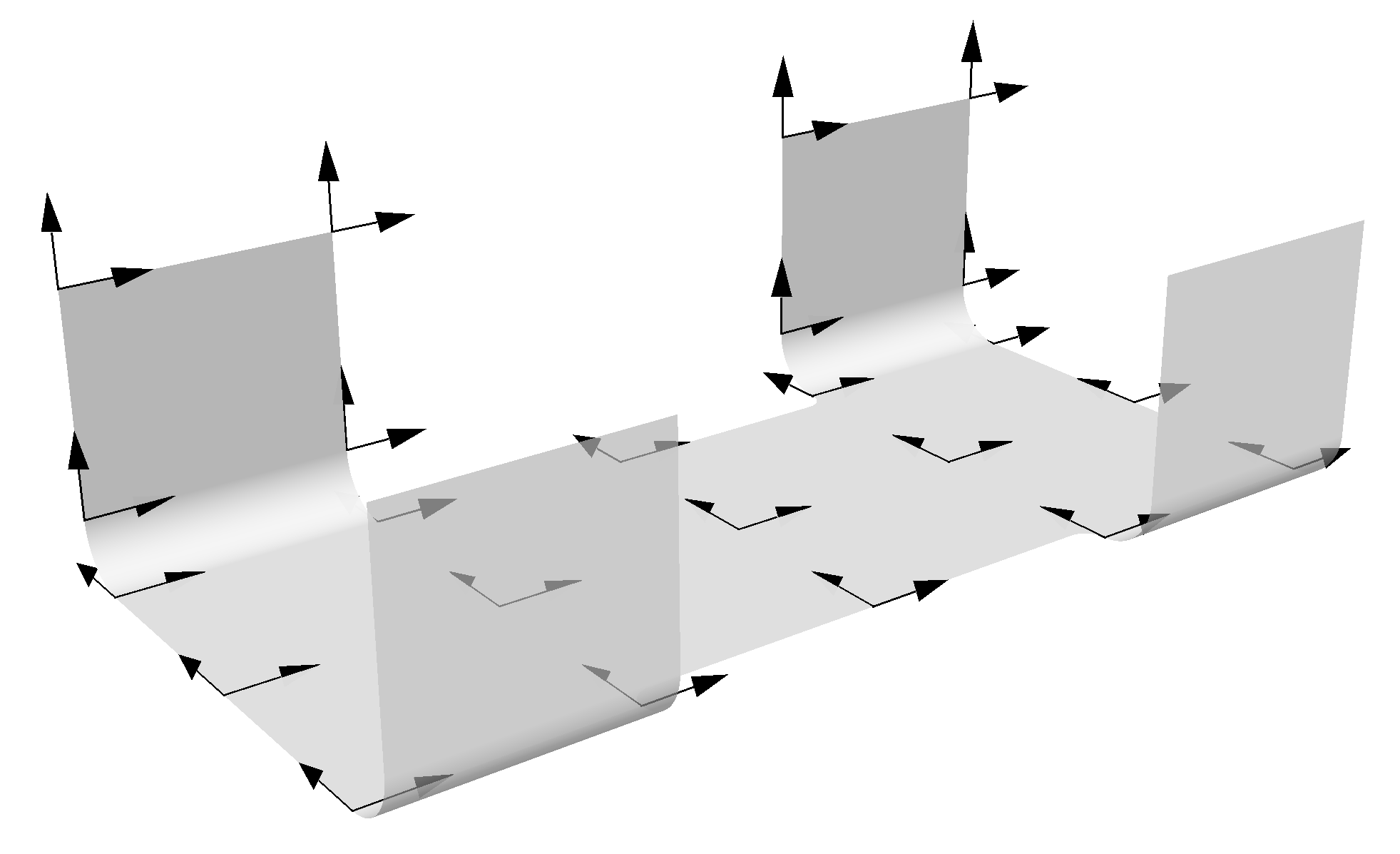}
\caption{Riemannian Metric Tensor}
\end{subfigure}
\hspace{1em}
\begin{subfigure}{0.455\textwidth}
\includegraphics[trim = 0cm  0cm 0cm 0cm, clip, width=\textwidth]{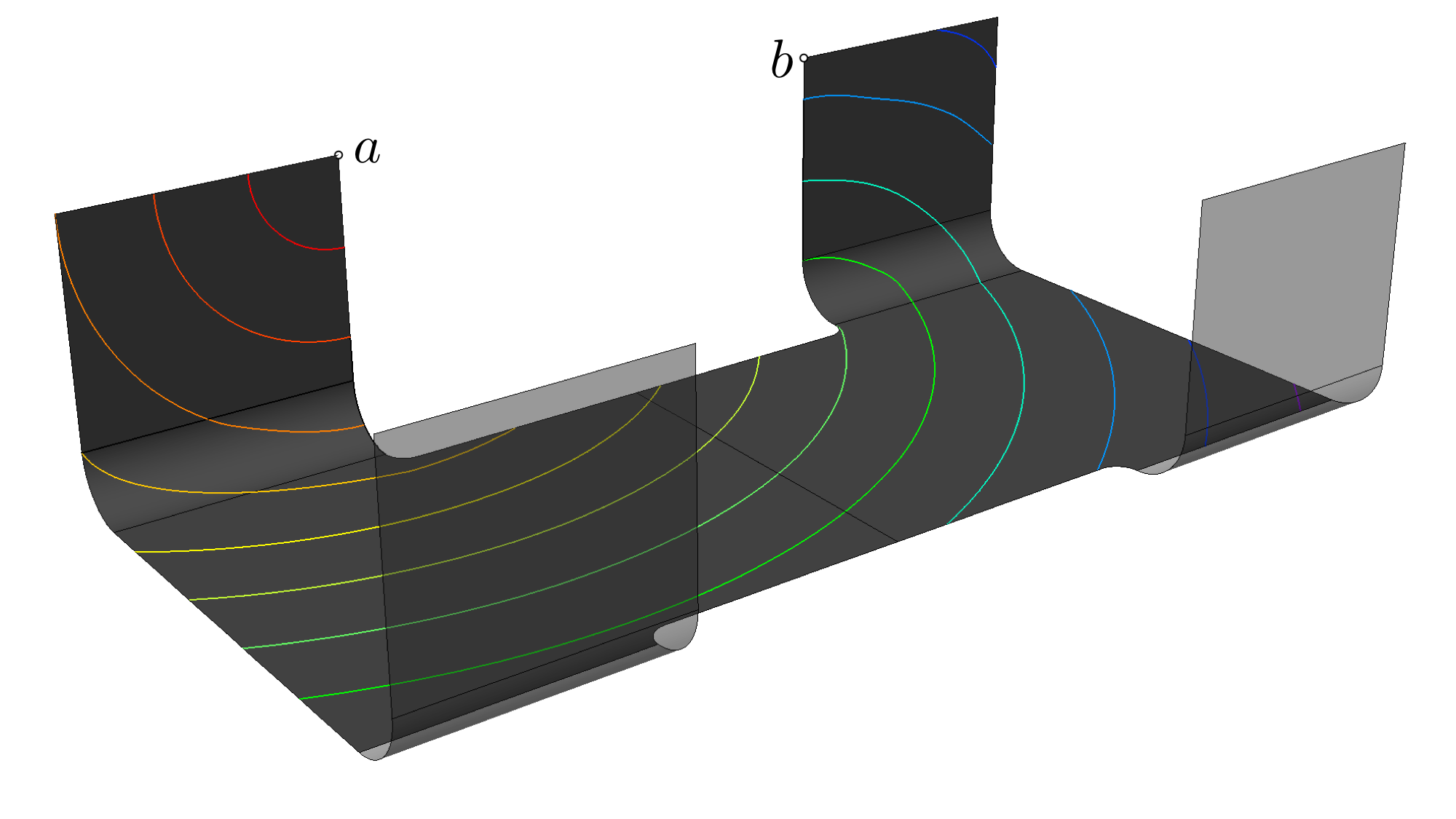}
\caption{Induced Surface Metric}
\end{subfigure}
\caption{In (a), {\Bd a set of unit-length orthogonal vectors under the Reimannian metric induced by $\mathbb{R}^3$ are depicted on the surface}. In (b), the {\Bd distance metric on the surface induced by the Riemannian metric of (a)} is depicted by contours of increasing distance from point $a$.}\label{fig:riemannian_metric}
\end{figure}

Of primary interest in geometry (and really, the defining object of ``geometry'') is the \textbf{metric tensor}: a symmetric, positive definite member of $\mathcal{T}^2(M).$ It is frequently denoted as $g$ or $\langle \cdot,\cdot \rangle$. 
The metric tensor generalizes the idea of an inner product onto the manifold, yielding an inner-product on each tangent space $T_pM$ of the manifold. 

For a given Riemannian metric $\emptymetric$ on an arbitrary path-connected manifold $M$, define the \textbf{length} of a piecewise-smooth curve $\gamma:\mathbb{I}\rightarrow M$ by
\begin{equation}\label{eq:curve_length}
	L(\gamma) = \int_0^1 \Big(\metric{\gamma^{\prime}(t)}{\gamma^\prime(t)}\Big)^{\frac{1}{2}}dt.
\end{equation}
The Riemannian metric induces a {\Bd distance} metric on the manifold, $\distancemetric:M\times M \rightarrow \mathbb{R}$.
For points $p,q \in M$ and the set of all curves $\{\gamma_\iota\}_\iota$ in which $\gamma_\iota(0) = p, \gamma_\iota(1) = q$, the distance is defined by {\Bd
\begin{equation}\label{eq:distance_metric}
	\distancemetric(p,q) = \inf_{\gamma \in \{\gamma_\iota\}_\iota} L(\gamma).
\end{equation}
}
{\Bd Despite their similar names, a Riemannian metric tensor and its induced distance metric on the surface are very different objects, as} shown in Figure \ref{fig:riemannian_metric}.

When $i:S \rightarrow\mathbb{R}^n$ is an immersion of surface $S$ in Euclidean space, the Euclidean metric tensor of $\mathbb{R}^n, \langle \cdot,\cdot\rangle_\mathbb{E}$ induces a metric tensor on $S$ via the pull-back of the immersion, $\mathrm{I} :=  i^*  \langle \cdot,\cdot\rangle_\mathbb{E}$---this is called the \textbf{first fundamental form}. 

\begin{figure}
\centering
\includegraphics[trim = 0cm  0cm 0cm 0cm, clip, width=.95\textwidth]{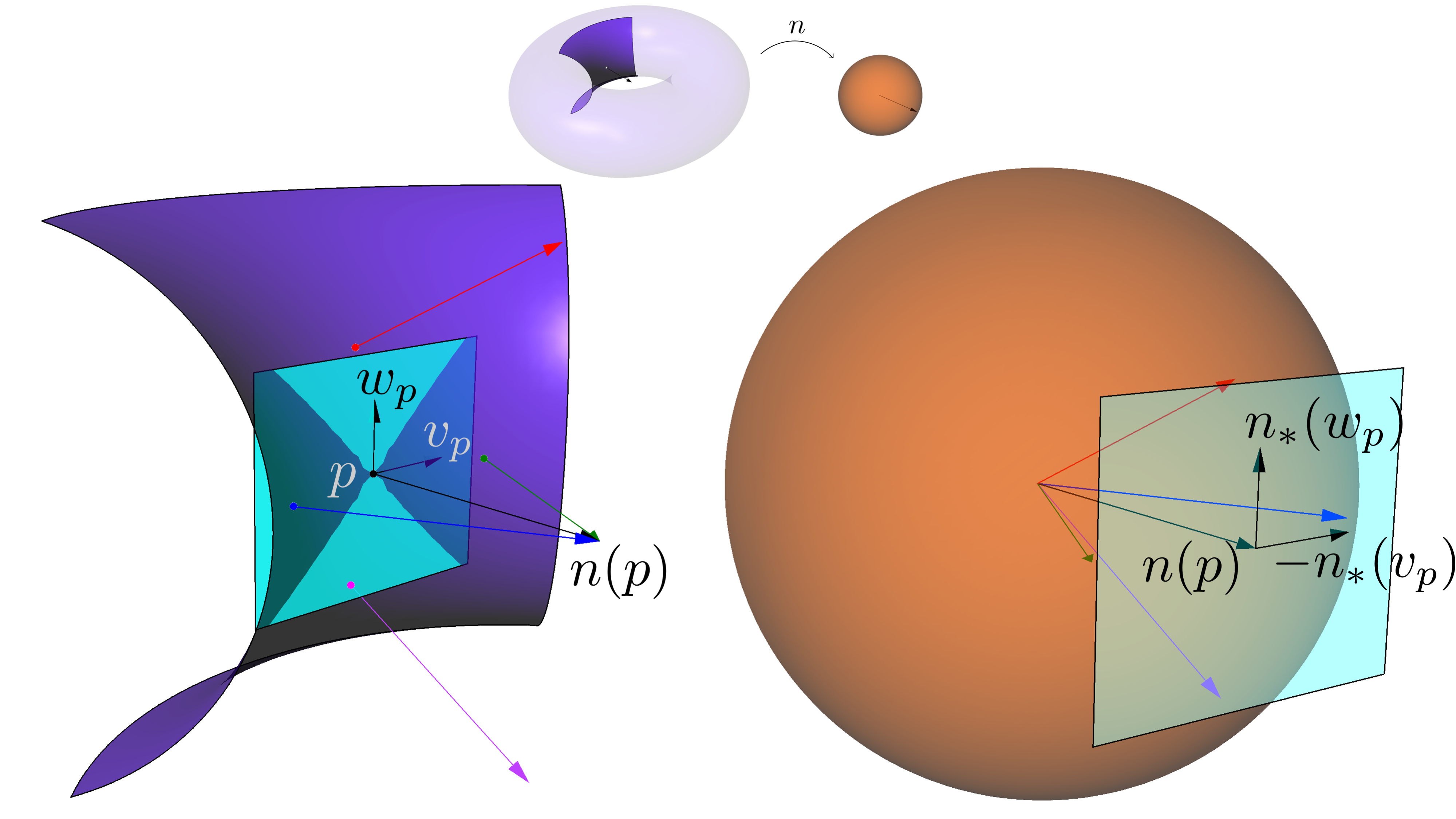}
\caption{The Gauss map is depicted on a torus. Above, a torus is mapped to the unit sphere via it's normal. A particular subsection is highlighted, which is enlarged in the bottom left. Vectors in various colors depict normal vectors, which are mapped to their respective points in the sphere on the right. The pushforward of vectors $v_p$ and $w_p$ in the tangent space of $p$ (used to define the Weingarten map) is also depicted.}\label{fig:gauss_map}
\end{figure}

An  orientable, differentiable surface $\surf$ embedded in $\mathbb{R}^3$ will have a well-defined normal, $n_p$ at each point $p \in S,$ which is continuous with normals in its neighborhood. The \textbf{Gauss map} for the surface is given by $n:\surf\rightarrow\mathbb{S}^2 \subset \mathbb{R}^3$ with $p\mapsto n_p$, as depicted in Figure \ref{fig:gauss_map}. 
Let $v_p \in T_pS,$ so $n_*(v_p) \in T_{n(p)}\mathbb{S}^2$. 
However, because we are in $\mathbb{R}^3$ and $n_p$ is the unit normal at both $p \in S$ and $n(p) \in \mathbb{S}^2$, their tangent planes may be identified via $\sim:T_{n(p)}\mathbb{S}^2\rightarrow T_{p}S, v_{n(p)} \mapsto v_p,$ where both are represented in a common basis for $\mathbb{R}^3$. Then the \textbf{Weingarten} map $w_n:TS\rightarrow TS$ is defined by $w_n(v_p) = \sim\big(n_*(v_p)\big)$. The \textbf{second fundamental form} $\rm{II} \in \mathcal{T}^2(S)$ is defined by
\begin{equation}
\mathrm{II}(p)(u_p,v_p) := -\mathrm{I}(p)\big(w_n(u_p),v_p\big) = -i^*\langle w_n(u_p),v_p\rangle_\mathbb{E}
\end{equation}
    If $\gamma:(-\epsilon,\epsilon)\rightarrow i(S)$ is a curve parameterized by arclength and $c(0) = p \in i(S), c'(0) = X \in T_p(S),$ then $\mathrm{II}(p)(X,X) = \langle c''(0),n(p)\rangle_{\mathbb{E}}$---the second fundamental form is the signed curvature at $p$ of the curve given by $U_p \cap N$, where $U_p \subset S$ is a neighborhood of $p$ and $N\subset \mathbb{R}^3$ is the plane intersecting $X$ and $ n(p)$ (see \cite[pp.~123]{Spivak:1999v2}). 
    Furthermore, the second fundamental form is symmetric, so for an orthonormal basis $X_p,Y_p$ of $T_pS$, the matrix representation of $\mathrm{II}(p)\big(X_p,Y_p\big)$ is symmetric. 
    As a result, eigenvalues $\kappa_1(p),\kappa_2(p)$ of of this matrix (called the \textbf{principal curvatures}) can be extracted, describing the maximal and minimal signed curvatures of curves in normal planes at $p,$ with orthogonal eigenvectors describing the directions. 
    The value of the function $\kappa:S\rightarrow \mathbb{R}, p\rightarrow \kappa_1(p)\cdot\kappa_2(p)$ is called the \textbf{Gaussian curvature} of $S$ at $p$. Depictions of negative and positive Gaussian curvature are represented in Figure \ref{fig:gaussian_curvature}.

\begin{figure}
\centering
\begin{subfigure}{0.41\textwidth}
\includegraphics[trim = 0cm  0cm 0cm 0cm, clip, width=\textwidth]{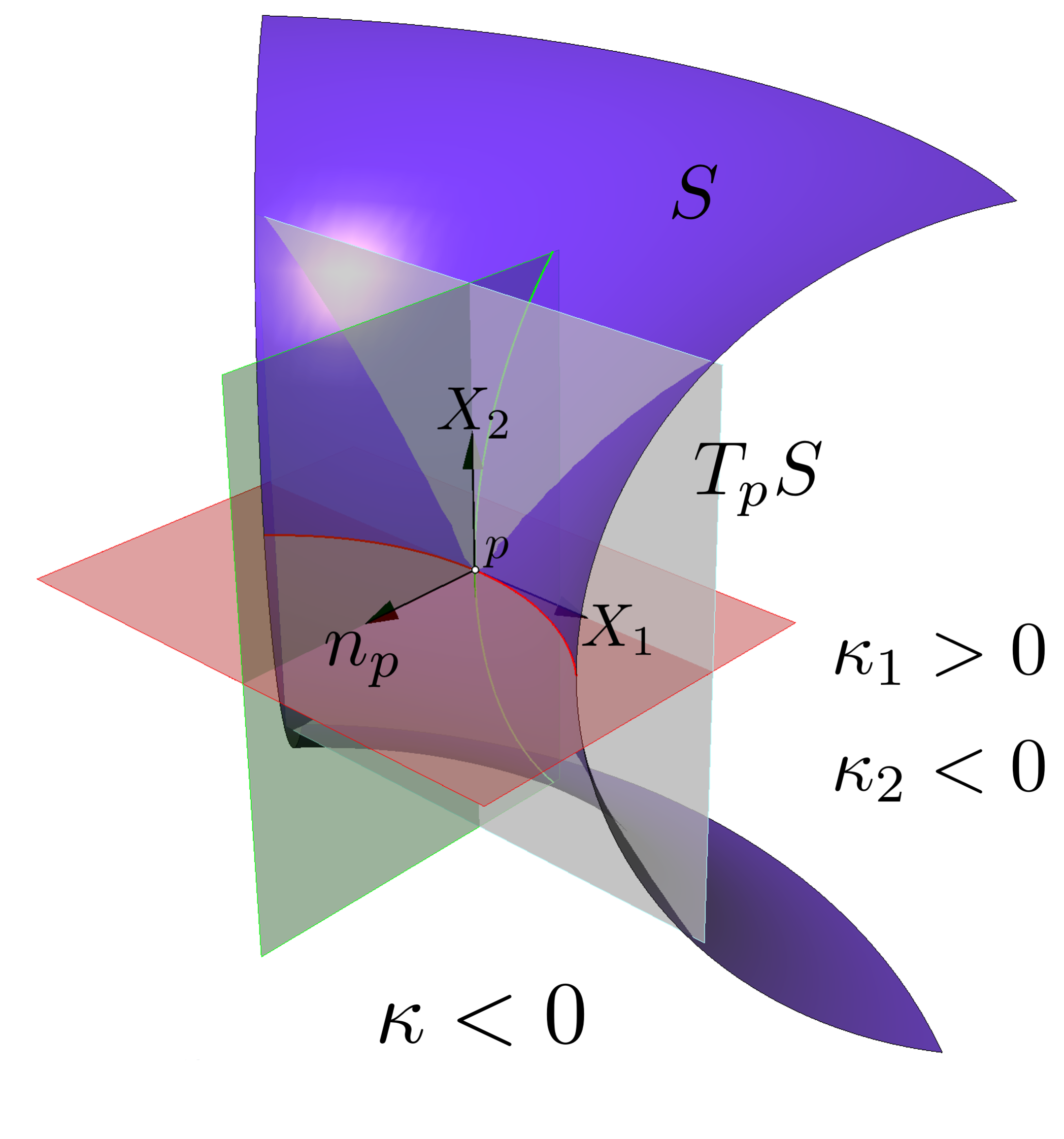}
\caption{Negative Gaussian Curvature}
\end{subfigure}
\begin{subfigure}{0.505\textwidth}
\includegraphics[trim = 0cm  0cm 0cm 0cm, clip, width=\textwidth]{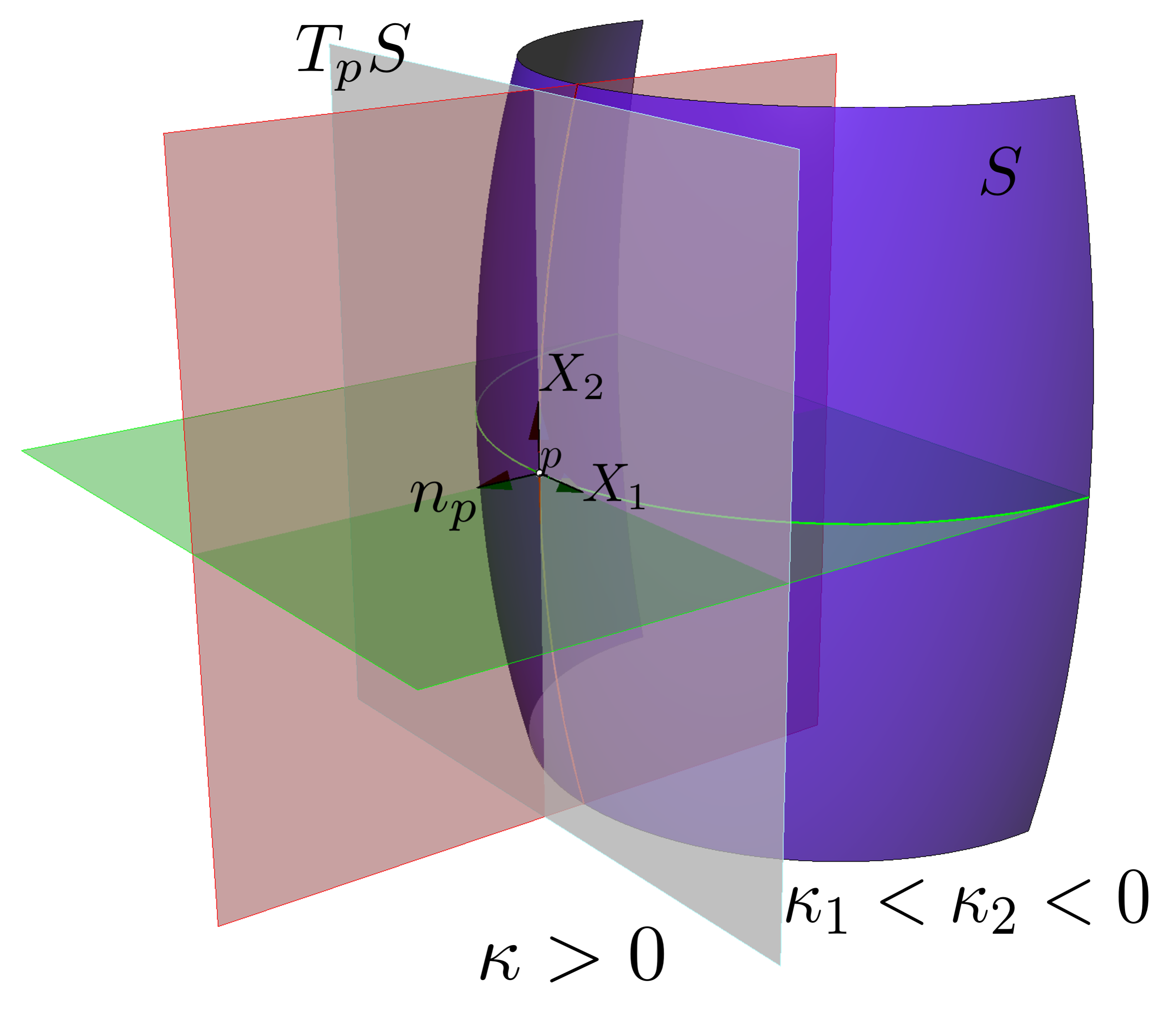}
\caption{Positive Gaussian Curvature}
\end{subfigure}
\caption{Depictions of principal and Gaussian curvatures at points $p$ are given. 
Here, principal curvatures established by intersections of normal planes in red and green with the surface, and are shown as curves in red and green. Gaussian curvature is the product of these principal curvatures. Note that for the negative curvature case, osculating circles have opposite directions, while in the positive curvature case, both curvature directions are the same.}\label{fig:gaussian_curvature}
\end{figure}

Finally, while this entire discussion hinged on an embedding in $\mathbb{R}^3$, Gauss's Theorema Egregium states that Gaussian curvature is invariant under local \textbf{isometries} (metric-preserving diffeomorphisms). 
Thus all that is required is a local isometry of a surface with metric in $\mathbb{R}^n$ (Gaussian curvature may be defined in larger codimensions \cite[pp.~191--194]{Spivak:1999v2}) to discuss it's Gaussian curvature. 
Such global isometries are guaranteed by the Nash Embedding Theorem, 
and local isometries are guaranteed by the Burstin-Janet-Cartan Theorem. 
Gaussian curvature will play a fundamental role in the parameterization and quad-layout extraction of a surface.

A neighborhood in a surface is called \textbf{flat} if it isometrically embeds in $\mathbb{R}^2$.
Each point $p$ in a flat neighborhood has zero Gaussian curvature \cite[pp.~178--179,190--191]{Spivak:1999v2}. 
If the Gaussian curvature on a surface induced by a metric is flat everywhere, it is called a \textbf{flat metric}. 

\begin{figure}
\centering
\includegraphics[trim = 0cm  0cm 0cm 0cm, clip, width=.65\textwidth]{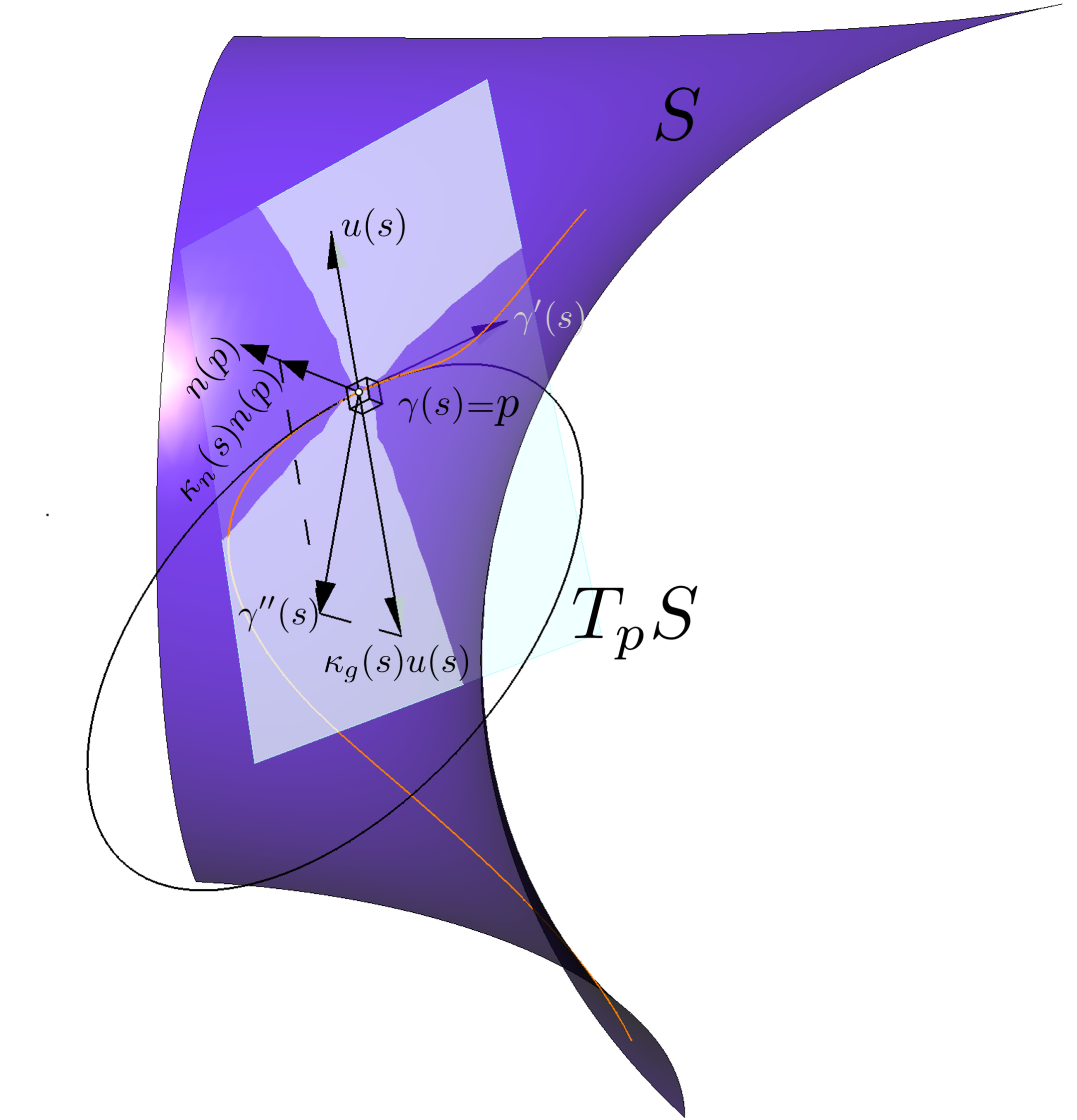}
\caption{Geodesic $(\kappa_g)$ and normal $(\kappa_n)$ curvatures of a curve, $\gamma$, smoothly embedded on a surface and parameterized by arclength are depicted. 
Here, $\gamma^\prime$ is tangent to the curve, while  $\gamma^{\prime \prime}$ points to the center of the osculating circle. 
The unit normal at point $\gamma(s) = p\in \surf, s \in \mathbb{I}$ is given by $n(p)$. 
The unit binormal at $p$ is given as $u(s) = n(p) \times \gamma^\prime(s).$
The normal curvature of the curve at $s$ is given by $\kappa_n(s)$, while the geodesic curvature of $\gamma$ at $s$ is given by $\kappa_g(s)$.}\label{fig:normal_geodesic_curvature}
\end{figure}

A curve $\gamma$ immersed in a surface immersed in $\mathbb{R}^3$ has inherent curvature which can be expressed as a combination of its normal and geodesic curvatures.
Assume $\gamma:\mathbb{I}\rightarrow \surf$ is parameterized by arclength, with $s \in \mathbb{I}$.
Under the induced Euclidean metric $\emptyflatmetric$, $n$ is the unit normal map and $\gamma^\prime(s)$ is the vector tangent to $\gamma$ at $s$.
 Then  $\gamma^{\prime \prime}(s)$ is a vector orthogonal to both $\gamma^{\prime}(s)$ and $n(s)$ in the direction of the center of the osculating circle at $\gamma(s)$ and $\frac{1}{|\gamma^{\prime \prime}(s)|}$ is the radius of the osculating circle.
Let $\gamma(s) = p \in \surf$.
Take $u(s)$ as the unit length binormal given by $n(p) \times \gamma^\prime(s)$.
Then the \textbf{normal curvature} of $\gamma$ at $s$ is given by 
\[
	\kappa_n(s) = \flatmetric{\gamma^{\prime \prime}(s)}{n(p)}.
\]
The \textbf{geodesic curvature} of $\gamma$ at $s$ is given by
 \[
	\kappa_g(s) = \flatmetric{\gamma^{\prime \prime} (s)}{ u(s)}.
\]
These are both pictorially represented in Figure \ref{fig:normal_geodesic_curvature}.
(As with Gaussian curvature, geodesic and normal curvature can be defined for an arbitrary Riemannian metric without embedding in $\mathbb{R}^3$.)

Despite appearing to be concepts completely related to a particular embedding in Euclidean space, Gaussian and geodesic curvatures are intrinsically related to the topology of the surface via the \textbf{Gauss-Bonnet Theorem}, which states that for surface $S$ with (possibly empty) boundary $\partial S$,
\begin{equation}\label{eq:Gauss_Bonnet}
\int_S \kappa dS + \int_{\partial S} \kappa_g d(\partial S) = 2\pi\chi(S),
\end{equation}
where $\chi(S)$ is the Euler Characteristic of the surface. Thus a valid metric must necessarily satisfy the Gauss-Bonnet Theorem.

A \textbf{geodesic} between points $p$ and $q$ is a critical point of the energy
\begin{equation}\label{eq:geodesic}
E(\gamma) = L(\gamma)
\end{equation}
where $\gamma:\mathbb{I}\rightarrow M$ has $\gamma(0)=p,\gamma(1) = q$.
One such geodesic is a shortest path between $p$ and $q$ in the specified metric. 
The \textbf{parallel translation} of a vector  $V_0 \in T_pM$ to $V_1 \in T_qM$ is the representation of of $V_1$ in $T_qM$ in which lengths and angles (inherent from the metric tensor) are preserved as measured from a geodesic between points $p$ and $q$ (see Figure \ref{fig:parallel_translation}).
The \textbf{Levi-Cevita connection} is the set of bijective linear maps, $\tau_t:T_{\gamma(0)}M\rightarrow T_{\gamma(1)}M$ induced by parallel translation such that $\tau_t(V_0) = V_t$.  
For a curve $\gamma$ with $\gamma(0) = p \in M, \gamma^{\prime}(0) = X_p$, the \textbf{covariant derivative}  of the vector field $Y$ along $\gamma$ is
\[
	\nabla_{X_p}Y = \lim_{h\rightarrow 0} \frac{1}{h}\big(\tau_h^{-1} Y_{\gamma(h)} - Y_p\big).
\]
Computationally, it is represented using Christoffel symbols.

\begin{figure}
\centering
\includegraphics[trim = 0cm  0cm 0cm 0cm, clip, width=.5\textwidth]{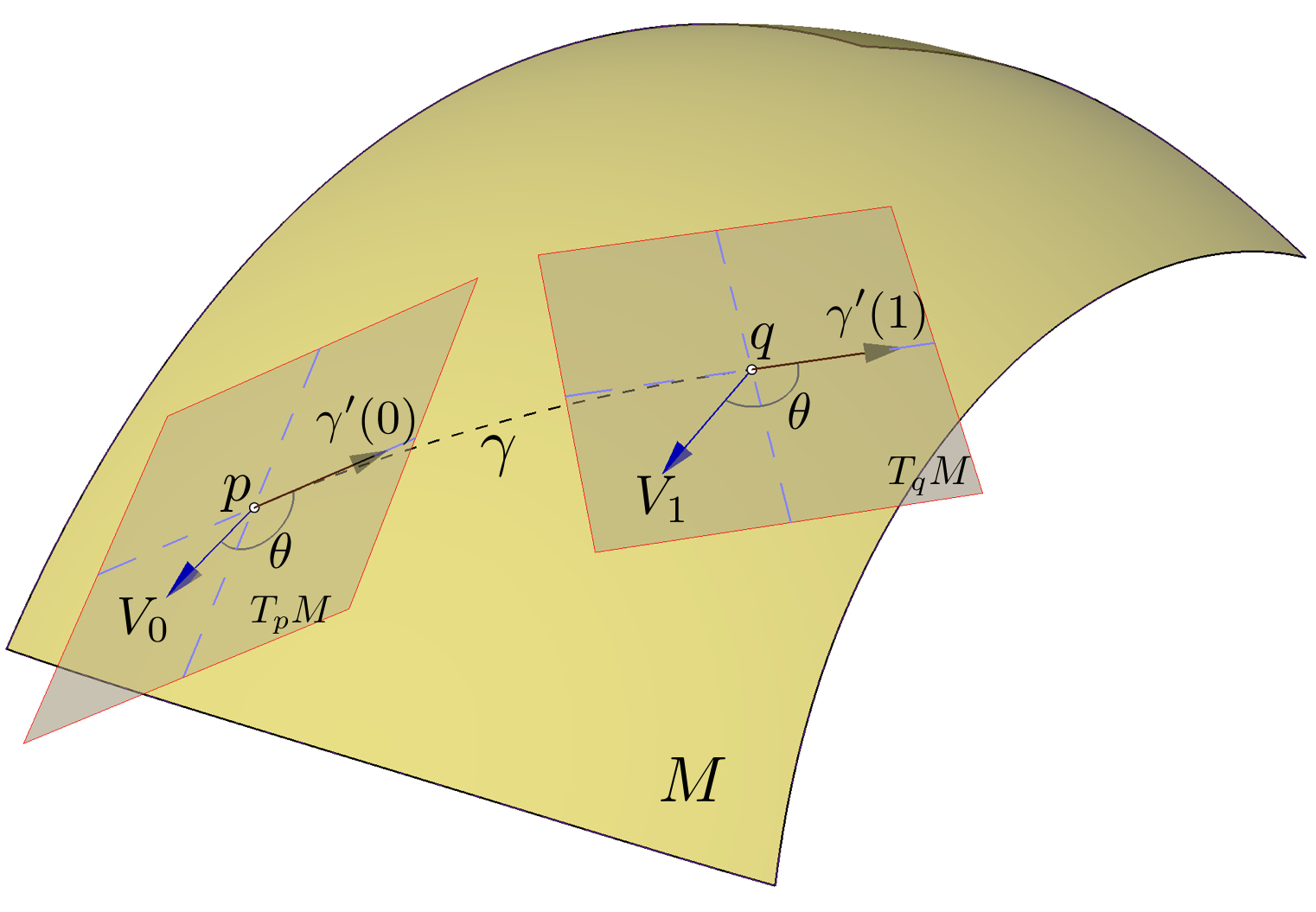}
\caption{Given a vector in the tangent plane of a surface and a geodesic, $\gamma$, parallel translation of this vector into a different tangent plane  under the Levi-Cevita connection maintains vector length and angle (with respect to the tangent vector of the geodesic curve) in the new tangent plane under the same Riemannian metric.}\label{fig:parallel_translation}
\end{figure}

One final tool is necessary for the purposes of this paper. A piecewise smooth loop $\omega:\mathbb{I}\rightarrow M$ may be defined as the concatenation $\omega = \gamma_1 \cdot (\dots) \cdot \gamma_n,$ with the precise parameterization (i.e. placement of parentheses above) being extraneous for the following purposes. Then one may define $P_\omega:T_pM\rightarrow T_pM$ as the (invertible, linear) map defined by parallel translation in a loop through concatenation of geodesics. The \textbf{holonomy group} at $p \in M$ defined by the Levi-Cevita connection is
\begin{equation}
\text{Hol}_p = \{P_\omega: \omega \text{ is a loop based at } p\}
\end{equation}
The holonomy group at a point is related to the curvature encompassed by the closed loops $\gamma$ via the \textbf{Ambrose-Singer Theorem} \cite{ambrose:1953}.
In Figure \ref{fig:sphere_holonomy}, one can see that parallel translation along the three depicted geodesics yields a rotation by $\frac{\pi}{2}$ in the tangent space; this rotation was induced by the curvature of the domain encompassed by the path.
The \textbf{reduced holonomy group} of the Levi-Cevita connection is
\begin{equation}
\text{Hol}_p^0 = \{P_\omega: \omega \text{ is a contractible loop based at } p\}.
\end{equation}
A metric is flat if and only if $\text{Hol}_p^0$ is trivial, meaning that $P_\omega$ is the identity map for any $p \in M$ \cite[pp.~283]{Besse:1987}.

\begin{figure}
\centering
\includegraphics[trim = 0cm  0cm 0cm 0cm, clip, width=.5\textwidth]{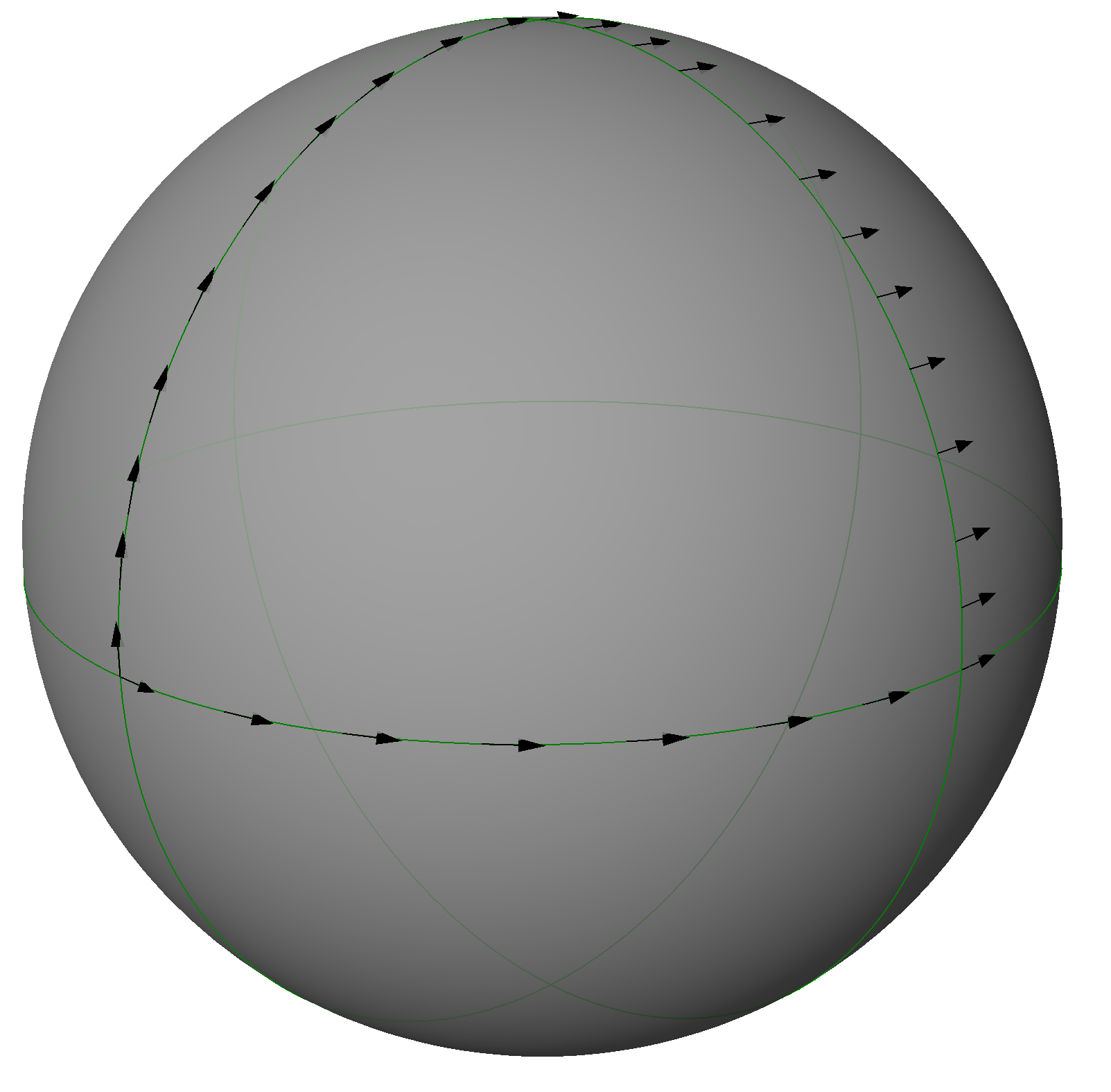}
\caption{A sphere is embedded in $\mathbb{R}^3$ takes the induced Euclidean metric.
Under the induced Levi-Cevita connection, a vector is parallel translated in a closed loop on the sphere.
Upon return to the base point, this vector has rotated by $\frac{\pi}{2}$ radians.
As such, $\frac{\pi}{2}$ is a member of the holonomy group of the surface at this base point. }\label{fig:sphere_holonomy}
\end{figure}

%% file: Appendix.tex
\ifthenelse{\boolean{isELS}}
{
\subsection{Appendix}\label{sec:appendix}
}
{
\section{Appendix}\label{sec:appendix}
}
\ifthenelse{\boolean{isELS}}
{
\subsubsection{Proof of Lemma \ref{lem:existence}}
}
{
\subsection{Proof of Lemma \ref{lem:existence}}
}

We make use of the following lemma.
\begin{lemma}\label{lem:generators}
For a surface $S$ of genus $g$ with $\boundarycomp$ boundary components, there is a set of curves cutting $S$ into a simply-connected surface in which the intersection of the curves with $\partial S$ is discrete.
\end{lemma}
\begin{proof}
If $\surf$ is a topological $2$-sphere with no boundaries, we are done. Then assume that either $\genus$ or $\boundarycomp > 0.$
 
Let $\surf$ be of genus $\genus$ with $\boundarycomp$ boundary components.
If $k \geq 1$ pick $p \in \partial \surf;$ otherwise, pick $p$ arbitrarily.
The presentation of the fundamental group $\pi_1(S, p)$ at basepoint $p$ is given by 
\[
	 \pi_1(\surf, p) = 
	\begin{cases}
		\langle a_1,\dots,a_{\genus},b_1,\dots,b_{\genus} | [a_1,b_1]\dots[a_\genus, b_\genus] = 1 \rangle & \text{ if } \boundarycomp = 0\\
		\langle a_1,\dots,a_{\genus},b_1,\dots,b_{\genus}\rangle & \text{ if } \boundarycomp = 1\\
		\langle a_1,\dots,a_{\genus},b_1,\dots,b_{\genus},c_1,\dots,c_{\boundarycomp-1}\rangle & \text{ if } \boundarycomp > 1
	\end{cases}
\]
where $1$ is the identity element of the group and $[a_i,b_i] = (a_1b_1a_1^{-1}b_1^{-1})$ is the commutator of $a$ and $b$. 
Let $\surf/\partial \surf$ be the quotient space topology with quotient map $q:\surf\rightarrow \surf/\partial \surf$.
Write $\surf_g$ as an arbitrary a closed surface of the same genus as $\surf$.
Then 
\[
	\surf/\partial \surf \simeq 
		\begin{cases}
			\surf & \text{ if } k = 0,\\
			\surf_g \vee \Big(\bigvee_{i=2}^k \mathbb{S}^1\Big) & \text{ if } k \geq 1.
		\end{cases}
\]
where $\vee$ represents the wedge sum (see \cite{Hatcher} p. 10) and $\mathbb{S}^1$ is the one-sphere.
Thus, because the fundamental group respects homotopy equivalences, 
\[
	 \pi_1\big(\surf/\partial \surf, q(p)\big) = 
	\begin{cases}
		\langle a_1,\dots,a_{\genus},b_1,\dots,b_{\genus} | [a_1,b_1]\dots[a_\genus,b_\genus] = 1 \rangle & \text{ if } \boundarycomp \leq 1\\
		\langle a_1,\dots,a_{\genus},b_1,\dots,b_{\genus}, d_1,\dots,d_{\boundarycomp-1} | [a_1,b_1]\dots[a_\genus, b_\genus] = 1\rangle & \text{ if } \boundarycomp > 1.
	\end{cases}
\]
Because the quotient map is the identity away from $\partial \surf$, the generators $a_i,b_i$ can be represented as curves with the same image in both $\surf$ and $\surf / \partial \surf$ for $i = 1,\dots, \genus$.
 For $d_j$, the loops of $\surf/\partial \surf$ have preimages which are curves in $S$ connecting the boundary components. 
 Denote curve corresponding to the preimage of $d_j$ under the quotient map as $q^{-1}(d_j)$. Then by construction the set
\[
	\hat{\Gamma} :=
		\left\{a_i\right\}_{i=1}^\genus \cup
		\left\{b_i\right\}_{i=1}^\genus \cup
		\left\{q^{-1}(d_j)\right\}_{j=1}^{k-1}.\footnote{Here, $\Gamma$ represents basic nonzero elements of the groupoid $\pi_1(\surf,\partial \surf)$, but an introduction of relative homotopy groups is outside the scope of this paper.}
\]
 cuts $\surf$ into connected components $\surf_\ell$ such that for each $\ell, \pi_1(\surf_\ell) = 0$. 
Using the Classification of Surfaces (Theorem \ref{th:classification}) to equivalently represent $\surf$ as the (typical) quotient space of a $4\genus$-sided polygon with $\boundarycomp$ holes (see e.g. \cite{Hatcher} p. 5), it is easy to see that $\hat{\Gamma}$ can be chosen to enforce $\surf - \hat{\Gamma}$ is one connected component and discrete boundary intersection.
 Here, if $\genus = 0$, take this polygon as a disk with $\boundarycomp \geq 1$ boundary components and without any quotient space topology. 
\end{proof}

As before, define $\mathring{\singpts} := \singpts - \partial \surf$.
Because $\surf - \mathring{p}$ is homotopy equivalent to a surface of genus $\genus$ and with $(\boundarycomp + \#\mathring{p})$ boundary components, the proof of Lemma \ref{lem:existence} follows a similar approach as Lemma \ref{lem:generators}.

\begin{proof}[Proof of Lemma \ref{lem:existence}]
First, note that if $\singpts = \emptyset, \genus = 0,$ the surface is a topological sphere, which is simply connected and the results hold.

Let $\mathring{\singpts} = \singpts - \partial \surf,$ and define $\hat{\surf} = \surf - \mathring{\singpts}$. 
Because $\surf$ is Hausdorff, at each $p,q \in \mathring{\singpts},$ define an open disk-like neighborhood 
	$U_p, U_q$ such that $U_p \cap U_q = \emptyset, p \in U_p, q \in U_q.$ 
Let 
	$h:\hat{\surf}\rightarrow h(\hat{\surf})$ 
	be a retraction taking each $U_p - \{p\}$ to $\partial U_p$ via a radial projection. 
Here, 
	$h(\hat{\surf})$ 
	is a surface of genus $g$ with 
	$(\boundarycomp + \#\mathring{\singpts})$ boundary components. 
Then by Lemma \ref{lem:generators}, there is a set of curves with the prescribed conditions on $h(\hat{\surf})$.

Using the inclusion, this set of curves maps into $\surf$ to yield a graph $\graph$ which satisfies the desired requirements except $\mathring{\singpts} \not \subset \graph$. 
To extend $\graph$ appropriately, let $\gamma_p$ be the cutting curve meeting $U_p$. The intersection is discrete. 
Let $\tilde{\gamma}_p=h^{-1}\left(\partial U_p \cap \gamma_p \right)$ be the set homeomorphic to $[0,1)$ in $\hat{\surf}$ and take the closure of $\tilde{\gamma_p}$ in $\surf$ to get a domain homeomorphic to the unit interval taking $\left(\partial U_p \cap \gamma_p \right)$ to $p$. 
The composition of each $\gamma_p$ with each $\overline{\tilde{\gamma}}_p$ yields a set of curves that appropriately extend $\graph$.
\end{proof}

\ifthenelse{\boolean{isELS}}
{
\subsubsection{Direct Proofs for a \QuadMeshImmersion{} Inducing a \QuadLayout{}}
}
{
\subsection{Direct Proofs for a \QuadMeshImmersion{} Inducing a \QuadLayout{}}
}

The proofs of Lemma \ref{lem:integrable_fields} Theorem \ref{thrm:equivalence} were simplified using the equivalence of \cite{Chen:2019} that a \textquadmeshmetric{} induces a quadrilateral layout. 
Here, we present the proofs without this assumption.
It is hoped that these will more clearly present the machinery of the tools used, though perhaps at the expense of more details.

\subsubsection{Proof that the \textQuadMeshMetric{'s} \textCrossField{} are Locally Integrable}\label{sec:integrability_crossfield}
\begin{proof}
Because $\emptyquadmeshmetric$ is flat, by definition it has zero  Gaussian curvature.
Then $\surfmsinggraph$ is locally isometric to $\mathbb{R}^2$ 
(see \cite{Spivak:1999v2} p. 241). 
Then for any $p \in \surf - \singpts$ there is some neighborhood $U_p$ with a function $\phi_p:U_p\rightarrow \phi_p(U_p) \subset \mathbb{R}^2$ such that $\emptyquadmeshmetric = \phi_p^*(\emptymetric_{\mathbb{R}^2})$.

Fix one such $p$.
Because $U_p$ is contractible, it has trivial holonomy group, so the components of the \textcrossfield{} decompose into four unique vector fields, $\{X_i\}_{i=0}^4$.
By definition of a \textcrossfield{}, these vector fields are symmetric with a rotation of $X_{i_p}$ by $\frac{\pi}{2}$ in $T_p\big(\surfmsinggraph\big)$ yielding $X_{[(i+1) \mathrm{mod}\;4]_p}$.

Let $\{e_i\}_{i=0}^1$ be the Cartesian basis in $\mathbb{R}^2$ with coordinates $u_i$.
Then 
\[
	\nabla^{\mathbb{R}^2}_{e_i} \metric{(\phi_*)(X_m)}{(\phi_*) (X_m)}_{\mathbb{R}^2} 
	= \frac{\partial}{\partial u_i}\Big(\metric{(\phi_*)(X_m)}{(\phi_*) (X_m)}_{\mathbb{R}^2}\Big) + 2\metric{\nabla^{\mathbb{R}^2}_{e_i}\big((\phi_*)(X_m)\big)}{(\phi_*)(X_m)}_{\mathbb{R}^2} 
	= 0.
\]
by definition of the covariant derivative under the Levi-Cevita connection on $\mathbb{R}^2$. But the first is zero because 
\[
	\metric{(\phi_*)(X_m)}{(\phi_*) (X_m)}_{\mathbb{R}^2} 
	= \quadmeshmetric{X_m}{X_m} 
	= 1.
\]
Furthermore, $\nabla_{e_i}^{\mathbb{R}^2}\big((\phi_*)(X_m)\big) $ is just the partial derivative of each component because the Christoffel symbols in Cartesian coordinates are all zero. 
Then taking $(\phi_*) (X_m) = \sum_{j=0}^1 a_j e_j$, we have
\[
	0 = 2 \Big(\sum_{j=0}^1 a_j \frac{\partial a_j}{\partial u_j}\Big) = \sum_{j=0}^1 \frac{\partial (a_j^2)}{\partial u_i}.
\]
Thus each $a_j^2$ must be constant, and so must each $a_j$.
Hence for each $e_i$
\[
	\nabla^{\mathbb{R}^2}_{e_i} \big((\phi_*) (X_m)\big) = 0.
\]
Since  $\{e_i\}_{i=0}^1$ form a basis for each tangent space and the covariant derivative under the Levi-Cevita connection is preserved by isometries,
\[
	\phi_*\big(\nabla^Q_{X_i} X_j \big)= \nabla^{\mathbb{R}^2}_{(\phi_*)(X_i)} \big((\phi_*) (X_j)\big) = 0,
\]
implying that 
\[
	\nabla^Q_{X_i} X_j = 0.
\]
Then because of symmetry of \textquadmeshmetric{'s} Levi-Cevita connection 
\[
[X_i,X_j] = \nabla^Q_{X_i} X_j - \nabla^{Q}_{X_j} X_i =  0.
\]
Then for $X_i,X_j$ linearly independent, there is a local coordinate system about $p$ in $U_p$ defined by integration on $X_i, X_j$ 
(see \cite{Spivak:1999v1} p. 158).
\end{proof}

\subsubsection{Alternative Proof of Theorem \ref{thrm:equivalence}}

\begin{proof}
	Proceed as in the earlier proof of the claim to show that the skeletal structure of the integral curves on $\surf$ partitions $\surf$ into a set of quadrilateral domains.
	However, rather than directly appealing to the equivalence between a quadrilateral layout and a \textquadmeshmetric{}, here we directly show that any integral curve of $C$ is finite.

	Let $\mathfrak{Q}_{p}$ be a quotient curve on the surface not traversing a boundary and not containing a singularity.
	Then $p$ is in the interior of some face.
	We are solely interested in the lengths and angular deviations of $\mathfrak{Q}_{p}$, which are preserved by isometric immersions.
	Thus, construct an isometric immersion $\tilde{\immersion}$ of the quadrilateral cell containing $p$ as was done in Proposition \ref{prop:quad_to_immersion}.
	Here, there are no cuts and no singularities except at the corners of the face, so the isometric immersion will be a rectangle in order to satisfy local coordinates being integral curves of $X_0$ and $X_1$ of the \textcrossfield{} $C$.
	Then the length of $\mathfrak{Q}_{p}$ in this rectangle is finite, being a line of constant $u$ or $v$ coordinate.
	Similarly, the length of $\mathfrak{Q}_p$ in any other rectangle in which it is contained is finite.
	Furthermore, the path of faces traversing this quotient curve is necessarily finite.
	If it is not a loop, the quotient curve is also necessarily finite.
	If the face path is a loop, note that $\mathfrak{Q}_p$ is some line of constant $u$ or $v$ in this rectangle.
	But the lengths of each edge traversed in the closed face-path must be the same, and $\mathfrak{Q}_p$ must intersect at the same height along each rectangle.
	Then after a loop in the face paths, it must return to be periodic, and thus of finite length.
\end{proof}